\documentclass[12pt]{article}
\usepackage[utf8]{inputenc}
\usepackage[english]{babel}

\usepackage{url}
\usepackage{amsmath,amssymb,amsfonts,amsthm,amsbsy, epsfig, psfrag,epsf}
\usepackage[usenames]{color}
\usepackage{rotating}
\usepackage[colorlinks=true, linkcolor=blue, urlcolor=blue, citecolor=blue]{hyperref}
\usepackage{boxedminipage}
\usepackage{algorithm}
\usepackage{algorithmic}
\usepackage{natbib}
\usepackage{graphicx,psfrag,epsf}
\usepackage{enumerate}
\usepackage{authblk}
\usepackage{xcolor}
\usepackage{fullpage}
\usepackage{epstopdf}
\definecolor{DSgray}{cmyk}{0,1,0,0}

\usepackage{rotating}
\usepackage{subfigure}
\newcommand{\ignore}[1]{}{}

\renewcommand{\baselinestretch}{1.8}
\usepackage[title]{appendix}

\newtheorem{theorem}{Theorem}[section]
\newtheorem{corollary}[theorem]{Corollary}
\newtheorem{proposition}[theorem]{Proposition}
\newtheorem{lemma}[theorem]{Lemma}

\newtheorem{remark}[theorem]{Remark}
\newtheorem{definition}[theorem]{Definition}
\newtheorem{assumption}[theorem]{Assumption}

\newcommand{\argmin}{\mathop{\rm arg\min}}
\newcommand{\argmax}{\mathop{\rm arg\max}}

\def\0{\boldsymbol{0}}
\def\A{\boldsymbol{A}}
\def\a{\boldsymbol{a}}

\def\ep{\mathbb{E}}

\def\e{\boldsymbol{e}}

\def\R{\mathbb{R}}
\def\M{\boldsymbol{M}}

\def\g{\boldsymbol{g}}

\def\H{\boldsymbol{H}}
\def\I{\boldsymbol{I}}
\def\OO{\boldsymbol{O}}
\def\p{\boldsymbol{p}}
\def\P{\boldsymbol{P}}
\def\Q{\boldsymbol{Q}}

\def\U{\boldsymbol{U}}
\def\u{\boldsymbol{u}}
\def\V{\boldsymbol{V}}
\def\v{\boldsymbol{v}}
\def\W{\boldsymbol{W}}
\def\w{\boldsymbol{w}}

\def\x{\boldsymbol{x}}

\def\y{\boldsymbol{y}}

\def\tM{\boldsymbol{\widetilde{M}}}
\def\tQ{\boldsymbol{\widetilde{Q}}}
\def\tU{\boldsymbol{\widetilde{U}}}
\def\tV{\boldsymbol{\widetilde{V}}}

\def\tA{\boldsymbol{\widetilde{A}}}
\def\tw{\boldsymbol{\widetilde{w}}}

\def\hU{\boldsymbol{\widehat{U}}}

\def\hU{\boldsymbol{\widehat{U}}}

\def\hu{\boldsymbol{\widehat{u}}}

\def\be{\boldsymbol{\beta}}

\def\ga{\boldsymbol{\gamma}}

\def\S{\boldsymbol{\Sigma}}
\def\O{\boldsymbol{\Omega}}

\def\La{\boldsymbol{\Lambda}}

\def\va{\boldsymbol{\varepsilon}}
\def\eps{\boldsymbol{\epsilon}}

\def\hlambda{\widehat{\lambda}}
\def\hSigma{\boldsymbol{\widehat{\Sigma}}}
\def\tSigma{\boldsymbol{\widetilde{\Sigma}}}
\def\hbe{\boldsymbol{\widehat{\beta}}}
\def\tbe{\boldsymbol{\widetilde{\beta}}}

\def\hga{\boldsymbol{\widehat{\gamma}}}
\def\hm{\boldsymbol{\widehat{\mu}}}

\def\H{\mathbf{H}}

\let\hat\widehat

\newcommand{\matnorm}[2]{{\left|\!\left|\!\left| #1 \right| \! \right| \! \right|_{#2}}}

\begin{document}

\def\spacingset#1{\renewcommand{\baselinestretch}
{#1}} \spacingset{1.4}


\title{\bf Distributed Estimation for Principal Component Analysis: an Enlarged Eigenspace Analysis
}
\author{Xi Chen \footnote{Xi Chen is the corresponding author; e-mail: xc13@stern.nyu.edu} \\
Stern School of Business, New York University \\
and \\
Jason D. Lee \footnote{e-mail: jasonlee@princeton.edu}\\
Department of Electrical Engineering, Princeton University \\
and \\
He Li \footnote{e-mail: hli@stern.nyu.edu}  \\
Stern School of Business, New York University \\
and \\
Yun Yang \footnote{e-mail: yy84@illinois.edu} \\
Department of Statistics, University of Illinois Urbana-Champaign \\
}
\date{}
\maketitle

\begin{abstract}
The growing size of modern data sets brings many challenges to the existing statistical estimation approaches, which calls for new distributed methodologies. This paper studies distributed estimation for a fundamental statistical machine learning problem, principal component analysis (PCA). Despite the massive literature on top eigenvector estimation, much less is presented for the top-$L$-dim ($L>1$) eigenspace estimation, especially in a distributed manner. We propose a novel multi-round algorithm for constructing top-$L$-dim eigenspace for distributed data. Our algorithm takes advantage of shift-and-invert preconditioning and convex optimization. Our estimator is communication-efficient and achieves a fast convergence rate. In contrast to the existing divide-and-conquer algorithm, our approach has no restriction on the number of machines. Theoretically, the traditional Davis-Kahan theorem requires the explicit eigengap assumption to estimate the top-$L$-dim eigenspace. To abandon this eigengap assumption, we consider a new route in our analysis: instead of exactly identifying the top-$L$-dim eigenspace, we show that our estimator is able to \emph{cover} the targeted top-$L$-dim population eigenspace. Our distributed algorithm can be applied to a wide range of statistical problems based on PCA, such as principal component regression and single index model. Finally, We provide simulation studies to demonstrate the performance of the proposed distributed estimator.
\end{abstract}


\section{Introduction}
\label{sec:intro}

The development of technology has led to the explosive growth in the size of modern data sets. The challenge arises, when memory constraints and computation restrictions make the traditional statistical estimation and inference methods no longer applicable. For example, in a sensor network, the data are collected on each tensor in a distributed manner. The communication cost would be rather high if all the data are transferred and computed on a single (central) machine, and it may be even impossible for the central machine to store and process computation on such large-scale datasets. Distributed statistical approaches have drawn a lot of attentions these days and methods are developed for various statistics problems, such as sparse regression (see, e.g., \cite{lee2017communication}), likelihood-based inference (see, e.g., \cite{battey2015distributed,jordan2019communication}),  kernel ridge regression \citep{zhang2015divide},  semi-parametric partial linear models \citep{zhao2016partially}, quantile regression (see, e.g., \cite{volgushev2019distributed,chen2019quantile,Chen2019Distributed}), linear support vector machine \citep{Wang2019Distributed}), Newton-type estimator \citep{chen2021newton}, and $M$-estimators with cubic rate \citep{shi2018massive,banerjee2019divide}. All these works are seeking for distributed statistical methods that are able to handle massive computation tasks efficiently for large-scale data and achieve the same convergence rate as those classical methods as well.

In a typical distributed environment, each machine has access to a different subset of samples of the whole data set. The communication and computation follow from a hierarchical master-slave-type architecture, where a central machine acts as a fusion node. Computation tasks for local machines and the central machine are different. After local machines finish their computation, the local results will be transferred to the master machine, where they will be merged together and the fusioned result will be transferred back to all local machines for the next step.

In this paper, we study the problem of principal component analysis (PCA) in a distributed environment. PCA \citep{pearson1901liii,hotelling1933analysis} is one of the most important and fundamental tools in statistical machine learning. For random vectors $\a_1, \ldots, \a_n$ in $\mathbb{R}^d$ with mean zero and covariance matrix $\S$, its empirical covariance matrix is $\hSigma = \frac{1}{n} \sum_{i=1}^{n} \a_i \a_i^\top$. The $L$-PCA ($L \leq d$) finds a $L$-dimension subspace projection that preserves the most variation in the data set, which is equivalent to the following optimization problem:
\begin{align}
\label{eq:pca-prob}
\max_{\U \in \mathbb{R}^{d \times L}: \U^T \U = \I_L} \matnorm{\hSigma \U}{\mathrm{F}},
\end{align}
where $\matnorm{\cdot}{\mathrm{F}}$ denotes the matrix Frobenius norm and $\I_L$ is the $L \times L$ identity matrix.
In other words, $\U \in \mathbb{R}^{d \times L}$ is the top-$L$-dim eigenspace of $\hSigma$. PCA has been widely used in many aspects of statistical machine learning, e.g., principal component regression \citep{jeffers1967two,jolliffe1982note}, single index model \citep{li1992principal}, representation learning \citep{Bengio2013Representation}.

Under distributed regime, \cite{Fan2017} proposed a novel one-shot type of algorithm which is often called divide-and-conquer (DC) method. In \cite{Fan2017}, DC method first computes local covariance matrices $\hSigma_i$ on each machine $k = 1, \ldots, K$. Eigenspaces $\hU_k, k = 1, \ldots, K$ are then computed locally using the traditional PCA algorithm and transmitted to the central machine. Central machine combines local eigenspaces $\hU_k$ into an aggregated covariance estimator, $\tSigma = \frac{1}{K} \sum_{k=1}^K \hU_k \hU_k^\top$. The final estimator is obtained as the top-$L$-dim eigenspace of $\tSigma$. DC method is easy to implement and requires only $\mathcal{O}(d L)$ communications for each local machine, where  $d$  denotes the data dimension, $n$ the total sample size, and $m$ the sample size on each local machine. Let us denote the condition number of the population covariance matrix $\S$ by $\rho$, i.e., $\rho = \lambda_1 / (\lambda_{L} - \lambda_{L+1})$, and the effective rank of $\S$ by $r = \mathrm{Tr}(\S) / \lambda_1$. For asymmetric innovation distributions, \cite{Fan2017} showed that when the number of machines is not very large (no greater than $\mathcal{O}(m / (\rho^2 r))$), DC method enjoys a optimal statistical convergence rate of order $\mathcal{O}(\rho \sqrt{L r / n})$. However, when the number of machines becomes larger, DC method only achieves a slow convergence rate of $\mathcal{O}(\rho \sqrt{L r / n}+\rho^2 \sqrt{L} r / m)$. This feature may not be desirable in distributed settings. For example, in a sensor network with a vast number of sensors, the number of machines may exceed the constraint set for the optimal rate. The precise definition of asymmetric innovation above is given in Section 4.2 of \cite{Fan2017}. Roughly speaking, a random variable $\a \in \R^d$ is distributed under asymmetric innovation if flipping the sign of one component of $\a$ changes its distribution.

One question naturally arises from the analysis of DC method, can we possibly relax the restriction on the number of machines? Motivated by this question, our paper presents a multi-round distributed algorithm for top-$L$-dim eigenspace estimation. 

The contribution of our method is two-fold. First, as compared to DC method in \cite{Fan2017}, we completely remove the assumption on the number of machines. Our method leverages shift-and-invert preconditioning (a.k.a., Rayleigh quotient iteration) from numerical analysis \citep{van2012matrix} together with quadratic programming and achieves a fast convergence rate. Moreover, most previous convergence analysis of eigenspace estimation relies on the assumption of an explicit eigengap between the $L$-th and the $(L+1)$-th population eigenvalues $\lambda_L$ and $\lambda_{L+1}$, i.e., $\lambda_L - \lambda_{L+1}> 0$, or other specific eigen-structures of $\S$. The second contribution of our paper is that we propose an enlarged eigenspace estimator that does not require any eigengap assumption.

In particular, let $\U_L$ denote the top-$L$-dim eigenspace  of the population covariance matrix $\S$, and $\hU_L$ the top-$L$-dim eigenspace of the empirical covariance $\hSigma$. Estimation consistency of $\hU_L$ is guaranteed by the (variant of) Davis-Kahan Theorem~\citep{davis1970rotation,Yu2014}: there exists an orthogonal matrix $\Q \in\mathbb R^{L\times L}$, such that
\begin{align} 
\label{eq:davis-kahan}
	\matnorm{\U_L - \hU_L \Q}{2} \leq \frac{\sqrt{2}\, \matnorm{\hSigma - \S}{2}}{\min(|\hlambda_{L-1} - \lambda_L|,|\hlambda_{L+1} - \lambda_L|)},
\end{align}
where $\matnorm{\cdot}{2}$ denotes the matrix spectrum norm. 
Since the empirical eigenvalue $\hlambda_{l}$ is expected to be concentrated around its population counterpart $\lambda_l$ for all $l\in[d]$, the consistency of $\hU_L$ relies on an eigengap condition requiring $\min(\lambda_{L-1} - \lambda_L,\lambda_L-\lambda_{L+1})$ to be strictly away from zero. Unfortunately, without such an eigenvalue gap condition, the top-$L$-dim subspace $\U_L$ is not statistically identifiable and estimation error from $\hU_L$ can be arbitrarily large (cf.~a counter-example provided in~\cite{Yu2014}). Fortunately, in many statistical applications of PCA such as the principal component regression (see Example 1 below), it suffices to retrieve the variation captured by the top eigenspace rather than exactly recover the top eigenspace in order to achieve a small in-sample prediction risk. To address the challenge of no explicit eigengap, we choose a different perspective. In particular,  we consider an \emph{an enlarged estimator} $\V_{>(1-\delta) \hlambda_L}$ (see Equation~\eqref{eq:intro-enlarged}), where $\delta$ is a pre-specified constant to quantify the amount of enlargement. \begin{align}\label{eq:intro-enlarged}
	&\underbrace{\u_1, \ldots, \u_L}_{\U_L}, \u_{L+1}, \ldots, \u_S,\u_{S+1}, \ldots, \u_d \\
	&\underbrace{\v_1, \ldots, \v_L, \v_{L+1}, \ldots, \v_S}_{\V_{> (1-\delta) \hlambda_L}}, \; \underbrace{\v_{S+1}, \ldots, \v_d}_{\V_{\leq (1-\delta) \hlambda_L}} \nonumber
\end{align}
Roughly speaking, we prove that our distributed estimator $\V_{>(1-\delta) \hlambda_L}$ satisfies inequality~\eqref{eq:davis-kahan} with the following property: the angle between the target $\U_L$ and the complement of our estimator $\V_{\leq (1-\delta) \hlambda_L}$ is sufficiently small (please see Theorem~\ref{thm:top_L_eigenvectors2} for more details). Such a property shows that the enlarged estimator $\V_{>(1-\delta) \lambda_L}$ almost cover the $\U_L$ even without an eigengap condition.

Our method is motivated by the shift-and-invert preconditioning. The idea of solving PCA via shift-and-invert preconditioning has long history in numerical analysis \citep{van2012matrix}. It is an iterative method that sequentially solves linear system to obtain increasingly accurate eigenvector estimates. Its connection with convex optimization has been studied in the past decade. In a single-machine setting, \cite{Garber:16:Eigen,Zhu:16:LazySVD} formulate each round of shift-and-invert preconditioning as a quadratic optimization problem and it can be solved with first-order deterministic (accelerated) gradient method like Nesterov accelerated method. \cite{garber2015fast,shamir2016fast,xu2018gradient} also relate the same convex optimization problem with variance-reduction stochastic technique (SVRG, see, e.g.,, \cite{johnson2013accelerating}). Furthermore, in distributed settings, \cite{garber2017communication} perform a multi-round algorithm but they only consider the estimation task of the first eigenvector. This paper proposes a general distributed algorithm that estimates the top-$L$-dim eigenspace without a restriction on the eigengap.

The proposed algorithm can facilitate many fundamental applications based on PCA in distributed environment. In particular, we illustrate two important applications, namely principal component regression (see Appendix~\ref{subsec:setting-pcr}) and single index model (see Appendix~\ref{subsec:setting-sim}).

\subsubsection*{Example 1: principal component regression}

Introduced by \cite{jeffers1967two,jolliffe1982note}, principal component regression (PCR) is a regression analysis technique based on PCA. Typically, PCR assumes a linear model $\y = \A \be^* + \eps$ with the further assumption that coefficient $\be^*$ lies in the low-rank eigenspace of data covariance matrix. Therefore, PCA can be performed to obtain the principal components $\hU_L$ of the observed covariance matrix $\hSigma = \frac{1}{n} \A^\top \A$ and the data matrix $\A$ is then projected on $\hU_L$. The estimator $\hbe$ of $\be^*$ is then obtained by regress $\y$ on this projected data matrix $\A \hU_L$. Many previous work has analyze the statistical property of PCR, see \cite{frank1993statistical,bair2006prediction}. Under a distributed environment, our distributed PCA algorithm can replace the traditional PCA algorithm in the above procedure and lead to a distributed algorithm for PCR. As we will show in Appendix~\ref{subsec:setting-pcr}, this distributed estimator achieves a similar error as in the single-machine setting.

\subsubsection*{Example 2: single index model}

Single index model \citep{li1992principal} considers a semi-parametric regression model $y = f(\langle \be^*, \a \rangle) + \epsilon$. Under some mild condition on the link function $f(\cdot)$, we would like to make estimation on the coefficient $\be^*$ using observed data $\{\a_i, y_i\}_{i=1}^n$ without knowing $f(\cdot)$. Some previous methods include semi-parametric maximum likelihood estimator \citep{horowitz2009semiparametric} and gradient-based estimator \citep{hristache2001direct}. Moreover, many works propose to use Stein's identity \citep{stein1981estimation,Janzamin2014} to estimate $\be^*$ (see, e.g., \cite{li1992principal,Yang2017} and references therein). Specifically, under Gaussian innovation where $\a$ is standard multi-variate normal random vector, the estimator $\hbe$ can be calculated from the top eigenvector of $\frac{1}{n} \sum_{i=1}^n y_i \cdot (\a_i \a_i^\top - \I_d)$. This method can be naturally extended to a distributed manner with a distributed eigen-decomposition of $\frac{1}{n} \sum_{i=1}^n y_i \cdot (\a_i \a_i^\top - \I_d)$.

\subsection{Notations}
\label{subsec:intro-note}

We first introduce the notations related to our work. We write vectors in $\mathbb{R}^d$ in boldface lower-case letters (e.g., $\a$), matrices in boldface upper-case letters (e.g., $\A$), and scalars are written in lightface letters (e.g., $t$). Let $\|\cdot\|$ denote vector norm (e.g., $\| \cdot \|_2$ is standard Euclidean norm for vectors). Matrix norm is written as $\matnorm{\cdot}{}$. For a matrix $\A \in \R^{n \times d}$, $\matnorm{\A}{2}$ and $\matnorm{\A}{\mathrm{F}}$ represent the spectral norm and Frobenius norm respectively. Furthermore, $\0$ represents zero vector with corresponding dimension and identity matrix with dimension $d \times d$ is shortened as $\I_d$. We use $\e_1, \ldots, \e_d$ to denote the standard unit vectors in $\R^d$, i.e., $\e_i = [0, \ldots, 0, 1, 0, \ldots, 0]$ where only the $i$-th element of $\e_i$ is $1$.

We use $\mathcal{O}_p$ to describe a high probability bound with constant term omitted. We also use $\widetilde{\mathcal{O}}_p$ to further omit the logarithm factors.

We adopt the standard definition of sub-Gaussian random vectors (see, e.g., \cite{Vershynin2010,rigollet2015high}) that a random vector $\a \in \mathbb{R}^d$ is said to be a $d$-dimensional sub-Gaussian with variance proxy $\sigma$ if $\mathbb{E}[\a] = \0$ and for any unit vector $\u$,
\begin{align*}
\mathbb{E}[\exp (s \a^\top \u)] \leq \exp \left(\frac{\sigma^2 s^2}{2}\right), \, \forall s \in \mathbb{R}.
\end{align*}

\subsection{Paper organization}
\label{subsec:intro-org}

The remainder of this paper is organized as follows. In Section~\ref{sec:setting}, we introduce the problem setups of the distributed PCA and give our algorithms. Section~\ref{sec:theory} develops the convergence analysis of our estimator. Finally, extensive numerical experiments are provided in Section~\ref{sec:exp}. The technical proofs and some additional experimental results are provided in the supplementary material. We also conduct analysis on two application scenarios, i.e., principal component regression and single index model in Appendix~\ref{sec:application} where we provide convergence analysis for both single-machine and distributed settings.

\section{Problem Setups}
\label{sec:setting}

In the following section, we collect the setups for our distributed PCA and present the algorithms.

Assume that there are $n$ \emph{i.i.d.} zero mean vectors $\a_i$  sampling from some distribution $\mathcal{D}$ in $\R^d$. Let $\A=[\a_1,\ldots, \a_n]^\top \in \R^{n \times d}$ be the data matrix. Let $\S$ be the population covariance matrix $\S=\ep_{\a \sim \mathcal{D}}[\a\a^\top]$ with the eigenvalues $\lambda_1(\S)\geq \lambda_2 (\S) \geq \ldots \lambda_d (\S) \geq 0$ and the associated eigenvectors are $\U = [\u_1,\ldots, \u_d] \in \mathbb{R}^{d \times d}$.

In the distributed principal component analysis, for a given number $L$, $1 \leq L \leq d$, we are interested in estimating the eigenspace spanned by $\U_L:=\{\u_1,\ldots, \u_L\}$ in a distributed environment. We assume $n$ samples are split uniformly at random on $K$ machines, where each machine contains $m$ samples, i.e., $n = mK$. We note that since our algorithm  aggregates gradient information across machines, it can handle the unbalanced data case without any modification. We choose to present the balanced data case only for the ease of presentation (see Remark \ref{rem:balance} for more details). The data matrix on each machine $k$ is denoted by $\A_k \in \R^{m \times d}$ for $k \in [K]$.

Let us first discuss a special case (illustrated in Algorithm~\ref{algo:distr_top}), where we estimate the top eigenvector, i.e., $L = 1$. The basic idea of our Algorithm~\ref{algo:distr_top} is as follows.

Let $\w^{(0)}$ be the initial estimator of the top eigenvector and $\overline{\lambda}_1$ a crude estimator of an upper bound of the top eigenvalue. Here we propose to compute $\w^{(0)}$ and $\overline{\lambda}_1$ only using the data from the first machine, and thus there does not incur any communication cost. For example, $\overline{\lambda}_1$ can be computed with $\overline{\lambda}_1 = \lambda_1(\A_1^\top\A_1/m) + 3 \eta / 2$, where $\lambda_1(\A_1^\top\A_1/m)$ is the top eigenvalue for the empirical covariance matrix on the first machine and $\eta$ is a special constant defined later in Equation~\eqref{eq:eta}. The $\w^{(0)}$ can be simply computed via eigenvalue decomposition of $\A_1^\top\A_1/m$.
We note the that Algorithm~\ref{algo:distr_top} is almost tuning free. The only parameter in constructing  $\overline{\lambda}_1$ is $\eta$. According to our theory, we could set $\eta=c_0 \sqrt{d/m}$ for some sufficiently large $c_0$ and the result is not sensitive to $c_0$. There are other tuning-free ways to obtain a crude top-eigenvalue estimator $\overline{\lambda}_1$ only using the sample on the first machine (e.g., the adaptive Algorithm 1 in \cite{Garber:16:Eigen} without tuning parameters).

\begin{algorithm}[!t]
    \caption{{\small Distributed Top Eigenvector (Distri-Eigen)}}
    \label{algo:distr_top}
    \hspace*{\algorithmicindent} \hspace{-0.72cm} {\textbf{Input:} Data matrix $\A_k$ on each machine $k=1,\ldots, K$. The initial top eigenvalue estimator $\overline{\lambda}_1$ and eigenvector estimator $\w^{(0)}$. The number of outer iterations $T$ and the number of inner iterations $T'$.}

    \begin{algorithmic}[1]
        \STATE Distribute $\overline{\lambda}_1$ to each local machine and each local machine computes $\H_k=\overline{\lambda}_1 \I -  \A_k^\top \A_k/m$.
         \FOR{$t=0,1,\dots, (T-1)$}
         \STATE Distribute $\w^{(t)}$ to each local machine and each local machine sets  $\w^{(t+1 )}_{0}=\w^{(t)}$

            \FOR{$j=0,1,\ldots, (T'-1)$}
              \FOR{each local machine $k=1,\ldots, K$}
                \STATE  Compute the local gradient information $\g_k = \H_k \w^{(t+1 )}_{j} - \w^{(t)}$
                \STATE  Transmit the local gradient information $\g_k$ to the central machine.
              \ENDFOR
              \STATE Calculate the global gradient information $\g=\frac{1}{K} \sum_{k=1}^K \g_k$.
              \STATE Perform the approximate Newton's step: $ \w^{(t+1)}_{j+1}= \w^{(t+1 )}_{j}- \H_1^{-1} \g$.
            \ENDFOR
         \STATE The central machine updates $\w^{(t+1)}=\frac{\w^{(t+1)}_{T'}}{\| \w^{(t+1)}_{T'}\|_2}$.
         \ENDFOR
        \STATE \textbf{Output:} $\w^{(T)}$.
    \end{algorithmic}
\end{algorithm}

Given $\w^{(0)}$ and $\overline{\lambda}_1$, we perform the \emph{shift-and-invert preconditioning} iteration in a distributed manner. In particular, for each iteration $t=0,1,\ldots,$
\begin{align}\label{eq:Rayleigh_iter}
  \tw^{(t+1)} = \left(\overline{\lambda}_1 \I - \frac{1}{n}\A^\top\A\right)^{-1} \w^{(t)},
  \; \; \w^{(t+1)} =\frac{\tw^{(t+1)}}{\|\tw^{(t+1)}\|_2}.
\end{align}
Therefore, the non-convex eigenvector estimation problem \eqref{eq:pca-prob} is reduced to solving a sequence of linear system. The key challenge is how to implement $(\overline{\lambda}_1 \I - \A^\top\A/n)^{-1}$ in a distributed setup.

To address this challenge, we formulate \eqref{eq:Rayleigh_iter} into a quadratic optimization problem. In particular, the update $\tw^{(t+1)}= (\overline{\lambda}_1 \I - \A^\top\A/n)^{-1} \w^{(t)}$ is equivalent to the following problem,
\begin{align}\label{eq:quadratic}
    \tw^{(t+1)} &=\argmin_{\w} \left[Q(\w):= \frac{1}{2} \w^\top \H \w - \w^\top \w^{(t)}\right], \\
    \H &\triangleq \overline{\lambda}_1 \I - \frac{1}{n}\A^\top\A . \notag
\end{align}
To solve this quadratic programming, the standard Newton's approach computes a sequence for $j=0,\ldots,$ with a starting point $\w^{(t+1)}_0 = \w^{(t)}$:
\begin{align}\label{eq:newton}
 \w^{(t+1)}_{j+1}=  \w^{(t+1)}_{j}- \left(\nabla^2 Q( \w^{(t+1)}_{j})\right)^{-1} \left[  \nabla Q( \w^{(t+1)}_{j}) \right],
\end{align}
where the Hessian matrix $\nabla^2 Q( \w^{(t+1)}_{j})$ is indeed $\H$. If we define, for each machine $k \in [K]$,
\begin{align} \label{eq:local-hessian}
\H_k &=\overline{\lambda}_1 \I - \frac{1}{m} \A_k^\top \A_k, \\
Q_k(\w) &= \frac{1}{2} \w^\top \H_k \w - \w^\top \w^{(t)}. \notag
\end{align}

\begin{algorithm}[!t]
    \caption{{\small Distributed Top-$L$-dim principal subspace}}
    \label{algo:distr_multi}
    \hspace*{\algorithmicindent} \hspace{-0.72cm} {\textbf{Input:} The data matrix $\A_k$ on each machine $k=1,\ldots, K$.  The number of top-eigenvectors $L$.}

    \begin{algorithmic}[1]
      \STATE Initialize $\V_0=[]$, $\A_{k,0}=\A_k$
         \FOR{$l=1,\dots, L$}
         \STATE Compute the initial $l$-th eigenvalue estimator $\overline{\lambda}_l$ and eigenvector estimator $\w_l^{(0)}$.
         \STATE Call Algorithm \ref{algo:distr_top} with $\{\A_{k,l-1}\}_{k=1}^K$ on each local machine to obtain $\w_l$ on the central machine.
         \STATE Project $\w_l$ to $\V_{l-1}^{\perp}$ by computing $\v_l = \frac{(\I-\V_{l-1}\V_{l-1}^\top) \w_l}{\|(\I-\V_{l-1}\V_{l-1}^\top) \w_l \|_2}$ as the estimated $l$-th eigenvector
         \STATE Update $\V_l=[\V_{l-1}, \v_l]$
         \STATE Transmit $\v_l$ to each local machine.
          \FOR{each local machine $k=1,\ldots, K$}
                \STATE  Update the data matrix $\A_{k,l}=\A_{k,l-1}(\I-\v_l\v_l^\top)$
          \ENDFOR
        \ENDFOR
        \STATE \textbf{Output:} $\V_L$.
    \end{algorithmic}
\end{algorithm}

It is easy to see that $\H=\sum_{k=1}^K \H_k / K$ and  $Q(\w)=\sum_{k=1}^K Q_k(\w) / K$. Therefore, in the Newton's update \eqref{eq:newton}, computing the full Hessian matrix $\nabla^2 Q( \w^{(t+1)}_{j})$ requires each machine to communicate a $d \times d$ local Hessian matrix $\H_k$ to the central machine. This procedure incurs a lot of communication cost. Moreover, taking the inverse of  the whole sample Hessian matrix $\H$ almost solves the original linear system~\eqref{eq:Rayleigh_iter}. To address this challenge, we adopt the idea from \cite{Shamir:14,jordan2019communication,Fan:19:communication}. In particular, we approximate the Newton's iterates by only using the Hessian information on the \emph{first machine}, which significantly reduces the communication cost. This approximated Newton's update can be written as,
\begin{align}\label{eq:approx_newton}
   \w^{(t+1)}_{j+1} & =  \w^{(t+1)}_{j}- \left(\nabla^2 Q_1( \w^{(t+1)}_{j})\right)^{-1} \left[  \nabla Q( \w^{(t+1)}_{j}) \right]\\
                    & =\w^{(t+1 )}_{j}- \H_1^{-1} \left[\frac{1}{K} \sum_{k=1}^K ( \H_k \w^{(t+1 )}_{j}- \w^{(t)})  \right], \notag
\end{align}
where $\H_1$ is the Hessian matrix of the first machine. This procedure can be computed easily in a distributed manner, i.e., each machine computes local gradient $\g_k = \H_k \w^{(t+1 )}_{j}- \w^{(t)}$, and these gradient vectors are communicated to the central machine for a final update $\g = \sum_{k=1}^K \g_k / K$. Therefore, in each inner iteration, the communication cost for each local nodes is only $\mathcal{O}(d)$. See Algorithm \ref{algo:distr_top} for a complete description.

\begin{remark} \label{rem:newton}
In this remark,  we explain why we choose the Newton approach in the  inner loop~\eqref{eq:approx_newton},  instead of the quasi-Newton method (the Broyden–Fletcher–Goldfarb–Shanno (BFGS) method) or other gradient methods with line search (e.g., Barzilai-Borwein gradient method in \cite{Wen2013AFM}). Due to the special structure of the PCA problem in our quadratic programming~\eqref{eq:quadratic}, the Hessian matrix is fixed and will not change over iterations. In other words, as shown in Algorithm~\ref{algo:distr_top}, we compute each local Hessian $\H_k$ for $k=1, \ldots, m$ (see Equation~\eqref{eq:local-hessian}), and the inverse of the local Hessian on the first machine $\H_1$ only once. Therefore, the distributed Newton method is computationally more efficient for the PCA problem. In comparison, BFGS is often used when the inverse of Hessian matrix $\H$ is hard to compute and changes over iterations, which is not the scenario of our PCA problem. Moreover, as we will show later in Section~\ref{sec:theory}, the Newton method has already achieved a linear convergence rate in the inner loop (see Lemma~\ref{lemma:inner_loop}), BFGS cannot be faster than that. In fact, although BFGS will eventually achieve a linear convergence rate, it can be quite slow at the very beginning with a crude Hessian inverse estimation.
\end{remark}

\begin{remark}
\label{rem:computation}
	In this remark, we compare the computational and communication costs between our method and the DC approach.
	Notice that the communication cost of our Algorithm~\ref{algo:distr_top} from each local machine is $\mathcal{O}(TT'd)$, where $TT'$ is the total number of  iterations. By our theoretical results in Section \ref{sec:theory} (see Corollary~\ref{corr:top_eigenvector}), for a targeting error rate $\varepsilon$, we only require $T$ and $T'$ and to be an logarithmic order of $1/\varepsilon$ (i.e., $TT' =\mathcal{O}(\log^2 (1/\varepsilon))$). Therefore, the total number of iterations is quite small. While it is more than $\mathcal{O}(d)$ communication cost of the DC approach, it is still considered as a communication efficient protocol. In distributed learning literature (e.g., \cite{jordan2019communication}), a communication-efficient algorithm usually refers to an algorithm that only transmits an $\mathcal{O}(d)$ vector (instead of $\mathcal{O}(d^2)$ Hessian matrices) at each iteration.

	When the full data of $n$ samples can be stored in the memory, the oracle PCA method incurs a computation cost (i.e., runtime) of $\mathcal{O}(nd^2+d^3)$, where $nd^2$ is for the computation of the sample covariance matrix and $d^3$ is for performing the eigen-decomposition. In the distributed setting with $m$ samples on each local machine, the DC approach incurs the computation cost of $\mathcal{O}(md^2 + d^3)$ since it is a one-shot algorithm. In comparison, our method incurs the $\mathcal{O}(md^2 + d^3 + TT' d^2)$ computational cost, in order to achieve the optimal convergence rate. We note that our method incurs one-time computational of the Hessian inverse with $\mathcal{O}(d^3)$ and each iteration only involves the efficient computation of the gradient (i.e., $\mathcal{O}(d^2)$). Therefore, the extra computational overhead over the DC $\mathcal{O}(TT' d^2)$  is a smaller order term in $d$ as compared to $\mathcal{O}(d^3)$. Moreover, the number of iterations $TT'$ is relatively small and thus the extra computation as compared to the DC is rather limited. In practice, one can easily combine two approaches. For example, one can initialize the estimator using the DC method, and further improve its accuracy using our method.
\end{remark}

For the top-$L$-dim eigenspace estimation, we extend a framework from \cite{Zhu:16:LazySVD} to our distributed settings. In our Algorithm \ref{algo:distr_multi}, we first compute the leading eigenvector $\v_1$ of $\A^\top\A / n$ in a distributed manner with Algorithm~\ref{algo:distr_top}. The $\v_1$ is then transfered back to local machines and used to right-project data matrix, i.e., $\A_k (\I_d - \v_1 \v_1^\top)$ for $k \in [K]$. The next eigenvector $\v_2$ is obtained with these projected data matrices and Algorithm~\ref{algo:distr_top}. In other words, we estimate the top eigenvector of $(\I_d - \v_1 \v_1^\top) \hSigma (\I_d - \v_1 \v_1^\top)$ in distributed settings. This procedure is repeated $L$ times until we obtain all the $L$ top eigenvectors $\V_L = [\v_1, \ldots, \v_L]$. This deflation technique is quite straight-forward and performs well in our later convergence analysis.

\begin{remark} \label{rem:centralization}
Our paper, and also the earlier works \citep{Zhu:16:LazySVD,Fan2017} all assume data vectors are centered, i.e., zero-mean data vectors $\ep[\a] = 0$. When the data is non-centered, we could adopt a two stage estimator, where the first stage centralizes the data in a distributed fashion and second stage applies our distributed PCA algorithm.  In particular, each local machine $k$ first computes the mean of local samples, i.e., $\bar{\a}_k=\frac{1}{m_k} \sum_{i \in \mathcal{D}_k} \a_i$, where $\mathcal{D}_k$ denotes the sample indices on the $k$-local machine and $m_k=|\mathcal{D}_k|$. Then each local machine transmits $(\a_k, m_k)$ to the center. The center computes their average $\bar{\a}=\frac{\sum_{k=1}^K m_k \bar{\a}_k }{ \sum_{k=1}^K m_k}$, which will be transmitted back to each local machine to center the data (i.e., each sample $\a_i$ will be $\a_i -\bar{\a}$). Given the centralized data, we can directly apply our distributed PCA algorithm. This centralization step only incurs one extra round of communication and each local machine only transmits an $O(d)$ vector to the center (which is the same amount of communication as in our algorithm that transmits the gradient).
\end{remark}

\section{Theoretical Properties}
\label{sec:theory}

This section exhibits the theoretical results for our setups in Section~\ref{sec:setting}.  The technical proofs will be relegated to the supplementary material (see Appendix~\ref{sec:pca-proof}).

\subsection{Distributed top eigenvector estimation}
\label{subsec:top-1-est}

We first investigate the theoretical properties of the top eigenvector estimation in Algorithm~\ref{algo:distr_top}. Let $\hSigma_k=\A_k^\top\A_k/m\in\R^{d\times d}$ denote the local sample covariance matrix on machine $k \in [K]$, and $\hSigma=K^{-1}\sum_{k=1}^K\hSigma_k$ the global sample covariance matrix using all data. Let $\hlambda_1\geq \hlambda_2\geq \cdots\geq \hlambda_d\geq 0$ and $\widehat{\u}_1,\widehat{\u}_2,\ldots,\widehat{\u}_d$ denote the sorted eigenvalues and associated eigenvectors of $\hSigma$. We are interested in quantifying the quality of some estimator $\w^{(t)}$. More specifically, we will reserve the letter $\delta$ to denote the relative eigenvalue gap threshold, and will measure the closeness between $\w^{(t)}$ and the top eigenvector $\widehat \u_1$ via proving
\begin{align}
\label{eq:gap-free_bnd}
\sum_{l:\, \hlambda_l \leq (1-\delta)\,\hlambda_1}|\langle \hu_l,\w^{(t)}\rangle|^2 \leq \frac{\varepsilon^2}{\delta^2},
\end{align}
for error $\varepsilon > 0$ and any constant $\delta \in (0,1)$. In particular, the result~\eqref{eq:gap-free_bnd} above is always stronger (modulo constants) than the usual bound
\begin{align}
\label{eq:normal_bnd}
\widehat\theta^{(t)}:\,=\arccos |\langle \hu_1,\w^{(t)}\rangle| \leq C\,\varepsilon\, \frac{\hlambda_1}{\hlambda_1-\hlambda_2},
\end{align}
that involves the relative gap between the first two eigenvalues of $\hSigma$. Here $C$ is a constant. To see this, we can simply choose $\delta = (\hlambda_1-\hlambda_2)/\hlambda_1$ in Equation~\eqref{eq:gap-free_bnd}. Then $\sin^2 \widehat\theta^{(t)}=1-|\langle \hu_1,\w^{(t)}\rangle|^2=\sum_{l:\, \hlambda_l \leq (1-\delta)\,\hlambda_1}|\langle \hu_l,\w^{(t)}\rangle|^2\leq \varepsilon^2/\delta^2$, implying
$\widehat\theta^{(t)} \leq \arcsin (\varepsilon/\delta) \leq C\,\varepsilon\,\hlambda_1/(\hlambda_1-\hlambda_2)$ for some universal constant $C>0$. Moreover, it has to be assumed that $\hlambda_1 > \hlambda_2$ in the usual bound \eqref{eq:normal_bnd}, which may not be held in some applications.

From our enlarged eigenspace viewpoint, the result in Equation~\eqref{eq:gap-free_bnd} indicates that the top eigenvector estimator $\w^{(t)}$ is almost covered by the span of $\{\hu_l:\, \hlambda_l > (1-\delta)\,\hlambda_1\}$. 

As we will show in the theoretical analysis later, the success of our algorithm relies on the initial values of both eigenvalue and eigenvector. We first clarify our choice of initial eigenvalue estimates.

For the top eigenvector estimation in Algorithm~\ref{algo:distr_top}, since we have the following high probability bound (see Equation~\eqref{eq:eigenvalue_concent} in Lemma~\ref{lemma:mat_concentration} in the appendix for the justification), $\matnorm{\hSigma-\hSigma_1}{2} \leq \eta/2$ for some constant $\eta>0$. If we choose $\overline\lambda_1^{(0)} = \lambda_1(\hSigma_1)+3\eta/2$, then it is guaranteed that, $2\eta\geq \overline\lambda_1^{(0)} - \hlambda_1 \geq\eta$. Lemma~\ref{lemma:mat_concentration} (Equation~\eqref{eq:eigenvector_concent}) also provides a concentration bound on our initial value of eigenvectors, $\sum_{l:\, \hlambda_l \leq (1-\delta)\,\hlambda_1} \left| \langle \hu_l, \w^{(0)} \rangle\right|^2 \leq 3/4$ with high probability. Here $\{\hu_l: \hlambda_l \leq (1-\delta)\,\hlambda_1\}$ are all the eigenvectors for the full sample covariance matrix $\hSigma$ whose associated eigenvalues have a relative gap $\delta$ from the largest eigenvalue $\hlambda_1$ and $\w^{(0)}$ is the top eigenvector for the sample covariance matrix on the first machine. 

Given our initial estimators $\overline{\lambda}_1^{(0)}$ and $ \w^{(0)}$, we have the following convergence guarantee for our Algorithm~\ref{algo:distr_top}. With the above guarantees of initial estimator $\overline{\lambda}_1$ and $\w^{(0)}$, our first lemma characterizes the convergence rate of the \emph{outer loop} in Algorithm~\ref{algo:distr_top}.

\begin{lemma} \label{lemma:outer_loop}
Suppose the initial estimator $\overline{\lambda}_1$ satisfies
\begin{align} \label{eq:eta}
\eta \leq \overline{\lambda}_1-\hlambda_1\leq 2\eta\quad\mbox{for some $\eta>0$}.
\end{align}
For any $\w\in\R^d$, and $\v\in\R^d$ that satisfies
\begin{align}\label{eq:assumption_out1}
\|\w\|_2=1, \quad \sum_{l:\, \hlambda_l \leq (1-\delta)\,\hlambda_1} \left| \langle \hu_l, \w \rangle\right|^2 \leq \frac{3}{4},
\end{align}
and
\begin{align}
\label{eq:assumption_out2}
\|\v-\H^{-1} \w\|_2\leq \varepsilon \leq (8\eta)^{-1},
\end{align}
and for each index $l=1,\ldots$ such that $\hlambda_l \leq (1-\delta)\,\hlambda_1$, we have
\begin{align}\label{eq:con_outer}
\frac{| \langle \widehat\u_l, \v\rangle |}{\|\v\|_2} \leq \frac{8\eta}{\delta\hlambda_1}\, \frac{|\langle \widehat\u_l, \w\rangle|}{\|\w\|_2} +8\eta\varepsilon.
\end{align}
Moreover, we have
\begin{align}\label{eq:con_outer_2}
\sum_{l:\, \hlambda_l \leq (1-\delta)\,\hlambda_1} \frac{| \langle \widehat\u_l, \v\rangle |^2}{\|\v\|_2^2} \leq \frac{128\eta^2}{\delta^2\hlambda_1^2}\,\sum_{l:\, \hlambda_l \leq (1-\delta)\,\hlambda_1} \frac{|\langle \widehat\u_l, \w\rangle|^2}{\|\w\|_2^2} +128\eta^2\varepsilon^2.
\end{align}
\end{lemma}

For the outer loop in our Algorithm~\ref{algo:distr_top}, $\w$ and $\v / \|\v\|_2$ in Lemma~\ref{lemma:outer_loop} can be explained as the $t$-th round and $(t+1)$-th round estimators $\w^{(t)}$ and $\w^{(t+1)}$, respectively. This lemma implies that up to a numerical tolerance $\varepsilon$ for inverting $\H$ (Condition~\eqref{eq:assumption_out2}), each application of the outer loop reduces the magnitude of the projection of $\w^{(t)}$ onto $\widehat\u_l$ by a factor of $\mathcal O\big((\delta\hlambda_1)^{-1}\eta\big) \ll 1$ given $\eta\ll1$ (if we have a good initial estimator of $\hlambda_1$ and $\delta\hlambda_1 = \O(1)$). Notice that if $\w^{(t)}$ satisfies condition~\eqref{eq:assumption_out1}, our Equation~\eqref{eq:con_outer_2} claims that $\w^{(t+1)} = \v / \|\v \|_2$ satisfies Condition~\eqref{eq:assumption_out1} as well. This condition is justified if $\w^{(0)}$ satisfies Condition~\eqref{eq:assumption_out1}, which is a conclusion from Lemma~\ref{lemma:mat_concentration} in the appendix.

Our second lemma characterizes the convergence rate of distributively solving the linear system $\H\w=\w^{(t)}$ in the \emph{inner loop} of Algorithm~\ref{algo:distr_top}. Recall that in Equation~\eqref{eq:Rayleigh_iter}, $\tw^{(t+1)}=\H^{-1}\w^{(t)}$ denote the exact solution of this linear system.

\begin{lemma}\label{lemma:inner_loop}
Suppose the initial estimator $\overline{\lambda}_1$ satisfies
\begin{align*}
\overline{\lambda}_1-\hlambda_1\geq \eta \geq \frac{1}{2}\,\matnorm{\hSigma-\hSigma_1}{2}.
\end{align*}
Then for each $j=0,1,\ldots,(T'-1)$, we have
\begin{align}\label{eq:con_inner}
\|\w_{j+1}^{t+1} - \tw^{(t+1)}\|_2 \leq \frac{2\,\matnorm{\hSigma-\hSigma_1}{2}}{\eta} \, \|\w_{j}^{t+1} - \tw^{(t+1)}\|_2.
\end{align}
\end{lemma}

Here $\matnorm{\hSigma-\hSigma_1}{2}$ on the RHS of \eqref{eq:con_inner} is due to the approximation using the Hessian matrix $\H_1$ on the first machine in place of original Hessian matrix $\H$. As we will show later, by standard matrix concentration inequalities, we have $\matnorm{\hSigma-\hSigma_1}{2} = \mathcal{O}\left(\sqrt{d/m}\right)$ with high probability. As a consequence, the inner loop of Algorithm~\ref{algo:distr_top} has a contraction rate of order $\mathcal O(\eta^{-1}\sqrt{d/m})$, which is inversely proportional to the gap $\overline\lambda_1-\hlambda_1$ (due to the condition number of the Hessian $\H$).

Combining these two lemmas, we come to our first main theoretical result for the convergence rate of Algorithm~\ref{algo:distr_top}.

\begin{theorem}\label{thm:top_eigenvector}
Let $\kappa:= \matnorm{\hSigma-\hSigma_1}{2} =\mathcal O_P(\sqrt{d/m})$. Assume
\begin{align*}
2\eta \geq \overline{\lambda}_1-\hlambda_1\geq \eta \geq \frac{1}{2}\,\kappa,
\end{align*}
and the initial eigenvector estimator $\w^{(0)}$ satisfies
\begin{align*}
 \sum_{l:\, \hlambda_l \leq (1-\delta)\,\hlambda_1} \left| \langle \hu_l, \w^{(0)} \rangle\right|^2 \leq \frac{3}{4}.
\end{align*}
Then for each $T$ and $T'$ as the outer and inner iterations in Algorithm~\ref{algo:distr_top}, respectively, and the relative eigenvalue gap $\delta\in(0,1)$, we have
\begin{align}
\label{eq:top_eigenvector}
\sum_{l:\, \hlambda_l \leq (1-\delta)\,\hlambda_1}|\langle \widehat\u_l,\w^{(t)}\rangle|^2 \leq \Big(\frac{128\eta^2}{\delta^2\hlambda_1^2}\Big)^T + \frac{512\,\eta}{1-128\eta^2/(\delta\hlambda_1)^2} \, \Big(\frac{4\kappa^2}{\eta^2}\Big)^{T'}.
\end{align}
\end{theorem}

We can further simplify Equation~\eqref{eq:top_eigenvector} by choosing proper $\eta$ and $T'$.
\begin{corollary} \label{corr:top_eigenvector}
In particular, if $\eta\leq \delta\hlambda_1/16$, and we choose $T'=T$, and $\eta =\big(\kappa\delta\hlambda_1\big)^{1/2}/3=\mathcal O_P(\sqrt[4]{d/m})$, then the final output $\w^{(T)}$ satisfies
\begin{align} \label{eq:top_eigenvector2}
\sum_{l:\, \hlambda_l \leq (1-\delta)\,\hlambda_1}|\langle \widehat\u_l,\w^{(T)}\rangle|^2  \leq 257\,  \Big(\frac{6\kappa}{\delta\hlambda_1}\Big)^{2T}.
\end{align}
\end{corollary}

As indicated in \eqref{eq:top_eigenvector2}, when $6\kappa / \delta\hlambda_1 \ll 1$, our Algorithm~\ref{algo:distr_top} enjoys a linear convergence rate. Moreover, to ensure this convergence, when the absolute eigengap $\delta \hlambda_1$ is small, $\kappa = \mathcal O_P(\sqrt{d/m})$ needs to be smaller, i.e., $\kappa = o(\delta \hlambda_1)$ Recall that $\kappa:= \matnorm{\hSigma-\hSigma_1}{2}$, which is defined in Theorem \ref{thm:top_eigenvector}. This indicates that more samples are needed on each local machine.

\begin{remark}
\label{rem:meaning_gap_free}
Under the setting without an explicit eigengap, our goal is not to construct a good estimator of the top eigenvector. Instead, we aim to construct an estimator that captures a similar amount of variability in the sample data as the top eigenvector. Recall that by Theorem~\ref{thm:top_eigenvector}, we construct an estimator $\w$ such that $\sum_{l:\, \hlambda_l \leq (1-\delta)\,\hlambda_1}|\langle \hu_l,\w\rangle|^2 \leq \varepsilon$ for some error term $\varepsilon > 0$. We can see that,
\begin{align}\label{eq:cov_ineq}
\w^\top \hSigma \w > \left(1-\delta\right)\left(1-\varepsilon\right)  \hlambda_1.
\end{align} This fact can be easily derived as follows (see also the proof of Theorem 3.1 in \cite{Zhu:16:LazySVD}),
\begin{align*}
\w^\top \hSigma \w = \sum_{i=1}^d \hlambda_i (\w^\top \hu_i)^2 & \geq \sum_{l:\, \hlambda_l > (1-\delta)\,\hlambda_1}  \hlambda_l (\w^\top \hu_l)^2 \geq (1-\delta) \hlambda_1 \sum_{l:\, \hlambda_l > (1-\delta)\,\hlambda_1}(\w^\top \hu_l)^2 \\
&\geq \left(1-\delta\right)\left(1-\varepsilon\right) \hlambda_1.
\end{align*}
According to \eqref{eq:cov_ineq}, our estimator $\w$ captures almost the same amount of variability of the sampled data (up to a $(1-\delta)(1-\varepsilon)$ multiplicative factor). This type of results are also known as the gap-free bound in some optimization literature (see e.g., \cite{Zhu:16:LazySVD}).

When the eigengap is extremely small, identifying the top eigenvector is an information-theoretically difficult problem. As an extreme case, when the gap is zero, it is impossible to distinguish between the top and the second  eigenvectors. In contrast, our setting is favorable in practice since the main goal of PCA/dimension reduction is to capture the variability of the data.

Moreover, we note that the parameter $\delta$  is a pre-specified parameter that measures the proportion of the variability  explain by the estimator $\w$.  For example, when setting $\delta=\varepsilon$, the estimator $\w$ will capture at least $(1-2\varepsilon)$ of the variability captured by the top eigenvector according to \eqref{eq:cov_ineq}. We can also choose $\delta=c_0/\hat \lambda_1$ for some constant $c_0$ so that $\delta \hat \lambda_1=\Omega(1)$.
\end{remark}

\subsection{Distributed top-$L$-dim principal subspace estimation}

With the theoretical results for the top eigenvector estimation in place, we further present convergence analysis on the top-$L$-dim eigenspace estimation in Algorithm~\ref{algo:distr_multi}.

Let $\hU_{\leq (1- \delta)\hlambda_L} = [\hu_{S+1},\ldots,\hu_d]$ denote the column orthogonal matrix composed of all eigenvectors of $\hSigma$ whose associated eigenvalues have a relative gap $\delta$ from the $L$-th largest eigenvalue $\hlambda_L$, that is, $S:\,=\argmax\{l:\,\hlambda_l > (1-\delta)\,\hlambda_L\}$.  We also denote $\hU_{> (1- \delta)\hlambda_L} = [\hu_1, \ldots, \hu_{S}]$ to be the enlarged eigenspace corresponding to the eigenvalues larger than $(1-\delta)\hlambda_L$.

We also use the notation $\hSigma^{(l)}=(\I-\V_{l-1}\V_{l-1}^\top) \hSigma (\I-\V_{l-1}\V_{l-1}^\top)$ for $l=0,1,\ldots,L-1$. Here $\V_l = [\v_1, \ldots, \v_l]$ consists of all the top-$l$ eigenvector estimations and $\V_0 = \0$. Notice that $\hSigma^{(l)}$ is just the matrix $\A^{(l)}\A^{(l)}/n$ where $\A^{(l)}:\,=[\A_{1,l}^\top,\ldots,\A_{K,l}^\top]^\top$ and $\A_{k,l}^T$ is the projected data matrix on machine $k$ ($k \in [K]$) for the $l$-th eigenvector estimation.

We first provide our choices of initial eigenvalue estimates. For Algorithm \ref{algo:distr_multi} for the top-$L$-dim principal, let $\hSigma_{k}^{(l)}=\A_{k,l}^\top\A_{k,l}/m$ and $\hSigma^{(l)}=K^{-1}\sum_{k=1}^K\hSigma_{k}^{(l)}$ denote the local and global projected sample covariance matrices at the outer iteration $l$. For the same constant $\eta$ defined above in \eqref{eq:eta}, we choose $\overline\lambda_l = \lambda_1(\hSigma_{1}^{(l)})+3\eta/2$ for $l\in [L]$. This follows from
\begin{align*}
\matnorm{\hSigma^{(l)}-\hSigma_{1}^{(l)}}{2} =\matnorm{(\I-\V_l\V_l^\top)(\hSigma-\hSigma_1)(\I-\V_l\V_l^\top)}{2} \leq \matnorm{\hSigma-\hSigma_1}{2} \leq \frac{\eta}{2},
\end{align*}
which implies that,
\begin{align*}
2\eta\geq \overline\lambda_l - \lambda_1(\hSigma^{(l)}) \geq \eta.
\end{align*}

Our main result is summarized as follows.

\begin{theorem}\label{thm:top_L_eigenvectors}
Let $\kappa = \matnorm{\hSigma-\hSigma_1}{2} =\mathcal O_P(\sqrt{d/m})$. Assume \begin{align*}
2\eta \geq \overline{\lambda}_l-\lambda_1(\hSigma^{(l)})\geq \eta \geq \frac{1}{2}\,\kappa,
\end{align*}
for each $l \in [L]$, where $\lambda_1(\hSigma^{(l)})$ denotes the largest eigen value of $\hSigma^{(l)}$.
Then we have
\begin{align}
\label{eq:top_L_eigenvectors}
\matnorm{\hU_{\leq (1- \delta)\hlambda_L}^{\top} \V_L}{2}^2 \leq \frac{64\hlambda_1 L^2}{\hlambda_L\delta}\sqrt{\Big(\frac{128\eta^2}{\delta^2\hlambda_L^2}\Big)^T + \frac{512\,\eta}{1-128\eta^2/(\delta\hlambda_L)^2} \, \Big(\frac{4\kappa^2}{\eta^2}\Big)^{T'}}.
\end{align}
\end{theorem}

By choosing specific settings of some parameters, the result in Theorem \ref{thm:top_L_eigenvectors} can be simplified as shown in the following corollary.
\begin{corollary}
\label{corr:top_L_eigenvector}
Similarly, if $\eta\leq \delta\hlambda_1/16$, and we choose $T'=T$, and $\eta =\big(\kappa\delta\hlambda_1\big)^{1/2}/3=\mathcal O_P(\sqrt[4]{d/m})$, then our estimator $\V_L$ satisfies
\begin{align*}
\matnorm{\hU_{\leq (1- \delta)\hlambda_L}^{\top} \V_L}{2}^2 = \mathcal{O} \left(\frac{\hlambda_1 L^2}{\hlambda_L\delta} \,  \Big(\frac{6\kappa}{\delta\hlambda_L}\Big)^{T} \right).
\end{align*}
\end{corollary}

Here we could also interpret our results in Theorem~\ref{thm:top_L_eigenvectors} from an ``angle'' point of view corresponding to the classical $\sin \Theta$ result. Since there is no eigengap assumption, it is impossible to directly estimate $\U_L$. Therefore, we choose a parameter $\delta$, and consider \emph{an enlarged eigenspace} $\U_{> (1- \delta)\hlambda_L}$. Our theoretical results (see Theorem~\ref{thm:top_L_eigenvectors} and Corollary~\ref{corr:top_L_eigenvector}) imply that the ``angle'' between our estimator $\V_L$ and $\hU_{\leq (1- \delta)\hlambda_L}$ is sufficiently small. This result extends the classical $\sin \Theta$ result.

Similar to the top-eigenvector case in Remark \ref{rem:meaning_gap_free}, our estimator $\V_L$ can also capture a similar amount of variability in the sampled data to $\hU_L:=\{\hu_1, \ldots, \hu_L\}$. We further describe this property in the following Corollary~\ref{corr:population-gap-free2}.

\begin{corollary} \label{corr:population-gap-free2} Assume our estimator $\V_L$ from Algorithm~\ref{algo:distr_multi} satisfies $\matnorm{\hU_{\leq (1- \delta)\hlambda_L}^{\top} \V_L}{2} \leq \frac{\delta}{16 \hlambda_1 / \hlambda_{L+1}}$, then we have,
\begin{align}
\label{eq:cov_ineq2}
&\hlambda_{L+1} \leq \matnorm{ \left(\I_d - \V_L \V_L^\top\right) \hSigma \left(\I_d - \V_L \V_L^\top\right) }{2} \leq \frac{\hlambda_{L+1}}{1-\delta}, \\
\label{eq:cov_ineq3}
&(1-\delta) \hlambda_l \leq \v_l^\top \hSigma \v_l \leq \frac{1}{1-\delta} \hlambda_l, \; \; \forall l \in [L].
\end{align}
\end{corollary}

Now we further extend the result in Corollary \ref{corr:top_L_eigenvector} to quantify the ``angle'' between our estimator $\V_L$ and the population eigenspace $\U_{\leq (1- 2\delta) \lambda_L}$.

\begin{corollary} \label{corr:population-gap-free} Assume our estimator $\V_L$ from Algorithm~\ref{algo:distr_multi} satisfies $\matnorm{\hU_{\leq (1- \delta)\hlambda_L}^{\top} \V_L}{2} \leq \varepsilon$ for some error term $\varepsilon > 0$, then we have,
\begin{align}
\label{eq:top_L_eigenvectors-population}
\matnorm{\U_{\leq (1- 2\delta)\lambda_L}^{\top} \V_L}{2} \leq \frac{\matnorm{\S - \hSigma}{2}}{(1 - \delta)(\hlambda_L - \lambda_L) + \delta \lambda_L} + \varepsilon,
\end{align}
where $\U_{\leq (1- 2\delta)\lambda_L}$ is the eigenvectors of the population covariance matrix $\S$ corresponding to eigenvalues less than or equal to $(1- 2\delta)\lambda_L$.
\end{corollary}

We further provide a different ``angle'' result on quantifying the complement of an enlarged space of $V_L$. Recall our definition $S :\,=\argmax\{l:\,\hlambda_l > (1-\delta_L)\,\hlambda_L\}$. We can classify $\hu_1, \ldots,   \hu_d$ and correspondingly our estimators $\v_1, \ldots, \v_d$ from Algorithm \ref{algo:distr_multi} into three regimes:
\begin{align}\label{eq:enlarged}
&\overbrace{\underbrace{\hu_1, \ldots, \hu_L}_{\hU_L}, \; \hu_{L+1}, \ldots, \hu_S}^{\hU_{S}}, \; \underbrace{\hu_{S+1}, \ldots, \hu_d}_{\hU_{\leq (1-\delta) \hlambda_L}} \\
&\overbrace{\underbrace{\v_1, \ldots, \v_L}_{\V_L}, \v_{L+1}, \ldots, \v_S}^{\V_{S}},  \underbrace{\v_{S+1}, \ldots, \v_d}_{\V_{\leq (1-\delta) \hlambda_L}} \nonumber
\end{align}
Corollary \ref{corr:top_L_eigenvector} shows that the ``angle'' between our estimator $\V_L$ and $\hU_{\leq (1-\delta) \hlambda_L}$ is sufficiently small. Similarly, we can show the counterpart of this result, which indicates that the ``angle'' between $\hU_L$ and $\V_{\leq (1-\delta) \hlambda_L}$ is also very small.
This result will be useful in our principal component regression example. To introduce our result, we denote $\hlambda_S$ to be the $S$-th largest eigenvalue of $\hat\S$.

\begin{theorem} \label{thm:top_L_eigenvectors2}
By running Algorithm \ref{algo:distr_multi} for obtaining the distributed top-$S$-dim principal subspace estimator $\V_S$, if there exists $\delta_S<\delta$ such that  $\matnorm{\hU_{\leq (1- \delta_S)\hlambda_S}^{\top} \V_S}{2} \leq \frac{\delta_S}{16 \hlambda_1 / \hlambda_{S+1}}$, then we have for the empirical eigenspace $\hU_L$
\begin{align}
	\label{eq:top_L_eigenvectors2}
	\matnorm{\hU_L^\top \V_{\leq (1-\delta) \hlambda_L}}{2} \leq S \; \frac{\delta_S \hlambda_1}{\hlambda_L(1-\delta_S) - \hlambda_S}.
\end{align}
Furthermore, for the population eigenspace $\U_L$, we can derive that
\begin{align}
    \label{eq:top_L_eigenvectors4}
    \matnorm{\U_L^\top \V_{\leq (1-2 \delta) \hlambda_L}}{2} \leq S \; \frac{\delta_S \hlambda_1}{\hlambda_L(1-\delta_S) - \hlambda_S} + \frac{\matnorm{\S - \hSigma}{2}}{(1 - \delta)(\hlambda_L - \lambda_L) + \delta \lambda_L}.
\end{align}
\end{theorem}
The reason why we impose the upper bound  on $\matnorm{\hU_{\leq (1- \delta_S)\hlambda_S}^{\top} \V_S}{2}$  is mainly to obtain the result in Equation~\eqref{eq:cov_ineq2} for $\V_S$. We also note that this upper bound can be easily satisfied as long as we run Algorithm \ref{algo:distr_multi} for sufficiently large number of iterations.

Let us recall the classical Davis-Kahan result for PCA in Equation~\eqref{eq:davis-kahan}. As we explained in the introduction, without an eigengap condition, the estimation  error can be arbitrarily large. However, our enlarged eigenspace estimator $\V_S$ in \eqref{eq:enlarged} (i.e., $\V_{> (1-\delta) \hlambda_L}$) will almost contain the top-$L$-dim eigenspace of the population covariance matrix. In particular, by Equation \eqref{eq:top_L_eigenvectors4} and Lemma~\ref{Lem:Qexist} in the appendix, we have shown that there exists a matrix $\Q$ satisfying $\matnorm {\Q}{2}\leq 1$ such that the error bound $\matnorm{\U_{L} -\V_{> (1-2\delta) \hlambda_L} \Q}{2}$ is sufficiently small.

Our enlarged eigenspace results find important applications to many statistical problems. In particular,  in Appendix~\ref{sec:application} in the supplement, we illustrate how the theoretical results can be applied to the principal component regression (Example 1) and the single index model (Example 2). We also provide simulation studies of these two applications in Appendix~\ref{sec:add-exp} in the supplement.

\begin{remark}
It is also worthwhile to note that we assume the data are evenly split only for the ease of discussions. In fact, the local sample size $m$ in our theoretical results is the sample size on the first machine (or any other machine that used to compute the estimation of Hessian $\H$) in Algorithm~\ref{algo:distr_top} and Algorithm~\ref{algo:distr_multi}. As long as the sample size $m$ on the first machine is specified, our method does not depend on the partition of the entire dataset.
\label{rem:balance}
\end{remark}

\section{Numerical Study}
\label{sec:exp}

In this section, we provide simulation experiments to illustrate the empirical performance of our distributed PCA algorithm.

Our data follows a normal distribution, $\ep[\a] = 0$ and the population covariance matrix $\ep[\a \a^\top] = \S$ is generated as follows:
\begin{align*}
\S = \U \La \U^T,
\end{align*}
where $\U$ is an orthogonal matrix generated randomly and $\La$ is a diagonal matrix. Since our experiments mainly estimate the top-$3$ eigenvectors, $\La$ has the following form,
\begin{align}
\label{eq:cov-mat}
\La = \mathrm{diag}\left(1+3\delta, 1+2\delta, 1+\delta, 1, \ldots, 1\right).
\end{align}
For example, when the relative eigengap $\delta$ is 1, $\La = \mathrm{diag}(4, 3, 2, 1, \ldots, 1)$.

For orthogonal matrix $\U = [u_{ij}] \in \R^{d \times d}$, we first generate all elements $u_{ij}, i,j = 1, \ldots, d$ such that they are \emph{i.i.d.} standard normal variables. We then use Gram-Schmidt process to orthonormalize the matrix and obtain the $\U$.

We will compare our estimator with the following two estimators:

(1) Oracle estimator: the PCA estimator is computed in the single-machine setting with pooled data, i.e., we gather all the sampled data and compute the top eigenspace of $\hSigma = \frac{1}{n} \A \A^\top$, where  $\A \in \R^{n \times d}$ i the data matrix.

(2) DC estimator (Algorithm 1 in \cite{Fan2017}): it first computes the top-$L$-dim eigenspace estimation $\hU_L^{(k)}, k=1, \ldots, K$ on each machine, and merges every local result together with $\tSigma = \frac{1}{K} \sum_{k=1}^K \hU_L^{(k)} \hU_L^{(k) \top}$. The final estimator is given by the eigenvalue decomposition of $\tSigma$.

Note that all the reported estimation errors are computed based on the average of $100$ Monte-Carlo simulations. Since the standard deviations of Monte-Carlo estimators for all the methods are similar and sufficiently small, we omit standard deviation terms in the following Figures and only report the average errors for better visualization. As shown in the following subsections, our distributed algorithm gets to a very close performance with the oracle one when the number of outer iterations $T$ is large enough and outperforms its divide-and-conquer counterpart.

For distributed PCA, we adopt the following error measurements from the bound~\eqref{eq:top_eigenvector} and bound~\eqref{eq:top_L_eigenvectors} with population eigenvectors replacing the oracle estimator. To be more specific, for the top eigenvector case, with the estimator $\hu_1$, population eigenvectors $\u_1, \ldots, \u_d$, population eigenvalues $\lambda_1, \ldots, \lambda_d$ and relative eigenvalue gap $\delta \in ( 0,1 )$, the error measurement is defined as
\begin{align} \label{eq:err-top}
\mathrm{error}(\hu_1) = \sum_{l: \lambda_l \leq (1 - \delta) \lambda_1} \left| \left\langle \u_l, \hu_1 \right\rangle \right|^2.
\end{align}

As for the top-$L$-dim eigenspace estimation,  let $\tU = [\u_{l_\delta},\ldots,\u_d]$ be the column orthogonal matrix composed of all eigenvectors of population covariance $\S$ whose associated eigenvalues have a relative gap $\delta$ from the $L$-th largest eigenvalue $\lambda_L$. That is, $l_\delta:\,=\argmin\{l:\,\hlambda_l \leq (1-\delta)\,\hlambda_L\}$. Recall that  $\hU_L$ is the estimator the top-$L$ eigenvectors. Then the corresponding error should be
\begin{align} \label{eq:err-multi}
\mathrm{error}(\hU_L) = \matnorm{\tU^\top \hU_L}{2}^2.
\end{align}

\begin{figure}[!t]
\centering
\subfigure[Top-$1$-dim eigenvector]{
\includegraphics[width=0.29\textwidth]{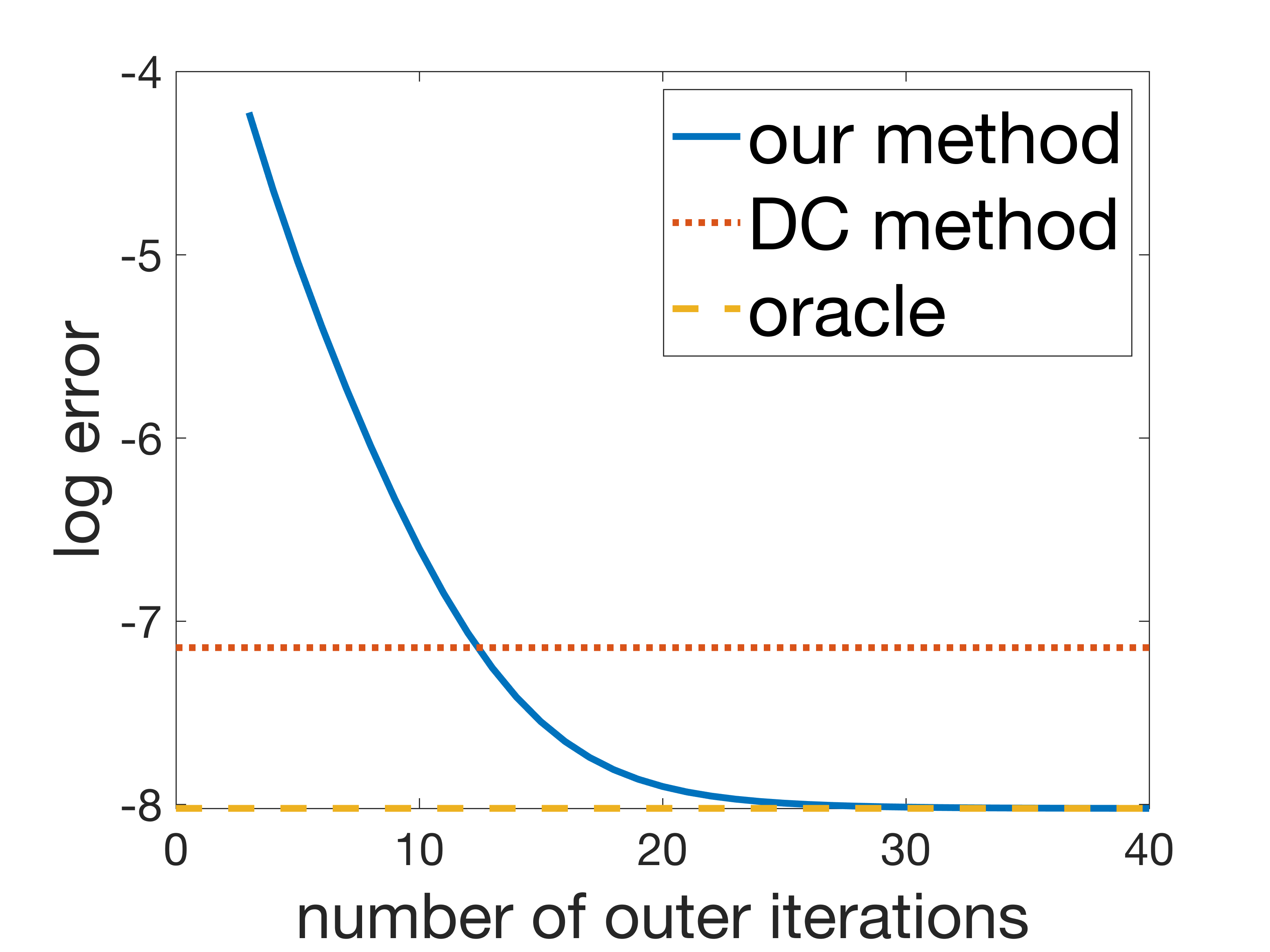}}
\subfigure[Top-$2$-dim eigenspace]{
\includegraphics[width=0.29\textwidth]{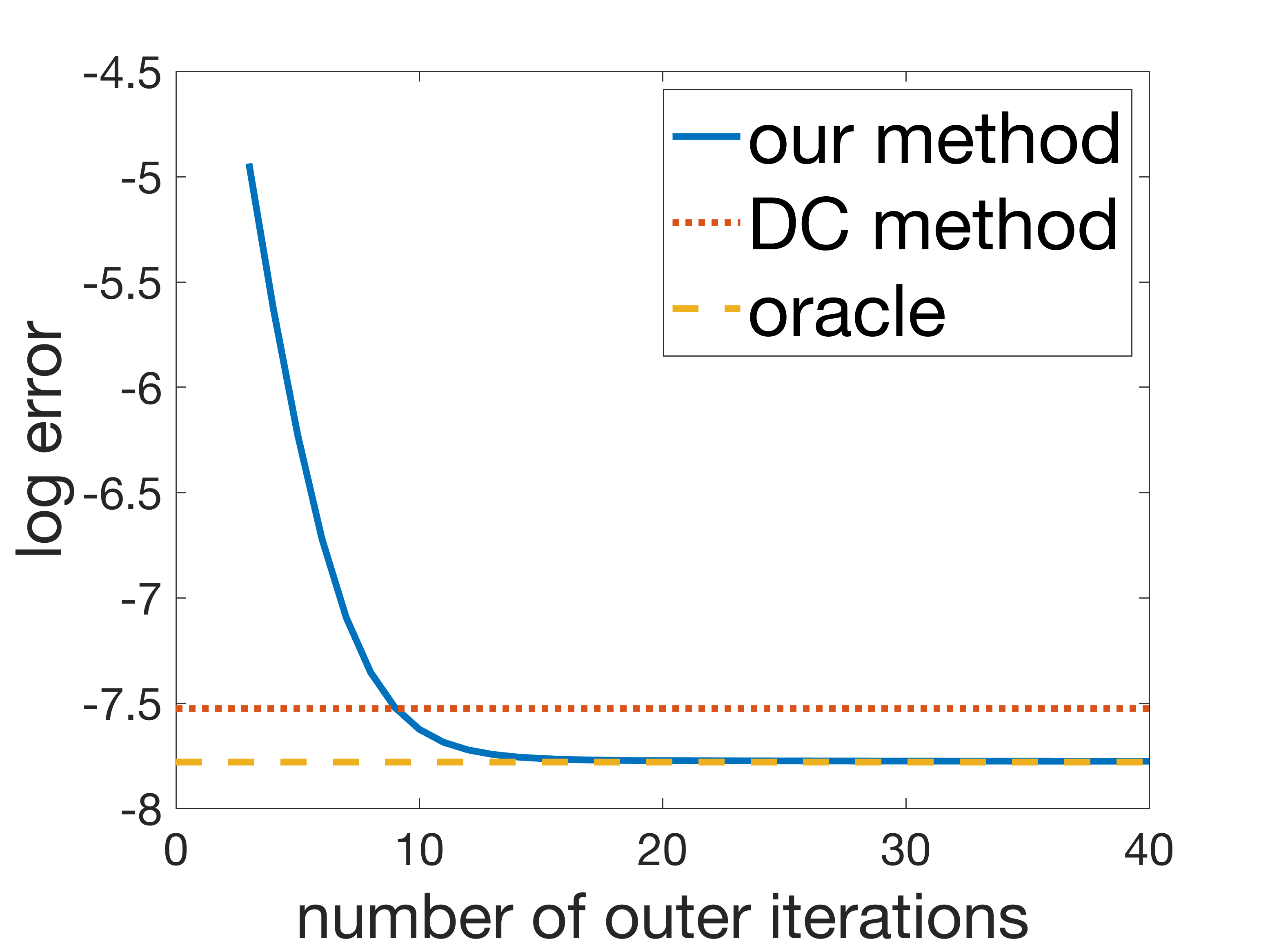}}
\subfigure[Top-$3$-dim eigenspace]{
\includegraphics[width=0.29\textwidth]{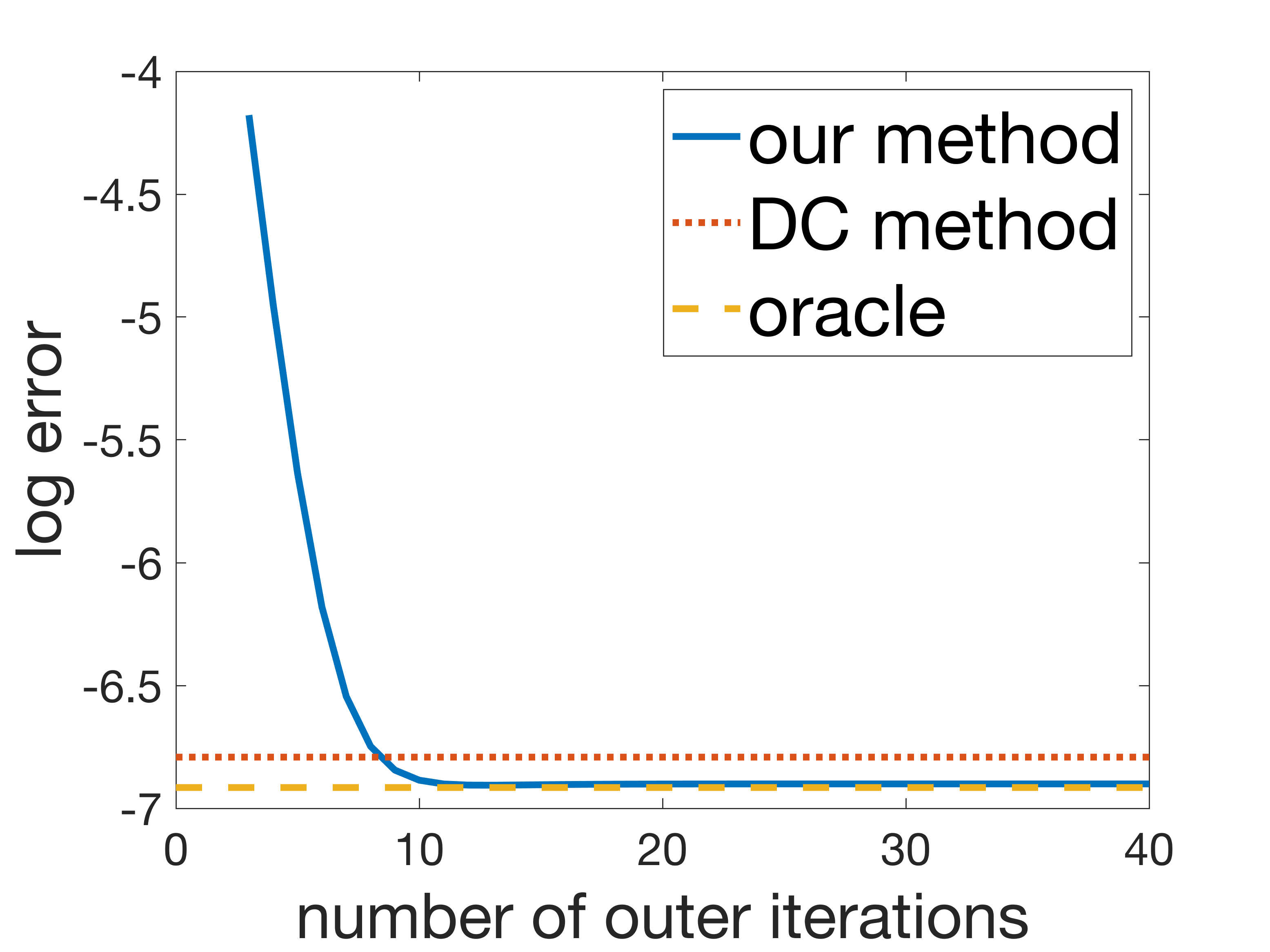}}
\subfigure[Top-$1$-dim eigenvector]{
\includegraphics[width=0.29\textwidth]{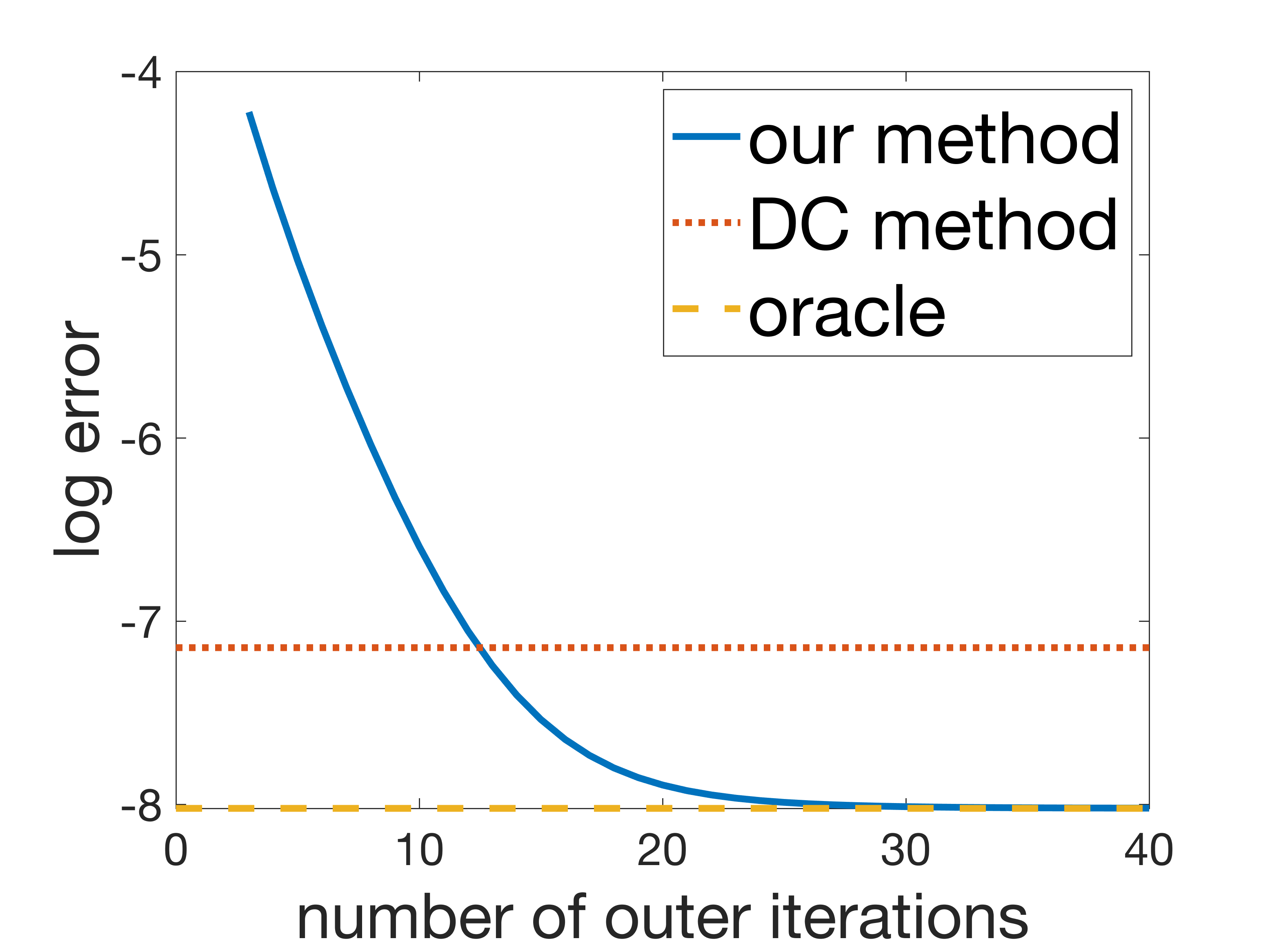}}
\subfigure[Top-$2$-dim eigenspace]{
\includegraphics[width=0.29\textwidth]{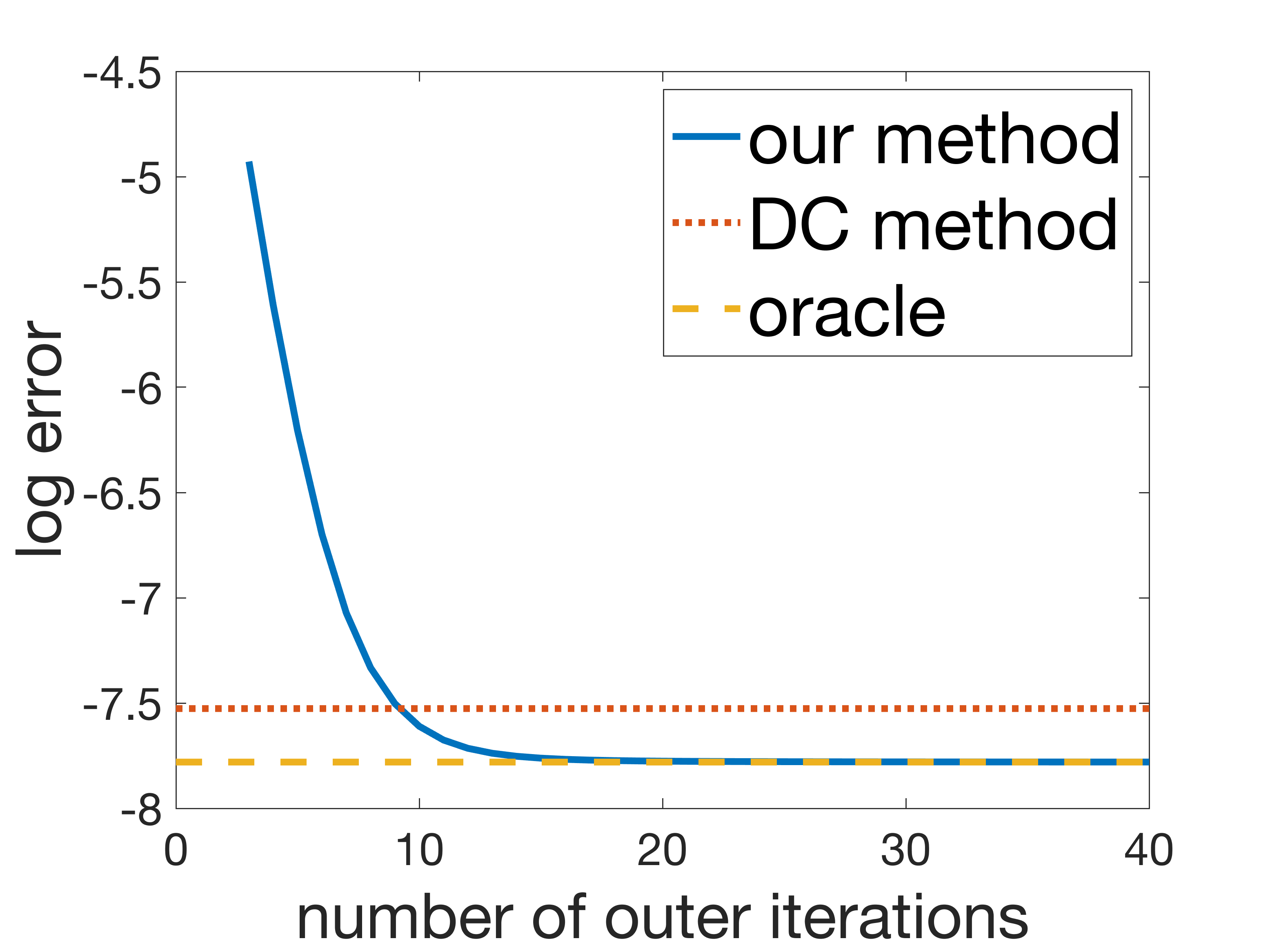}}
\subfigure[Top-$3$-dim eigenspace]{
\includegraphics[width=0.29\textwidth]{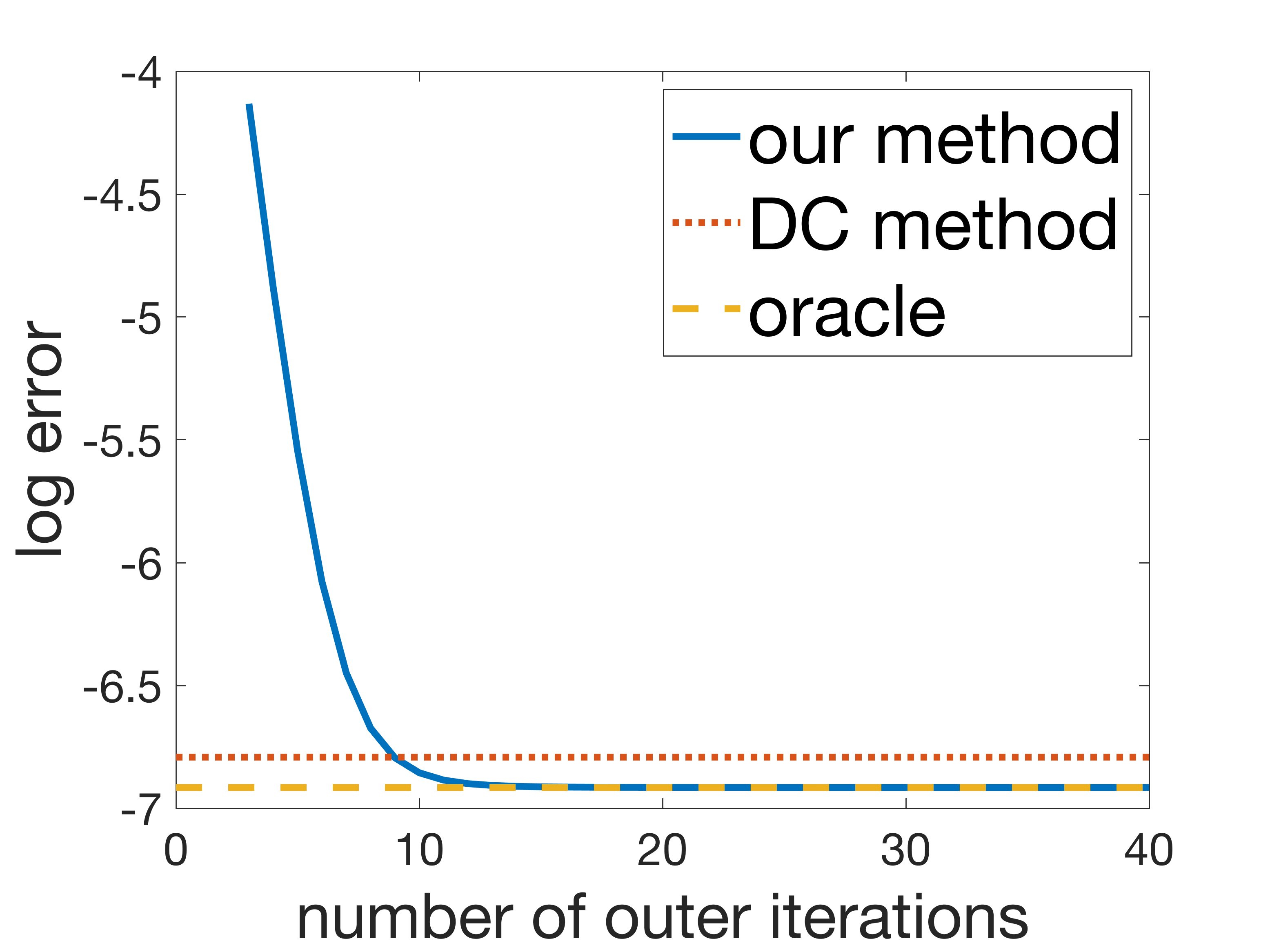}}
\caption{Comparison between algorithms when the number of outer iterations varies. The $x$-axis is the number of outer iterations and the $y$-axis is the \emph{logarithmic error}. The blue line is our error, the red line is the DC method performance and the yellow one is logarithmic error for the oracle estimator. Subfigures (a) to (c) represent the experiments with $5$ inner loops. Subfigures (d) to (f) represent the experiments with $10$ inner loops. Eigengap $\delta$ is fixed to be $1.0$.}
\label{pic:compare-outer1}
\end{figure}

\begin{figure}[!t]
\centering
\subfigure[Top-$1$-dim eigenvector]{
\includegraphics[width=0.31\textwidth]{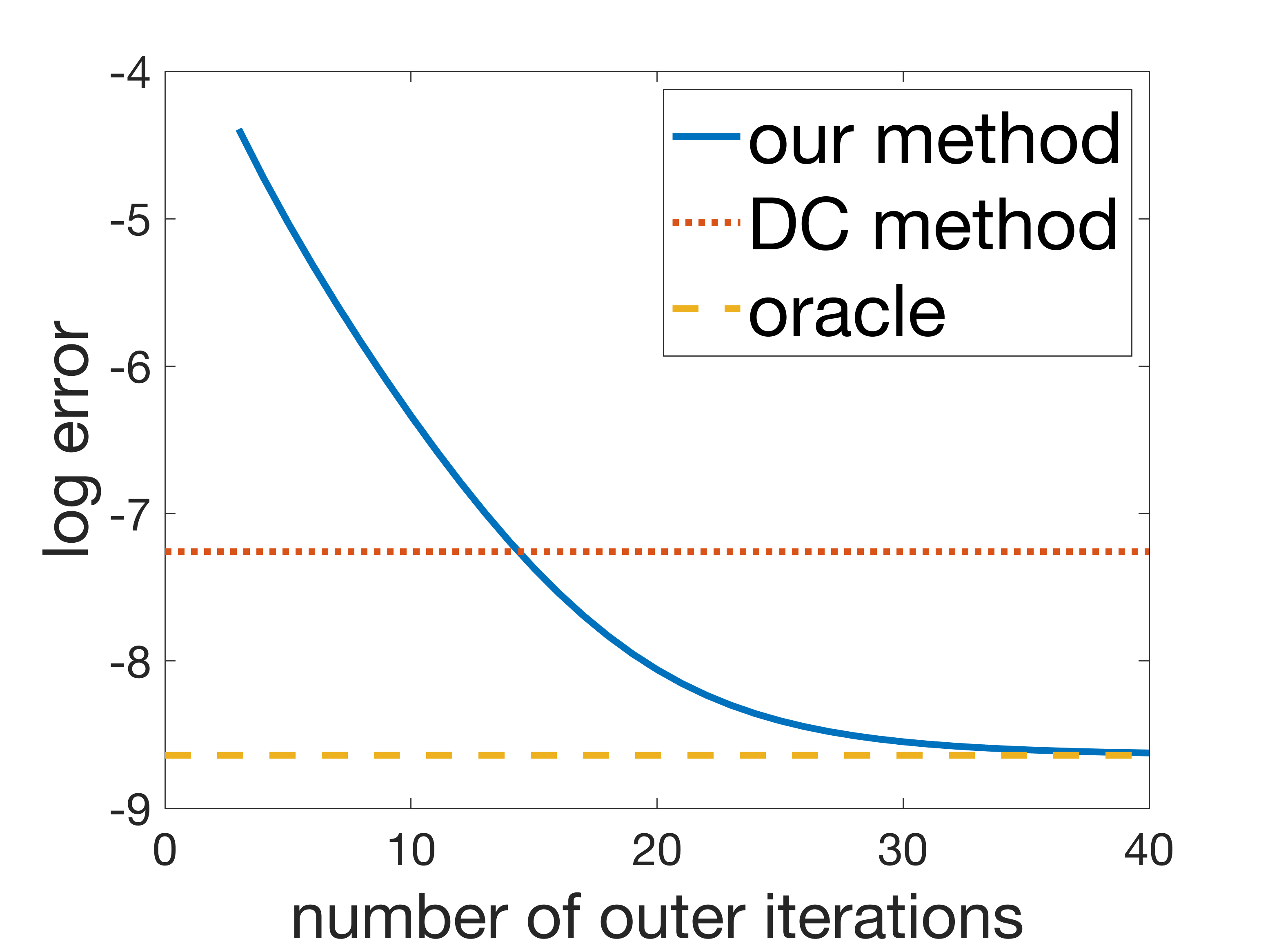}}
\subfigure[Top-$2$-dim eigenspace]{
\includegraphics[width=0.31\textwidth]{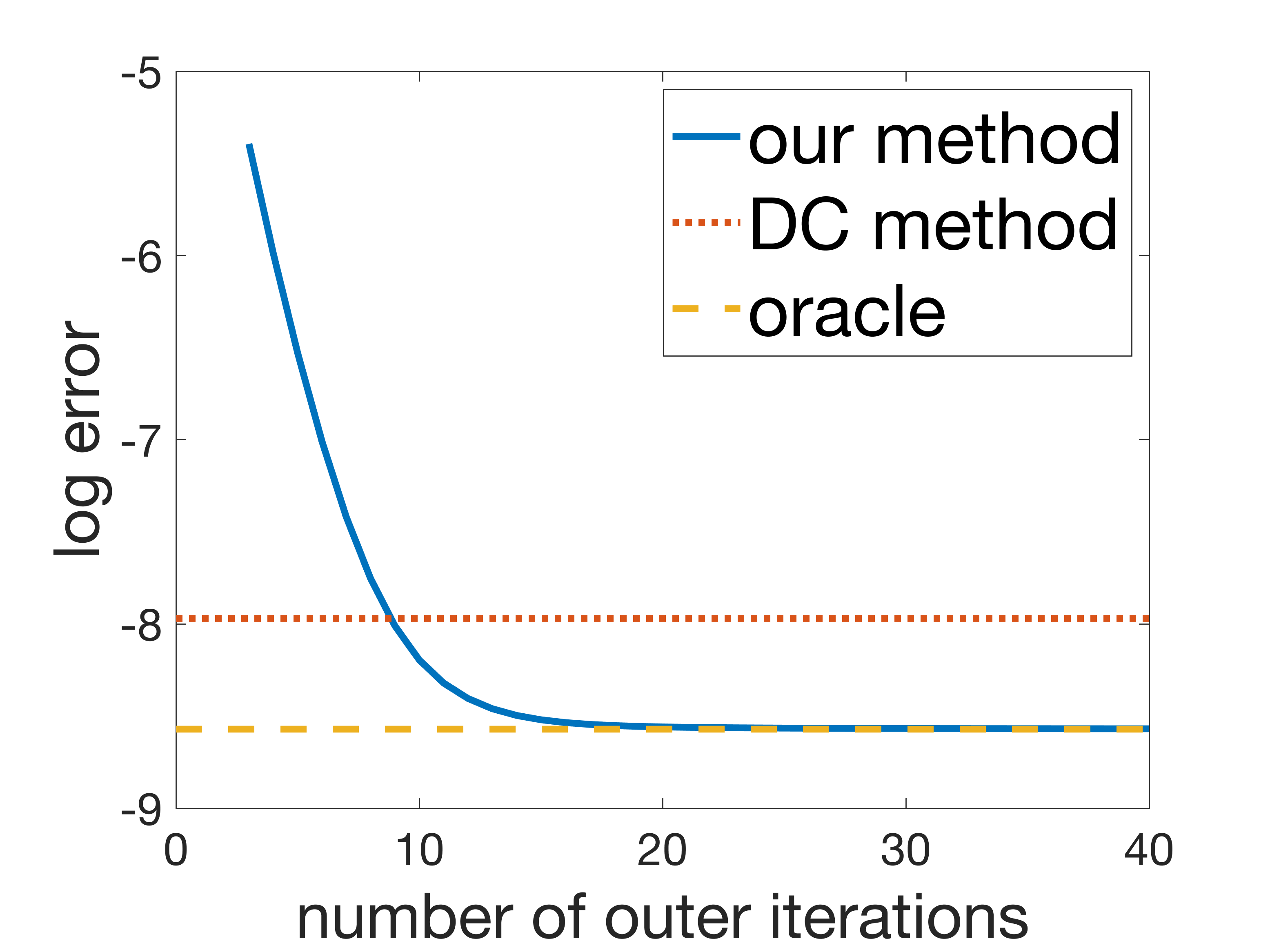}}
\subfigure[Top-$3$-dim eigenspace]{
\includegraphics[width=0.31\textwidth]{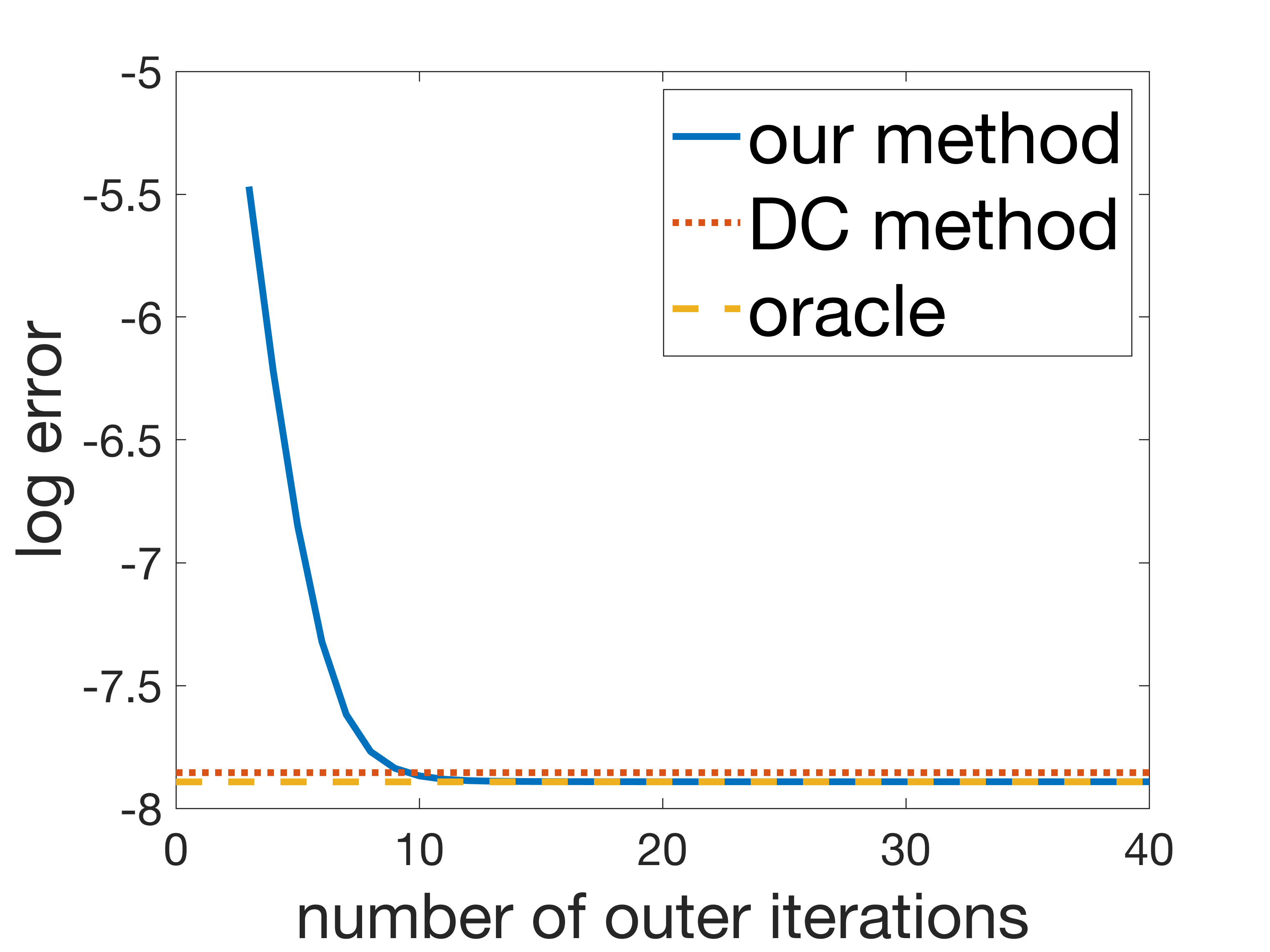}}
\subfigure[Top-$1$-dim eigenvector]{
\includegraphics[width=0.31\textwidth]{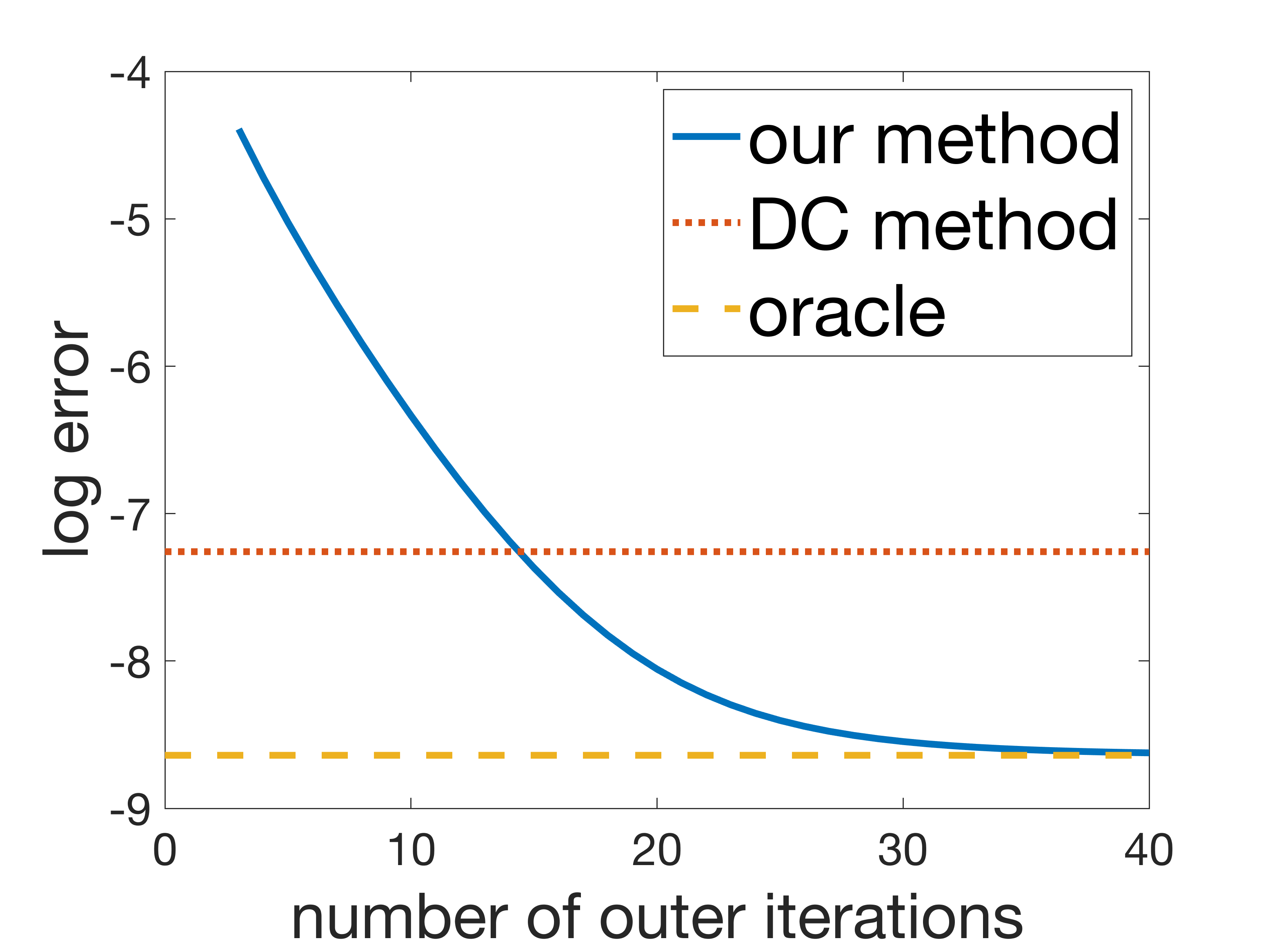}}
\subfigure[Top-$2$-dim eigenspace]{
\includegraphics[width=0.31\textwidth]{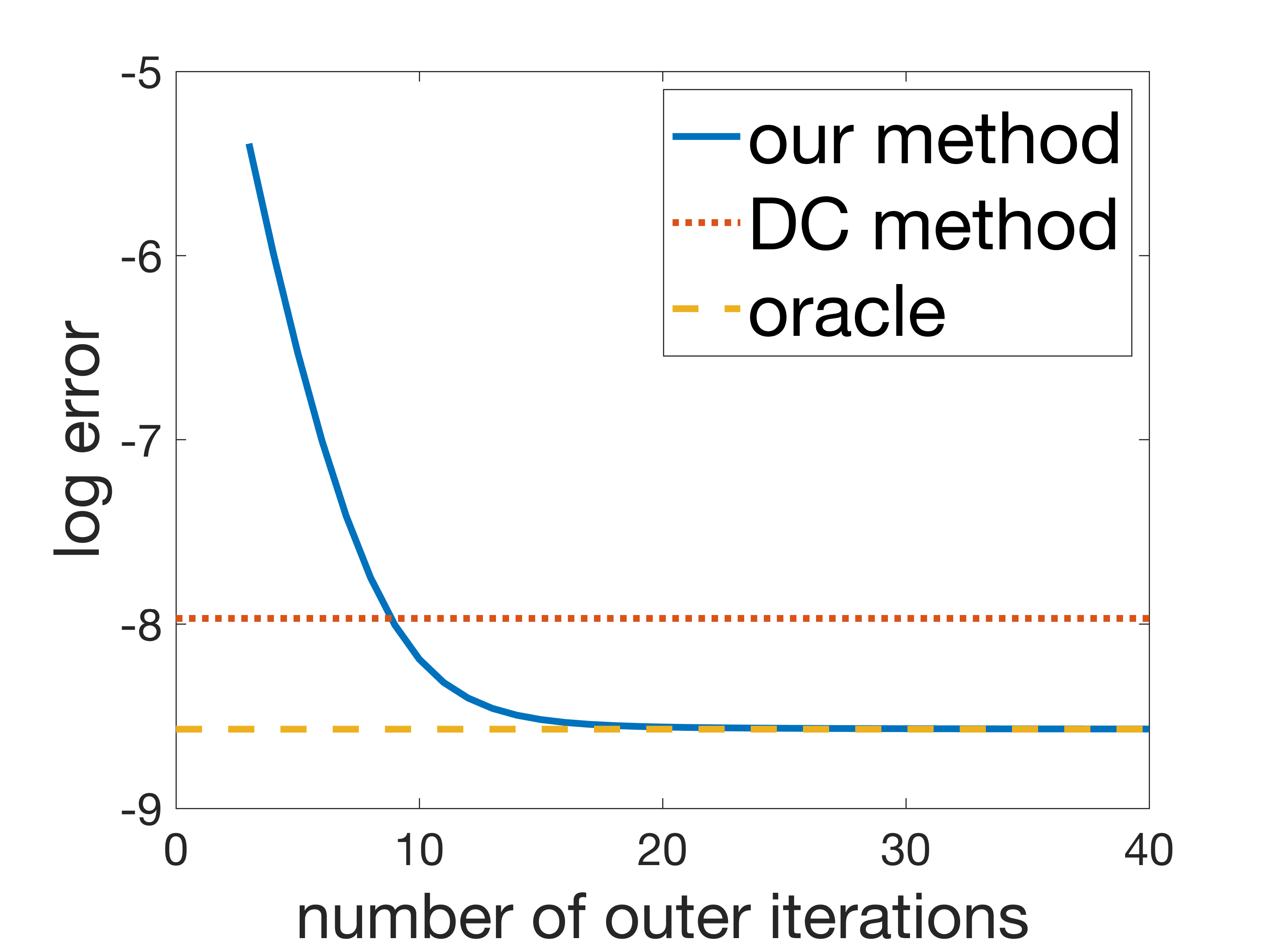}}
\subfigure[Top-$3$-dim eigenspace]{
\includegraphics[width=0.31\textwidth]{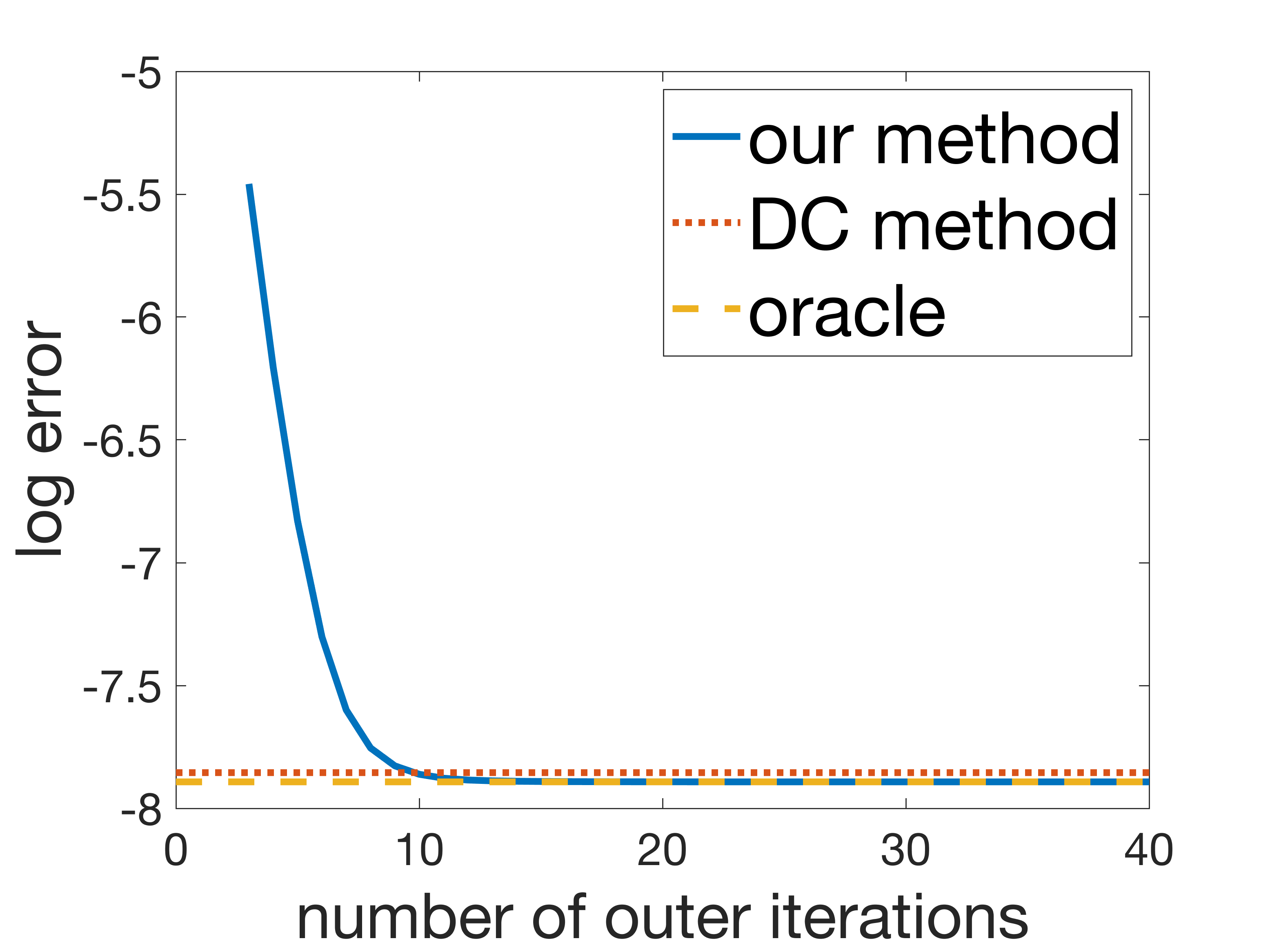}}
\caption{Comparison between algorithms when the number of outer iterations varies, under the same setting as in Figure \ref{pic:compare-outer1}. Subfigures (a) to (c) represent the experiments with $5$ inner loops. Subfigures (d) to (f) represent the experiments with $10$ inner loops. Eigengap $\delta$ is fixed to be $2.0$.}
\label{pic:compare-outer2}
\end{figure}

\subsection{Varying the number of outer iterations}
\label{subsec:exp-iter}

In this section, we present tests on how the performance of our distributed PCA changes with the number of outer iterations $T$ in Algorithm~\ref{algo:distr_top}. Consider data dimension $d$ to be $50$, sample size on each machine to be $500$, and the number of machines to be $200$, i.e., $\a \in \R^{50}$, $m = 500$ and $K = 200$.

We will report the \emph{logarithmic error}. As shown in Theorem~\ref{thm:top_eigenvector}, the logarithmic error follows an approximately linear decrease with respect to the number of outer iterations. A linear relationship between the number of outer iterations and logarithmic error verifies our theoretical findings.

We now check the performance of these three approaches (oracle one, our method and DC method) under the setting of a small eigengap. Specifically, we let eigengap $\delta$ to be $1.0$ and $2.0$. Our data is drawn independently, and $\a_i \sim \mathcal{N}(\0, \S)$ for $i = 1, \ldots, mK$. We vary the number of outer iterations $T$ to evaluate the performance.

As we fix the total sample size $n=10^5$, the errors of oracle estimator and DC estimator should be constants (illustrated by two horizontal dash lines in the graphs since they are not iterative algorithms). As shown below in Figure~\ref{pic:compare-outer1} and Figure~\ref{pic:compare-outer2}, our method converges to the oracle estimator in around $20$ iterations and outperforms the DC method. Moreover, as expected, we observe a approximately linear relation between logarithmic error and the number of outer iterations. We also observe that, empirically, setting the number of inner iterations $T' = 5$ in Algorithm~\ref{algo:distr_top} is good enough for most cases.

\subsection{Varying the eigengap}
\label{subsec:exp-eig}

In the convergence analysis of both our distributed algorithm and DC method, eigengap plays a central role in the error bound. When the eigengap between $\lambda_{L}$ and $\lambda_{L+1}$ becomes smaller, the estimation task turns to be harder and more rounds are needed for the same error. Theorem 4 in \cite{Fan2017} also shows a similar conclusion. In this part, we continue our experiment in Section~\ref{subsec:exp-iter}, and examine the relationship between estimation error and eigengap.

\begin{figure}[!t]
\centering
\subfigure[Top-$1$-dim eigenvector]{
\includegraphics[width=0.315\textwidth]{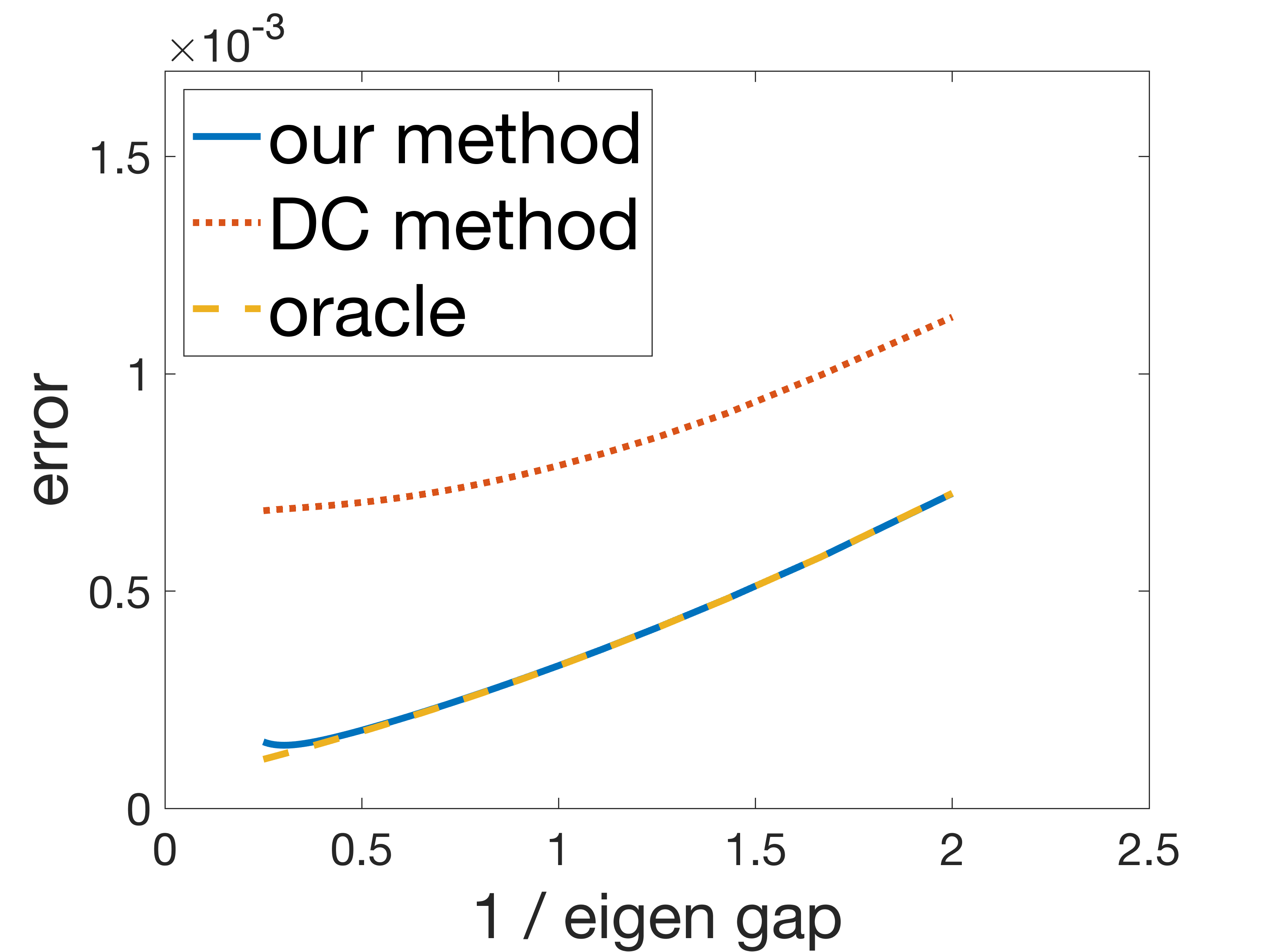}}
\subfigure[Top-$2$-dim eigenvector]{
\includegraphics[width=0.315\textwidth]{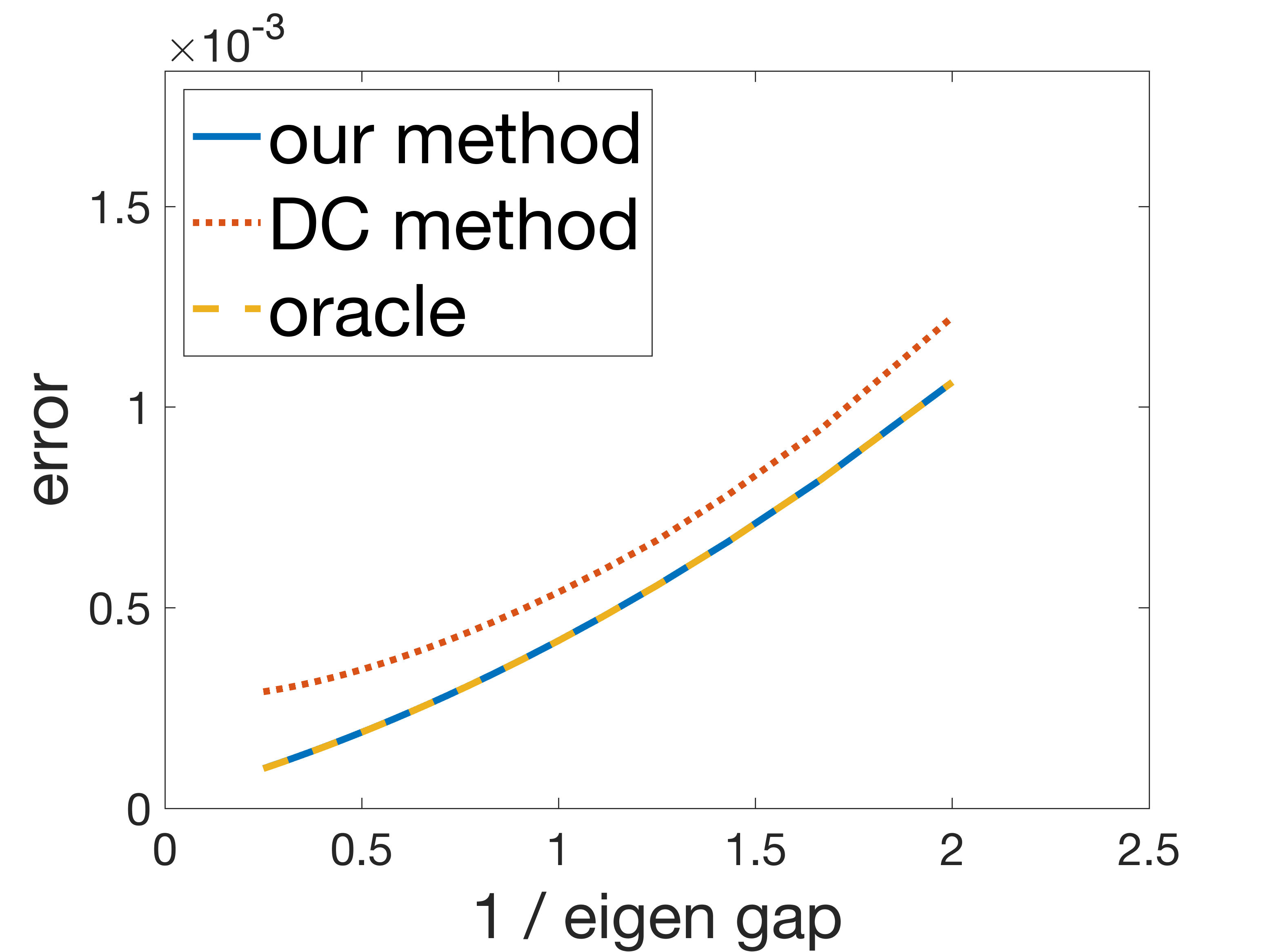}}
\subfigure[Top-$3$-dim eigenvector]{
\includegraphics[width=0.315\textwidth]{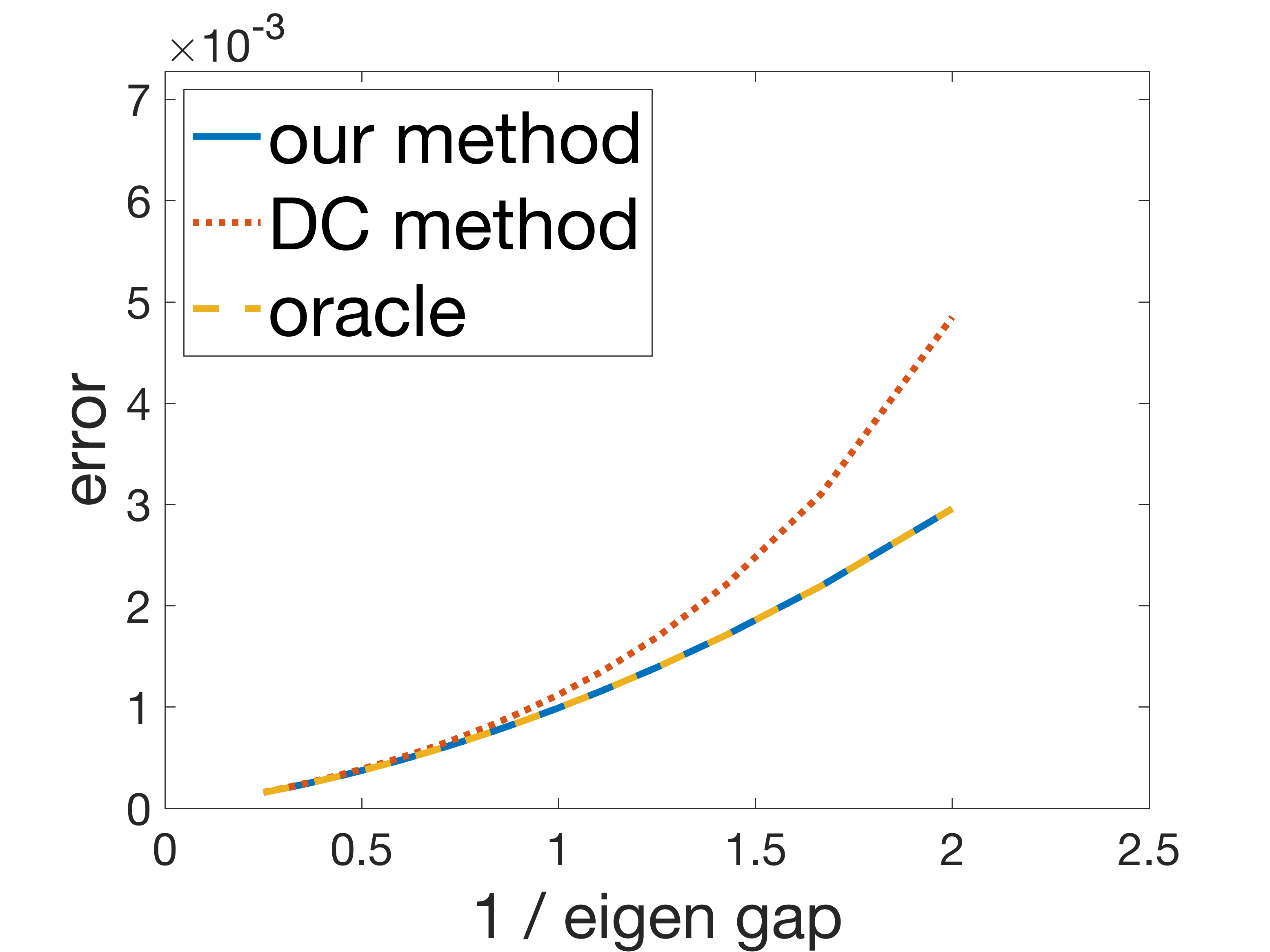}}
\caption{Comparison between algorithms when the eigengap varies. The $x$-axis is the reciprocal of eigengap and the $y$-axis is the logarithmic error.}
\label{pic:lambdalist}
\end{figure}

We fix the number of inner iterations to be $10$, and the number of outer iterations to be $40$, which, from Section~\ref{subsec:exp-iter}, is large enough for top-$3$-dim eigenspace. We still consider data dimension $d$ to be $50$, sample size on each machine to be $500$, and the number of machines to be $200$, i.e., $\a \in \R^{50}$, $m = 500$ and $K = 200$. Under this setting, we vary $\delta$ in \eqref{eq:cov-mat} and the results is shown in Figure~\ref{pic:lambdalist}. In Figure~\ref{pic:lambdalist}, the logarithmic error increases with respect to $1/\delta$, which agrees with our theoretical findings. Furthermore, our estimator has the same performance as the oracle one.

\subsection{Varying the number of machines for asymmetric innovation distributions}
\label{subsec:exp-mach}

In this section, we compare our method to the DC method by varying the number of local machines.
As mentioned in Theorem 4 in \cite{Fan2017}, DC method has a slower convergence rate (of order $\mathcal{O}(\rho \sqrt{L r / n}) + \mathcal{O}(\rho^2 \sqrt{L} r / m)$ instead of the optimal rate $\mathcal{O}(\rho \sqrt{L r / n})$) when the number of machines is greater than $\mathcal{O}\left(m / (\rho^2 r)\right)$ in the asymmetric innovation distributions (defined in Section~\ref{subsec:intro-note}) setting. Here $\rho$ is the condition number of the population covariance matrix, i.e., $\rho = \lambda_1 / (\lambda_{L} - \lambda_{L+1})$, and $r = \mathrm{Tr}(\S) / \lambda_1$ is the effective rank of $\S$.

We set data dimension $d$ to be $50$, local sample size to be $500$, i.e., $\a \in \R^{50}$, $m = 500$. We choose eigengap $\delta$ to be $0.5$, thus $\La = \mathrm{diag}(2.5, 2, 1.5, 1, ...,1)$. Here, without sticking on our Gaussian setting, we consider to use skew-distributed random variables. In particular, we generate $\a = [a_1, \ldots, a_d]^\top \in \R^d$ from beta distribution family such that for each $a_i, i = 1, \ldots, d$, we set its mean to be zero, variance to be $\La_{ii}$ and skewness to be $4$ or $6$, respectively.

\begin{figure}[!t]
	\centering
	\subfigure[Top-$1$-dim eigenvector]{
		\includegraphics[width=0.31\textwidth]{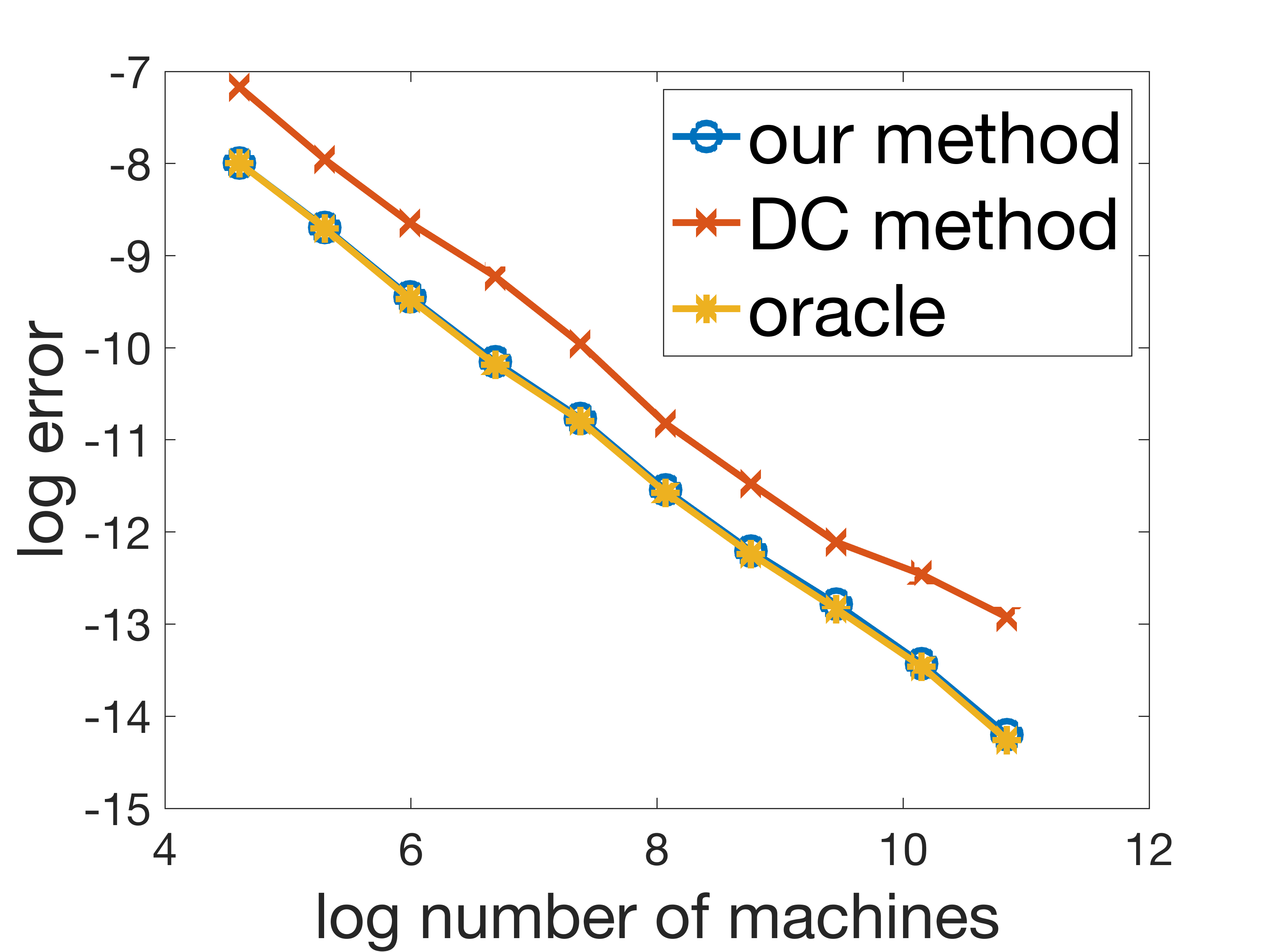}}
	\subfigure[Top-$2$-dim eigenvector]{
		\includegraphics[width=0.31\textwidth]{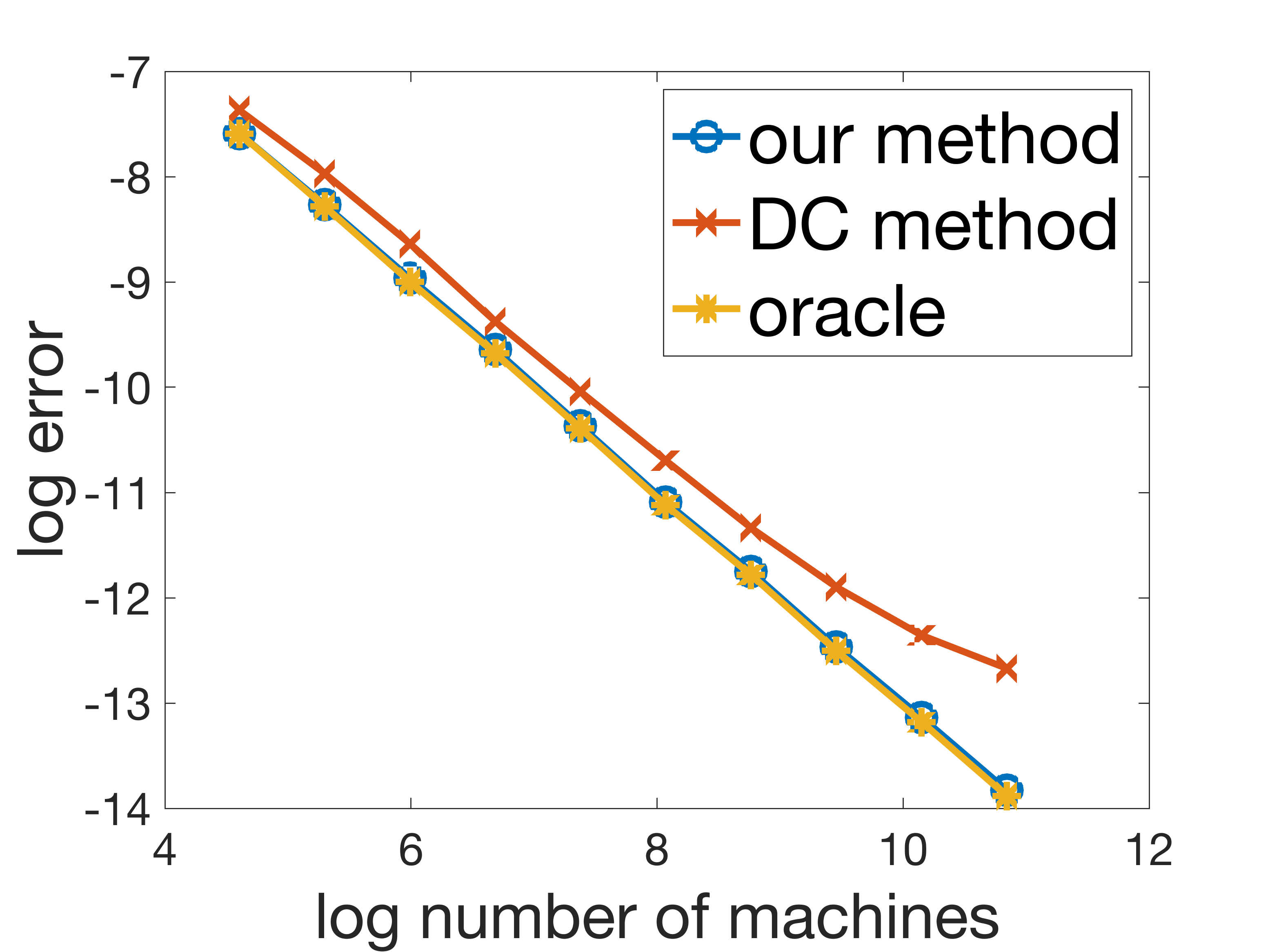}}
	\subfigure[Top-$3$-dim eigenvector]{
		\includegraphics[width=0.31\textwidth]{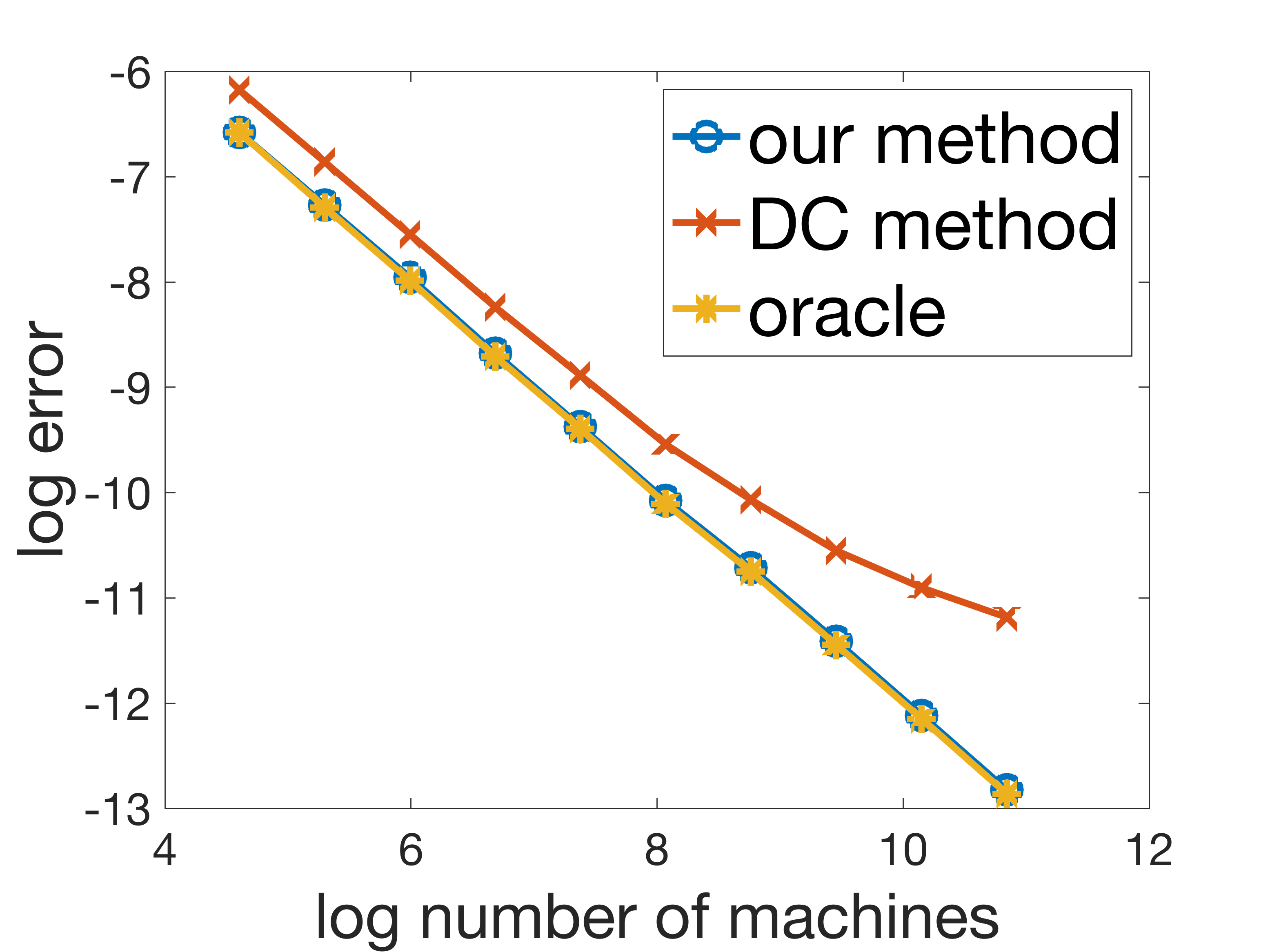}}
	\subfigure[Top-$1$-dim eigenvector]{
		\includegraphics[width=0.31\textwidth]{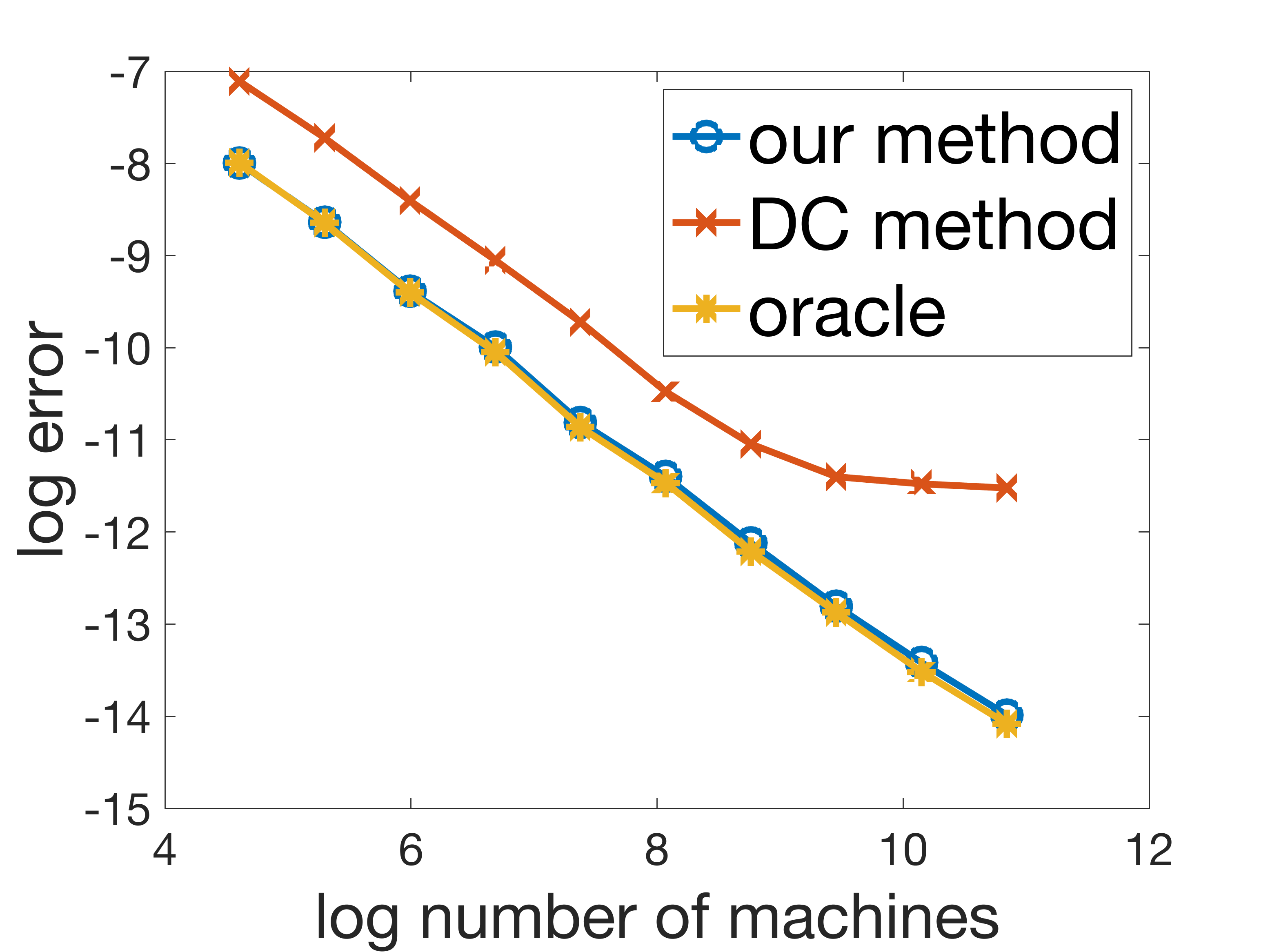}}
	\subfigure[Top-$2$-dim eigenvector]{
		\includegraphics[width=0.31\textwidth]{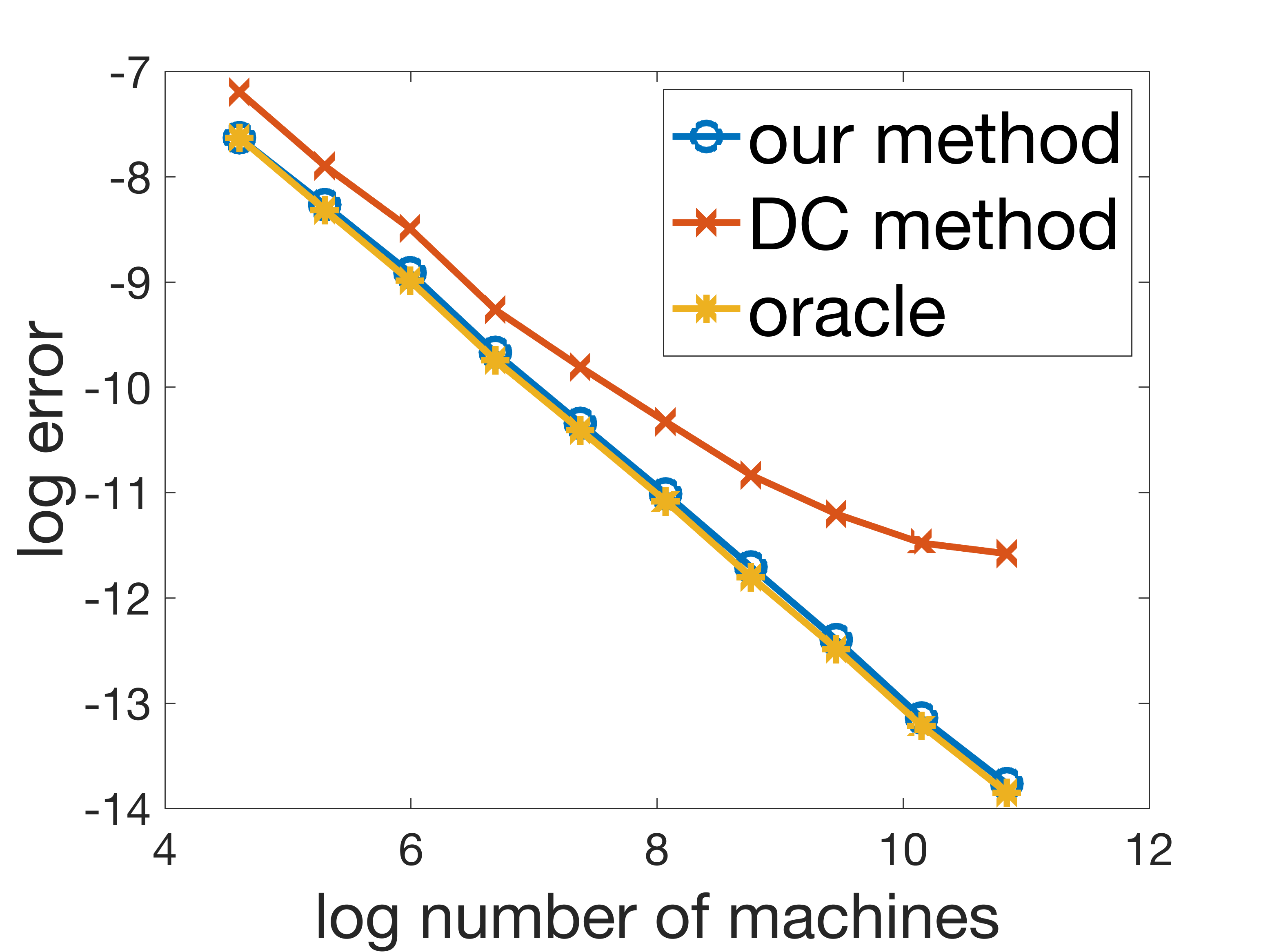}}
	\subfigure[Top-$3$-dim eigenvector]{
		\includegraphics[width=0.31\textwidth]{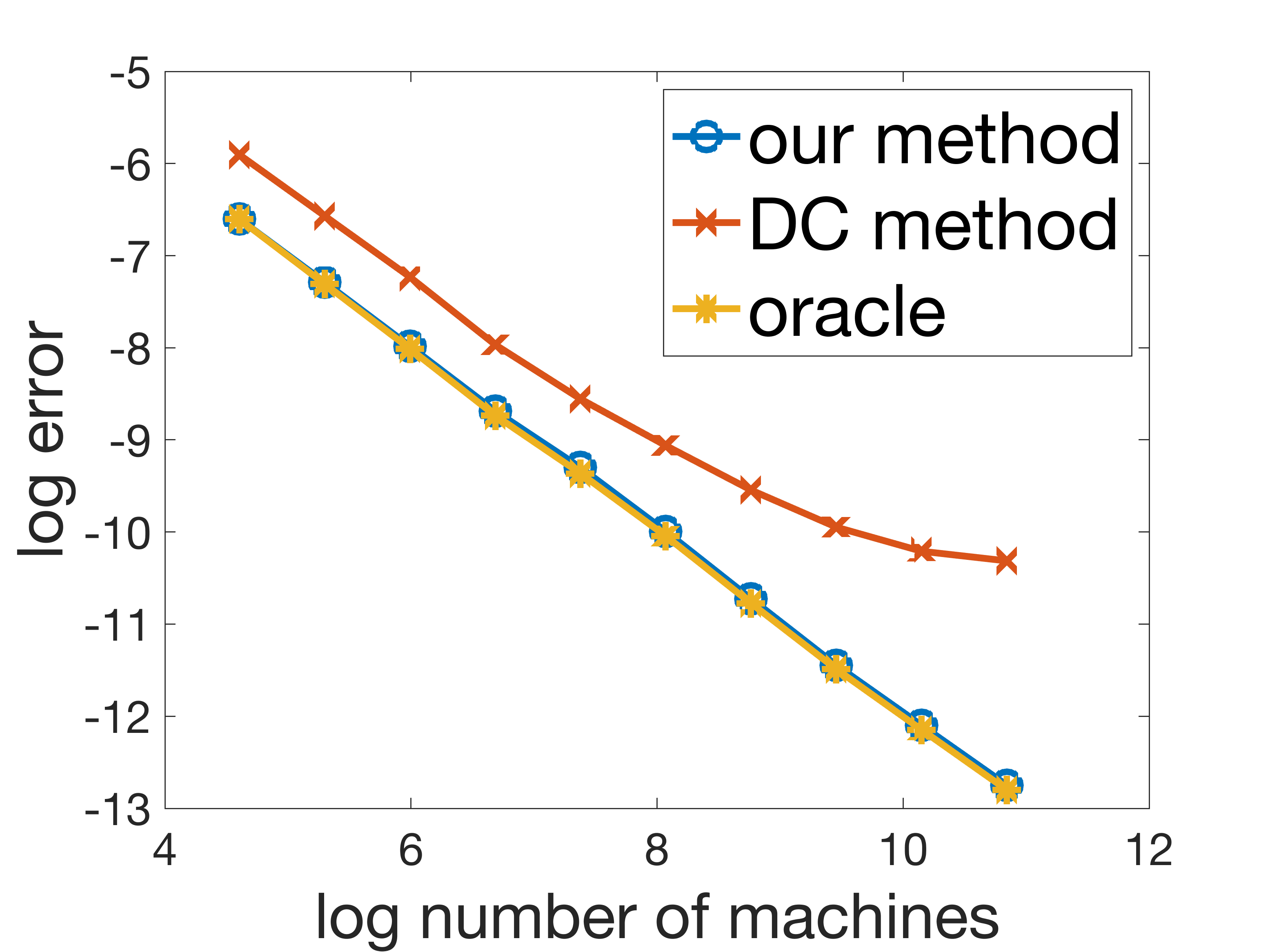}}
	\caption{Comparison between  algorithms when the number of machines varies. The $x$-axis is the log the number of machines and the $y$-axis is the logarithmic error. Subfigures (a) to (c) represent the experiments of top-$1$-dim to top-$3$-dim eigenvector estimation with skewness $4.0$ and (d) to (f), $6.0$.}
	\label{pic:compare-machine}
\end{figure}

We set the iteration parameters as in Section \ref{subsec:exp-eig} and the number of machines is varied from $100$ to $51,200$. Our results are shown in Figure~\ref{pic:compare-machine}. As can be seen from Figure~\ref{pic:compare-machine}, our method achieves the same statistical convergence rate as the oracle one. When the number of machines is small, the estimation error of the DC method also decreases at the same rate as the number of machines increases. However, the estimation error the of DC method becomes flat (or decreases at a much slower rate) when the number of machines is larger than a certain threshold. In that regime, our approach is still comparable to its oracle counterpart.

We also conduct simulation studies on principal component regression and Gaussian single index model cases and compare our approach with the oracle and the DC ones. Due to the space limitation, we defer these results to Appendix~\ref{sec:add-exp} in the supplementary material.

\section{Discussions and Future Work}

In this paper, we address the problem of distributed estimation for principal eigenspace. Our proposed multi-round method achieves fast convergence rate. Furthermore, we establish an error bound for our method from an enlarged eigenspace viewpoint, which can be seen as an extension to the traditional error bound. The insight behind our work is the combination of shift-and-invert preconditioning and convex optimization, with the adaption into distributed environment. This distributed PCA algorithm refines the divide-and-conquer scheme and removes the constraint on the number of machines from previous methods.

One important future direction is to further investigate the principal eigenspace problem under distributed settings. Specifically, computational approaches and theoretical tools can be established for other types of PCA problems, such as PCA in high dimension (see, e.g., \cite{johnstone2001distribution,fan2017asymptotics,cai2018rate}) and sparse PCA (see, e.g., \cite{johnstone2009consistency,cai2013sparse,vu2013minimax}).

\section*{Acknowledgment}

Xi Chen is supported by NSF grant IIS-1845444. Jason D. Lee is supported by NSF grant CCF-2002272. Yun Yang is supported by NSF grant DMS-1810831.
 
\setlength{\bibsep}{3pt}
\bibliographystyle{chicago}
\bibliography{contents/ref}

\clearpage

\numberwithin{equation}{section}

\begin{center}
{\LARGE\bf Supplement to Distributed Estimation for Principal Component Analysis: an Enlarged Eigenspace Analysis}
\end{center}
\medskip

The supplementary material is organized as follows:
\begin{enumerate}
    \item In Section~\ref{sec:pca-proof}, we provide proofs of our main theoretical results in Section~\ref{sec:theory}. Section~\ref{subsec:lemmas} gives some useful technical results. Proofs for distributed top eigenvector estimation and distributed top-$L$-dim eigenspace estimation are presented in Section~\ref{subsec:proof-top-1} and Section~\ref{subsec:proof-top-L}, respectively.
    \item In Section~\ref{sec:application}, we consider two application scenarios for our distributed PCA algorithm, i.e., principal component regression (PCR) and single index model (SIM). In particular, we provide settings and theoretical results in Section~\ref{subsec:setting-pcr} for PCR and in Section~\ref{subsec:setting-sim} for Gaussian SIM. Convergence rates are conducted for both single machine and distributed settings.
    \item In Section~\ref{sec:pcr-sim-proof}, proofs of theoretical results for the two applications, PCR and Gaussian SIM are given. 
    \item In Section~\ref{sec:add-exp}, we conduct some additional numerical experiments for our two applications.
\end{enumerate}

\begin{appendices}
\section{Proofs of Distributed PCA}
\label{sec:pca-proof}

\subsection{Technical lemmas}
\label{subsec:lemmas}

We start with a useful result.
\begin{lemma}\label{Lem:Qexist}
Let $A$, $B$ be two matrices with orthonormal columns such that for some $\eta>0$,
\begin{align*}
\matnorm{A^\top B^\perp}{2} \leq \eta,
\end{align*}
where $B^\perp$ denotes a matrix whose columns consist an orthonormal basis of span$(B)^\perp$. Then there exists a matrix $Q$, $\matnorm{Q}{2}\leq 1$, such that
\begin{align*}
\matnorm{A - B Q}{2} \leq \eta.
\end{align*}
\end{lemma}
\begin{proof}
Note that $\matnorm{A^\top B^\perp}{2} \leq \eta$ implies $A^\top B B^\top A^\top = I - A^\top B^\perp \succeq (1-\eta) I$. Therefore, the claimed result is a consequence by applying Proposition B.1 in~\cite{Zhu:16:LazySVD}.
\end{proof}

Next we provide a standard lemma to justify the claims that $\matnorm{\hSigma-\hSigma_1}{2} = \mathcal{O}_p\left(\sqrt{d/m}\right)$ and initial estimator conditions in Theorem~\ref{thm:top_eigenvector}.
\begin{lemma} \label{lemma:mat_concentration}
If our samples $\a_i, i = 1, \ldots, n$ are sub-Gaussian($\sigma$) vectors, then with high probability, we have,
\begin{align}
\label{eq:mat_concent}
 \matnorm{\hSigma-\hSigma_1}{2} = \mathcal{O}\left(\sqrt{\frac{d}{m}} \right).
\end{align}
The top eigenvalue on the first machine $\lambda_1(\A^\top \A / m)$ satisfies,
\begin{align}
\label{eq:eigenvalue_concent}
|\hlambda_1 - \lambda_1\left(\A^\top \A / m\right)| \leq \matnorm{\hSigma-\hSigma_1}{2}.
\end{align}
Furthermore, let $\w^{(0)}$ be the top eigenvector of $\hSigma_1$ on the first machine. We have the following gap-free concentration bound for $\w^{(0)}$ and $\hu_1, \ldots, \hu_d$,
\begin{align}
\label{eq:eigenvector_concent}
\sum_{l:\, \hlambda_l \leq (1-\delta)\,\hlambda_1} \left| \langle \hu_l, \w^{(0)} \rangle\right|^2 \leq \frac{\matnorm{\hSigma-\hSigma_1}{2}}{\delta \hlambda_1}.
\end{align}
\end{lemma}

\begin{proof}
By Corollary 5.50 in \cite{Vershynin2010}, with probability at least $1 - 2 e^{- d \sigma^{2}/ C}$,
\begin{align*}
\matnorm{\hSigma - \S}{2} \leq \matnorm{\hSigma - \S}{\mathrm{F}} \leq \sqrt{\frac{d}{n}} \sigma ,
\end{align*}
and
\begin{align*}
 \matnorm{\hSigma_1 - \S}{2} \leq \matnorm{\hSigma_1 - \S}{\mathrm{F}} \leq \sqrt{\frac{d}{m}} \sigma,
\end{align*}
where $C$ is a constant which only depends on the sub-Gaussian norm of the random vector $\a$.
Therefore, our first inequality~\eqref{eq:mat_concent} is a direct result of above matrix concentrations as well as triangle inequality for matrix spectral norm.

Denote $\hu_1$ and $\w^{(0)}$ to be the top eigenvector for $\hSigma$ and $\hSigma_1$, without loss of generality, we can assume $\hlambda_1 > \lambda_1(\A_1^\top \A_1 / m)$, then we have
\begin{align*}
\left| \hlambda_1 - \lambda_1(\A_1^\top \A_1 / m) \right| = \hu_1^\top \hSigma \hu_1 -  \w^{(0)\top} \hSigma_1 \w^{(0)} \leq \hu_1^\top \hSigma \hu_1 -  \hu_1^\top \hSigma_1 \hu_1 \leq \matnorm{\hSigma - \hSigma_1}{2}.
\end{align*}
With Davis-Kahan Theorem \citep{Yu2014}, it is easy to see,
\begin{align*}
\sum_{l:\, \hlambda_l \leq (1-\delta)\,\hlambda_1} \left| \langle \hu_l, \w^{(0)} \rangle\right|^2 \leq \frac{\matnorm{\hSigma-\hSigma_1}{2}}{\delta \hlambda_1}.
\end{align*}
\end{proof}

\subsection{Proofs of distributed top eigenvector estimation}
\label{subsec:proof-top-1}

\subsubsection*{Proof of Lemma~\ref{lemma:outer_loop}}
\begin{proof}
Write $\v=\H^{-1} \w+\e$, where $\|\e\|_2\leq \varepsilon$.
Since $\H^{-1}=(\overline \lambda_1\I - \hSigma)^{-1}=\sum_{l=1}^d(\overline{\lambda}_1-\hlambda_l)^{-1}\, \hu_l\hu_l^\top$, we have, for each $l=1,\ldots,d$,
\begin{align}\label{eq:Eigenvalues}
\langle \hu_l, \v\rangle = (\overline{\lambda}_1-\hlambda_l)^{-1}\, \langle \hu_l, \w\rangle + \langle \hu_l, \e\rangle.
\end{align}
This can yield a lower bound on $\| \v\|_2$,
\begin{align}
\label{eq:v-bound}
\|\v\|_2^2 &= \sum_{l=1}^d | \langle \hu_l, \v \rangle |^2 \geq \sum_{l:\, \hlambda_l > (1-\delta)\,\hlambda_1} | \langle \hu_l, \v \rangle |^2 \notag \\
&\geq \sum_{l:\, \hlambda_l > (1-\delta)\,\hlambda_1} \left[ \frac{1}{2} (\overline{\lambda}_1-\hlambda_l)^{-2} | \langle \hu_l, \w \rangle |^2 - |\langle \hu_l, \e\rangle|^2 \right] \notag \\
&\geq (32 \eta^2)^{-1} - \| \e \|_2^2 \geq (32 \eta^2)^{-1} - \epsilon^2 \geq (64 \eta^2)^{-1},
\end{align}
where we used the upper bound of $\overline{\lambda}_1-\hlambda_l$, and the conditions \eqref{eq:assumption_out1} and \eqref{eq:assumption_out2} in Lemma~\ref{lemma:outer_loop}. On the other hand, for each $l$ such that $\hlambda_l \leq (1-\delta)\,\hlambda_1$, we have $\overline\lambda_1 -\hlambda_l\geq \hlambda_1 - \hlambda_l \geq \delta\,\hlambda_1$. Consequently, equation~\eqref{eq:Eigenvalues} implies
\begin{align*}
|\langle \hu_l, \v\rangle| \leq (\delta\,\hlambda_1)^{-1} \, |\langle \hu_l, \w\rangle| +  | \langle \hu_l, \e\rangle| \leq  (\delta\,\hlambda_1)^{-1} \, |\langle \hu_l, \w\rangle| +  \varepsilon.
\end{align*}
A combination of the last two displays yields the first claimed bound. Similarly, the second claim bound follows by combining inequality~\eqref{eq:v-bound} with
\begin{align*}
\sum_{l:\, \hlambda_l \leq (1-\delta)\,\hlambda_1} |\langle \hu_l, \v\rangle|^2 &\leq 2\,(\delta\,\hlambda_1)^{-2} \sum_{l:\, \hlambda_l \leq (1-\delta)\,\hlambda_1} |\langle \hu_l, \w\rangle|^2 + 2\sum_{l:\, \hlambda_l \leq (1-\delta)\,\hlambda_1}   \langle \hu_l, \e\rangle^2\\
&\leq 2\,(\delta\,\hlambda_1)^{-2} \sum_{l:\, \hlambda_l \leq (1-\delta)\,\hlambda_1} |\langle \hu_l, \w\rangle|^2 +2\,\varepsilon^2.
\end{align*}
\end{proof}

\subsubsection*{Proof of Lemma~\ref{lemma:inner_loop}}
\begin{proof}
It is easy to verify the following two identities:
\begin{align*}
\tw^{(t+1)} &= \w_{j}^{(t+1)}  - \H^{-1}(\H\w_{j}^{(t+1)}  - \w^{(t)}),\\
\w_{j+1}^{(t+1)} &= \w_{j}^{(t+1)}  - \H_1^{-1}(\H\w_{j}^{(t+1)}  - \w^{(t)}).
\end{align*}
By taking the difference we obtain
\begin{align*}
\|\w_{j+1}^{(t+1)} - \tw^{(t+1)}\|_2 &= \| (\H^{-1}- \H_1^{-1}) (\H\w_{j}^{(t+1)}  - \w^{(t)})\|_2\\
&= \| (\I- \H_1^{-1}\H) (\w_{j}^{(t+1)}  -\H^{-1} \w^{(t)})\|_2\\
&\leq \matnorm{ \I- \H_1^{-1}\H}{2}\, \|\w_{j}^{(t+1)}  -\tw^{(t+1)}\|_2.
\end{align*}
We bound the first factor on the r.h.s.~as
\begin{align*}
\matnorm{ \I- \H_1^{-1}\H}{2} &\leq  \matnorm{\H_1^{-1}}{2}\,\matnorm{ \H_1-\H}{2}\leq \frac{2}{\eta} \,\matnorm{\hSigma-\hSigma_1}{2},
\end{align*}
where the last step follows from the fact $ \H_1-\H =  -(\hSigma-\hSigma_1)$, and the inequality
\begin{align*}
\matnorm{\H_1^{-1}}{2}^{-1}&=\lambda_{\min}(\H_1)=\lambda_{\min}\big(\overline\lambda_1\I - \hSigma + (\hSigma-\hSigma_1)\big)\geq \lambda_{\min}(\overline\lambda_1\I - \hSigma)-\matnorm{\hSigma-\hSigma_1}{2} \\
&\geq \eta-\matnorm{\hSigma-\hSigma_1}{2}  \geq \eta/2,
\end{align*}
where $\lambda_{\min}(\A)$ denotes the smallest singular value of symmetric matrix $\A$.

\end{proof}

\subsubsection*{Proof of Theorem~\ref{thm:top_eigenvector}}
\begin{proof}
Applying Lemma~\ref{lemma:inner_loop}, we have (recall $\w_{0}^{t+1} =\w_{T'}^{t}$)
\begin{align*}
\|\w_{T'}^{t+1} - \H^{-1}\w_{T'}^{t}\|_2 \leq \Big(\frac{2\kappa}{\eta}\Big)^{T'} \, \|\w_{T'}^{t}- \H^{-1}\w_{T'}^{t}\|_2\leq \Big(\frac{2\kappa}{\eta}\Big)^{T'} \frac{2}{\eta}:\,=\varepsilon_{T'},
\end{align*}
where we used the fact that $\matnorm{\I-\H^{-1}}{2} \leq 1 + \matnorm{\H^{-1}}{2} =1+(\overline\lambda_1-\hlambda_1)^{-1}\leq 2/\eta$, and $\|\w_{T'}^t\|_2=1$.
Now we can recursively apply inequality~\eqref{eq:con_outer_2} with $\varepsilon\leftarrow \varepsilon_{T'}$ to obtain
\begin{align*}
&\sum_{l:\, \hlambda_l \leq (1-\delta)\,\hlambda_1}|\langle \hu_l,\w^{(t)}\rangle|^2 \leq \frac{128\eta^2}{\delta^2\hlambda_1^2} \sum_{l:\, \hlambda_l \leq (1-\delta)\,\hlambda_1}|\langle \hu_l,\w^{(t-1)}\rangle|^2 +128\eta^2 \varepsilon^2_{T'}\\
&\qquad \leq \Big(\frac{128\eta^2}{\delta^2\hlambda_1^2}\Big)^2 \sum_{l:\, \hlambda_l \leq (1-\delta)\,\hlambda_1}|\langle \hu_l,\w^{(t-2)}\rangle|^2 + \Big(1 + \frac{128\eta^2}{\delta^2\hlambda_1^2} \Big)\,128\eta^2\varepsilon^2_{T'}\\
&\qquad \leq\ldots  \leq \Big(\frac{128\eta^2}{\delta^2\hlambda_1^2}\Big)^{T} \sum_{l:\, \hlambda_l \leq (1-\delta)\,\hlambda_1}|\langle \hu_l,\w^{(0)}\rangle|^2 +
128\eta^2\varepsilon^2_{T'}\,\sum_{t=0}^T \Big(\frac{128\eta^2}{\delta^2\hlambda_1^2}\Big)^{t} \\
 &\qquad \leq \Big(\frac{128\eta^2}{\delta^2\hlambda_1^2}\Big)^{T}+\frac{128\eta^2\varepsilon_{T'}^2}{1-128\eta^2/(\delta\hlambda_1)^2}.
\end{align*}
\end{proof}

\subsection{Proofs of distributed top-$L$-dim principal subspace estimation}
\label{subsec:proof-top-L}

\subsubsection*{Proof of Theorem~\ref{thm:top_L_eigenvectors}}
\begin{proof}
Our proof adapts the proof of Theorem 4.1(a) in \cite{Zhu:16:LazySVD} to our settings.

Let $\hm = \matnorm{\hSigma^{(L-1)}}{2}$. Due to the Courant minimax principle, we have $\mu \geq \hlambda_L$.

Note that column vectors in $\V_l$ are already eigenvectors of $\hSigma^{(l)}$ with eigenvalues zero. Let $\W_l$ be column orthogonal matrix whose columns are eigenvectors in $\V_l^\perp$ of $\hSigma^{(l)}$ with eigenvalues in the range $[0,\,(1-\delta+\tau_l)\,\hm]$, where $\tau_l=\frac{l}{2L}{\delta}$, for $l=0,1\ldots,L$.

We will show that for each $l=0,1,\ldots,L$, there exists a matrix $\Q_l$ such that
\begin{align}\label{eq:induction}
\matnorm{\hU_{\leq (1-\delta) \hlambda_L} - \W_l \Q_l}{2} \leq \varepsilon_l\in[0,1) \quad\mbox{and}\quad\matnorm{\Q_l}{2}\leq 1,
\end{align}
for some sequence $\{\varepsilon_l\}_{l=1}^L$ of small numbers. This would imply our claimed bound. In fact, the first inequality in the preceding display implies $\matnorm{\I- \hU_{\leq (1-\delta) \hlambda_L}^\top \W_L \Q_L}{2}\leq \varepsilon_l$. Therefore, the smallest singular value of $\hU_{\leq (1-\delta) \hlambda_L}^\top \W_L \Q_L$ is at least $1-\varepsilon_L>0$. This lower bound combined with $\matnorm{\Q_L}{2}\leq1$ implies the smallest singular value of $\hU_{\leq (1-\delta) \hlambda_L}^\top \W_L$ to be at least $1-\varepsilon_L$, or
\begin{align*}
\I - \hU_{\leq (1-\delta) \hlambda_L}^\top \W_L\W_L^\top \hU_{\leq (1-\delta) \hlambda_L} \preceq 1 - (1-\varepsilon_L)^2\I.
\end{align*}
Now since $\V_L$ is (column) orthogonal to $\W_L$, we obtain
\begin{align*}
\hU_{\leq (1-\delta) \hlambda_L}^\top \V_l\V_l^\top\hU_{\leq (1-\delta) \hlambda_L} \preceq \hU_{\leq (1-\delta) \hlambda_L}^\top (\I-\W_L\W_L^\top)\hU_{\leq (1-\delta) \hlambda_L} \preceq 2\varepsilon_L\I,
\end{align*}
which implies $\matnorm{\hU_{\leq (1-\delta) \hlambda_L}^\top \V_L}{2} \leq \sqrt{2\varepsilon_L}$.

When $l=0$, we simply choose $\W_0=\hU_{\leq (1-\delta) \hlambda_L}$, $\varepsilon_0=0$ and $\Q_0=\I$
Suppose for every $l\in\{0,\ldots,L-1\}$, there exists a matrix $Q_l$ with $\matnorm{Q_l}{2}\leq1$ satisfying $\matnorm{\hU_{\leq (1-\delta) \hlambda_L} - \W_l \Q_l}{2} \leq \varepsilon_l$ for some $\varepsilon_l\in[0,1)$. Now we construct $\Q_{l+1}$ as follows.

Since $\kappa = \matnorm{\hSigma-\hSigma_1}{2}  \geq \matnorm{(\I-\V_l\V_l^\top) (\hSigma-\hSigma_1)(\I-\V_l\V_l^\top)}{2}$, we can apply Theorem~\ref{thm:top_eigenvector} to $\w_{l+1}$ with $\delta \leftarrow \delta/2$ to obtain (note that columns of $\W_l$ and $\V_l$ corresponds to eigenvectors of $\hSigma^{(l)}$ with eigenvectors less than or equal to $(1-\delta+\tau_l)\matnorm{\hSigma^{(L-1)}}{2}\leq (1-\delta/2)\matnorm{\hSigma^{(l)}}{2}$)
\begin{align*}
\|\W_l^\top \w_{l+1}\|_2^2 \leq \varepsilon_{T,T'}^{(l)} \quad\mbox{and}\quad \|\V_l^\top \w_{l+1}\|_2^2 \leq \varepsilon_{T,T'}^{(l)},
\end{align*}
where $\varepsilon_{T,T'}^{(l)}:\,=\Big(\frac{128\eta^2}{\delta^2 \lambda_1(\hSigma^{(l)})^2}\Big)^T + \frac{512\,\eta}{1-128\eta^2/(\delta \lambda_1(\hSigma^{(l)}))^2} \, \Big(\frac{4\kappa^2}{\eta^2}\Big)^{T'}\leq 1/2$.
Since $\v_{l+1}$ is the projection of $\w_{l+1}$ into $\V_{l}^\perp$, we have
\begin{align*}
\|\W_l^\top\v_{l+1}\|_2^2 \leq \frac{\|\W_l^\top\w_{l+1}\|_2^2}{\|(\I-\V_l\V_l^\top)\w_{l+1}\|_2^2}\leq \frac{\varepsilon_{T,T'}^{(l)}}{1-\varepsilon_{T,T'}^{(l)}}\leq \frac{3}{2}\varepsilon_{T,T'}^{(l)}.
\end{align*}

We will make use of the following lemma from \cite{Zhu:16:LazySVD} (Lemma B.4).

\begin{lemma}[Eigen-space perturbation lemma]\label{lemma:eigenspace}
Let $\M\in\R^{d\times d}$ be a positive semidefinite matrix with eigenvalues $\lambda_1\geq\ldots\geq\lambda_{r}\geq\lambda_{r+1}=\cdots=\lambda_d=0$, and corresponding eigenvectors $\u_1,\ldots,\u_d$. Define  $\U=[\u_{j+1},\ldots,\u_r]\in\R^{d\times(r-j)}$ to be the matrix composing of all top $r$-eigenvectors with eigenvalues less than or equal to $\mu$. Let $\v\in\R^d$ be a unit vector such that $\|\U^\top\v\|_2\leq \varepsilon\leq 1/2$, $\v^\top\u_{r+1}=\cdots=\v^\top \u_d=0$. Define
\begin{align*}
\tM = (\I-\v\v^\top)\,\M\,(\I-\v\v^\top).
\end{align*}
Denote $[\tV,\tU,\v,\u_{r+1},\ldots,\u_d]\in\R^{d\times d}$ as the orthogonal matrix composed of eigenvectors of $\tM$, where $\tU$ consists of all eigenvectors (other than $\v,\u_{r+1},\ldots,\u_d$) with eigenvalues
 less than or equal to $\mu+\tau$. Then there exists a matrix $\Q$ such that $\matnorm{\Q}{2}\leq 1$ and
 \begin{align*}
 \matnorm{\U - \tU\Q}{2} \leq \sqrt{\frac{169\lambda_1^2\varepsilon^2}{\tau^2}+\varepsilon^2}.
 \end{align*}
\end{lemma}

By applying Lemma~\ref{lemma:eigenspace} with $\M = \hSigma^{(l)}$,\\ $\tM=\hSigma^{(l+1)}=(\I- \v_{l+1}\v_{l+1}^\top)\hSigma^{(l)}(\I- \v_{l+1}\v_{l+1}^\top)$, $r= d-l$, $\V=\W_l$, $\tV=\W_{l+1}$, $\v=\v_{l+1}$, $\mu=(1-\delta+\tau_l)\,\hm$, $\tau=(\tau_{l+1}-\tau_l)\,\hm$, we obtain a matrix $\tQ_{l}$ such that $\matnorm{\tQ_{l}}{2}\leq1$ and
\begin{align*}
\matnorm{\W_l-\W_{l+1}\tQ_l}{2}\leq \sqrt{\frac{507\hlambda_1^2\varepsilon_{T,T'}^{(l)}}{2(\tau_{s+1}-\tau_s)^2 \hm^2}+\frac{3}{2}\varepsilon_{T,T'}^{(l)}}\leq \frac{32\hlambda_1 L\sqrt{\varepsilon_{T,T'}^{(l)}}}{\hlambda_L\delta}.
\end{align*}
This inequality combined with inequality~\eqref{eq:induction} together implies
\begin{align*}
\matnorm{\W_{l+1} \tQ_l \Q_l - \hU_{\leq (1-\delta) \hlambda_L}}{2}& \leq \matnorm{\W_{l+1} \tQ_l \Q_l - \W_l\Q_l}{2}
+ \matnorm{\W_l\Q_l -\hU_{\leq (1-\delta) \hlambda_L}}{2}\\& \leq \varepsilon_l+\frac{32\hlambda_1 L\sqrt{\varepsilon_{T,T'}^{(l)}}}{\hlambda_L\delta}.
\end{align*}
By defining $\Q_{l+1} = \tQ_l \Q_l$, we obtain
\begin{align*}
\matnorm{\hU_{\leq (1-\delta) \hlambda_L}-\W_{l+1} \Q_{l+1}}{2} \leq \varepsilon_{l+1} :\,&=\varepsilon_l+\frac{32\hlambda_1 L\sqrt{\varepsilon_{T,T'}^{(l)}}}{\hlambda_L\delta}=\cdots \\
&=\sum_{k=0}^{l+1} \,\frac{32\hlambda_1 L\sqrt{\varepsilon_{T,T'}^{(k)}}}{\hlambda_L\delta}\\
&\leq (l+1) \,\frac{32\hlambda_1 L\sqrt{\varepsilon_{T,T'}^{(l)}}}{\hlambda_L\delta} \\
&\leq (l+1) \,\frac{32\hlambda_1 L}{\hlambda_L\delta} \sqrt{\Big(\frac{128\eta^2}{\delta^2 \hlambda_k^2}\Big)^T + \frac{512\,\eta}{1-128\eta^2/(\delta \hlambda_k)^2} \, \Big(\frac{4\kappa^2}{\eta^2}\Big)^{T'}},
\end{align*}
for $l=0,1,\ldots,L$.

\end{proof}

\subsubsection*{Proof of Corollary~\ref{corr:population-gap-free2}}

Please refer to Theorem 4.1 in ~\cite{Zhu:16:LazySVD}.

\subsubsection*{Proof of Corollary~\ref{corr:population-gap-free}}

\begin{proof} Notice that $\hU_{> (1-\delta) \hlambda_L} \hU_{> (1-\delta) \hlambda_L}^\top + \hU_{\leq (1-\delta) \hlambda_L} \hU^\top_{\leq (1-\delta) \hlambda_L} = \I_d$, we have
\begin{align*}
\matnorm{\V_L^\top \U_{\leq (1-2\delta) \lambda_L}}{2}
&= \matnorm{\V_L^\top \left(\hU_{> (1-\delta) \hlambda_L} \hU_{> (1-\delta) \hlambda_L}^\top + \hU_{\leq (1-\delta) \hlambda_L} \hU^\top_{\leq (1-\delta) \hlambda_L}\right) \U_{\leq (1-2\delta) \lambda_L}}{2} \\
&\leq \matnorm{\V_L^\top \hU_{> (1-\delta) \hlambda_L} \hU_{> (1-\delta) \hlambda_L}^\top \U_{\leq (1-2\delta) \lambda_L}}{2} + \matnorm{\V_L^\top \hU_{\leq (1-\delta) \hlambda_L} \hU^\top_{\leq (1-\delta) \hlambda_L} \U_{\leq (1-2\delta) \lambda_L}}{2} \notag \\
&\leq \matnorm{\hU_{> (1-\delta) \hlambda_L}^\top \U_{\leq (1-2\delta) \lambda_L}}{2} +  \varepsilon \\
&\leq \frac{\matnorm{\S - \hSigma}{2}}{(1 - \delta)(\hlambda_L - \lambda_L) + \delta \lambda_L} + \varepsilon
\end{align*}
where the last inequality follows from the Gap-free Wedin Theorem \citep[Lemma B.3]{Zhu:16:LazySVD}.

\end{proof}

\subsubsection*{Proof of Theorem~\ref{thm:top_L_eigenvectors2}}

\begin{proof} Denote $\widetilde{\S}^{(S)} = \hSigma^{(S)} + \sum_{j=1}^{S} \hlambda_j \v_j \v_j^\top$ where $\hSigma^{(S)} = \left(I - \V_S\V_S^\top\right) \hSigma \left(I - \V_S\V_S^\top\right)$. Then by Corollary~\ref{corr:population-gap-free2}, $V_{(1-\delta) \hlambda_L}$ is the eigenspace of $\widetilde{\S}^{(S)}$ with eigenvalues less than or equal to $(1-\delta_S) \hlambda_S$, and $\hU_L$ is the eigenspace of $\hSigma$ with eigenvalues greater than or equal to $\hlambda_L$, by Gap-free Wedin Theorem \citep[Lemma B.3]{Zhu:16:LazySVD}, we have
\begin{align*}
\matnorm{\hU_L^\top V_{(1-\delta) \hlambda_L}}{2} \leq \frac{\matnorm{\widetilde{\S}^{(S)} - \hSigma}{2}}{\hlambda_L - \hlambda_S/(1-\delta_S)}.
\end{align*}
When $\hlambda_L - \hlambda_S / (1-\delta_S)$ (given by $\delta > \delta_S$), we only need to bound $\matnorm{\widetilde{\S}^{(S)} - \hSigma}{2}$. Towards this goal, we first bound the following quantities, $\v_j^\top (\widetilde{\S}^{(S)} - \hSigma) \v_k$, $\v_j^\top(\widetilde{\S}^{(S)} - \hSigma) \v$ and $\v^\top(\widetilde{\S}^{(S)} - \hSigma)\v$, for all $j, k \in [S]$ and for $\v \in \V_S^\perp$:
\begin{enumerate}[(1)]
    \item Term $\v^\top(\widetilde{\S}^{(S)} - \hSigma)\v$ for $\v \in \V_S^\perp$: a simple calculation yields $\v^\top(\widetilde{\S}^{(S)} - \hSigma)\v = \v^\top(\hSigma - \hSigma)\v = 0$.
    \item Term $\v_j^\top (\widetilde{\S}^{(S)} - \hSigma) \v_k$ for $j = k \in [S]$: we have
    \begin{align*}
    \left | \v_j^\top (\widetilde{\S}^{(S)} - \hSigma) \v_j \right| = \left| \hlambda_j - \v_j^\top \hSigma \v_j \right| \leq \frac{\delta_S}{1 - \delta_S} \hlambda_j,
    \end{align*}
    where the last step is due to Theorem~4.1(c) in~\cite{Zhu:16:LazySVD}.
    \item Term $\v_j^\top (\widetilde{\S}^{(S)} - \hSigma) \v_k$ for $1\leq j < k \leq  [S]$: according to our construction, we have $\v_u^\top \v_v =0$ for each pair $u,v\in[S]$ and $u\neq v$. Thus we have $\v_j^\top\widetilde{\S}^{(S)}\v_k = \sum_{s=1}^S  \hlambda_s \v_j^\top\v_s \v_s^\top \v_k =0$ and the following inequality
    \begin{align*}
    \left | \v_j^\top (\widetilde{\S}^{(S)} - \hSigma) \v_k \right| = \left | \v_j^\top \hSigma \v_k\right| = \left | \v_j^\top \hSigma \v_k - \hlambda_k \v_j^\top \v_k \right| \leq \left\| (\hSigma - \hlambda_k \I_d) \v_k\right\|_2.
    \end{align*}
    Recall that $\hSigma^{(k)} = (\I - V_{k-1}V_{k-1}^\top)\hSigma (\I - V_{k-1}V_{k-1}^\top)$ and $\v_k\perp V_{k-1}$. Therefore, $\v_k^\top \hSigma \v_k=\v_k^\top \hSigma^{(k)} \v_k$, and
    \begin{align}
     \left\| (\hSigma - \hlambda_k \I_d) \v_k\right\|_2^2& = \v_k^\top (\hSigma - \hlambda_k \I_d)^2 \v_k = \v_k^\top  \hSigma^2 \v_k - 2 \hlambda_k \v_k^\top \hSigma \v_k + \hlambda_k \| \v_k \|^2_2 \notag \\
    &= \v_k^\top  \hSigma^{(k)2} \v_k - 2 \hlambda_k \v_k^\top \hSigma^{(k)} \v_k + \hlambda_k \| \v_k \|^2_2 = \v_k^\top (\hSigma^{(k)} - \hlambda_k \I_d)^2 \v_k.
    \end{align}
    According to Theorem~4.1(b) in~\cite{Zhu:16:LazySVD}, we have
    $$\hlambda_k\leq \matnorm{\hSigma^{(k)}}{2}\leq \frac{\hlambda_k}{1-\delta}$$.
    Let $\hlambda^{(k)}_1 \geq \hlambda^{(k)}_2 \geq \cdots\geq \hlambda^{(k)}_d\geq 0$ denote the sorted eigenvalues of $\hSigma^{(k)}$, and $\hu^{(k)}_j$, $j\in[d]$ the associated eigenvectors. From the preceding display, we have $ \hlambda_k \leq \hlambda^{(k)}_1 \leq (1-\delta)^{-1}\hlambda_k$. Simple calculations yield
    \begin{align*}
     \v_k^\top (\hSigma^{(k)} - \hlambda_k \I_d)^2 \v_k  &=  \sum_{j=1}^d (\hlambda_1^{(k)} - \hlambda^{(k)}_j)^2 (\v_k^\top \hu^{(k)}_j)^2 \\
    &\leq \sum_{j: \hlambda_j^{(k)} \leq (1-\delta_S) \hlambda_1^{(k)}} \big(\hlambda_1^{(k)}\big)^2 (\v_k^\top \hu^{(k)}_j)^2 + \sum_{j: \hlambda^{(k)}_j > (1-\delta_S) \hlambda^{(k)}_1} \delta_S^2 \big(\hlambda^{(k)}_1\big)^2 (\v_k^\top \hu^{(k)}_j)^2 \\
    &\leq \varepsilon\big(\hlambda_1^{(k)}\big)^2  + \delta_S^2 \big(\hlambda_1^{(k)}\big)^2 \leq 2 \delta_S^2 \big(\hlambda_1^{(k)}\big)^2 \leq 2\delta_S^2(1-\delta)^{-2} \hlambda_k^2,
    \end{align*}
    when $\varepsilon$ is small enough, $\varepsilon \leq \delta_S^2$. Notice that for $2 \leq k \leq S$, the above result remains the same once we notice that $\v_k^\top \hSigma^{(k-1)} \v^k = \v_k^\top \hSigma \v_k$.
    \item Term $\v_j^\top(\widetilde{\S}^{(S)} - \hSigma) \v = 0$ for $\v \in \V_S^\perp$ and $j\in[S]$.
\end{enumerate}
Combine the result above, for any $\x \in \R^d, \| \x \|_2 = 1$, denote $\x = \sum_{j=1}^d \alpha_j \v_j, \; \sum_{j=1}^d \alpha_j^2 = 1$. We can show that,

\begin{align*}
\left| \x^\top (\widetilde{\S}^{(S)} - \hSigma) \x\right| &= \left| \sum_{j,k=1}^d \alpha_j \alpha_k \v_j^\top (\widetilde{\S}^{(S)} - \hSigma) \v_k\right| \\
&\leq \sum_{j,k=1}^S \frac{\delta_S}{1 - \delta_S} |\alpha_j \alpha_k| \hlambda_1 \\
&\leq \frac{\delta_S}{1 - \delta_S} S \hlambda_1.
\end{align*}
Therefore, we have,
\begin{align} \label{eq:top_L_eigenvectors3}
\matnorm{\hU_L^\top \V_{\leq (1-\delta) \hlambda_L}}{2} \leq S \; \frac{\delta_S \hlambda_1}{\hlambda_L(1-\delta_S) - \hlambda_S},
\end{align}
Using the same argument as in the proof of Theorem~\ref{thm:top_L_eigenvectors} (c), we have the desired result.
\end{proof}
\section{Applications for Distributed PCA}
\label{sec:application}

Our distributed PCA can be applied to a wide range of important applications. In this section, we discuss two applications to  principal component regression and single index model. Model assumptions and theoretical results are provided in this section and further numerical experiments will be presented in Appendix~\ref{sec:add-exp} in the supplementary material.

\subsection{Distributed PCA for principal component regression}
\label{subsec:setting-pcr}
Principal component regression  \citep{jeffers1967two,jolliffe1982note}
is built on the following multivariate linear model,
\begin{align}\label{eq:model_PCR}
  \y = \A \be^* + \eps.
 \end{align}
In \eqref{eq:model_PCR}, $\A=[\a_1,\ldots, \a_n]^\top$ is the $n \times d$ observed covariate matrix with \emph{i.i.d.} rows, where each $\a_i$ is a zero-mean random vector with the covariance matrix $\S$, $\be^*$ is the $d \times 1$ coefficient, and $\y$ is the $n \times 1$ response vector. The noise $\eps = \left[ \epsilon_{1} , \dots , \epsilon_{n} \right]^{\top}$ is the error term with $\ep(\epsilon_i) = 0$ and $\eps$ is independent from $\a$. Since our main purpose here is to illustrate PCR in our distributed algorithm, we assume that data dimension $d$  is a constant. Of course, it would be interesting to extend to high-dimensional case, and we leave it for future investigation.
Moreover, for the ease of technical derivation and presentation,
we assume that $\be^*$ is normalized with $\| \be^*\|_2 = 1$.

Our goal is to estimate the coefficient $\be^*$ from $\A, \y$. Denote $\U_L = \left[ \u_{1} , \ldots , \u_{L} \right]$, $1 \leq L \leq d$ the subspace spanned by the top-$L$ eigenvectors of $\S$. In principal component regression (denoted as PCR below), $\be^*$ is assumed to lie in the same subspace, i.e.,
\begin{align*}
  \be^* = \U_L \ga^*,
\end{align*}
for some vector $\ga^{*} \in \R^{L}$. Our goal is to estimate $\be^*$.

Let $\hSigma = \frac{1}{n} \A^\top \A$ be the sample covariance matrix with eigenvalues $\hlambda_{1} \geq \ldots \geq \hlambda_{d} > 0$. In the traditional setting with an explicit eigengap, we assume $\lambda_L > \lambda_{L+1}$ and  estimate $\U_L$ by the top-$L$ eigenspace of the empirical covariance matrix $\hSigma$, i.e., $\hU_L = [\hu_1, \ldots, \hu_L]$. Then the covariate data matrix $\A$ is projected on this estimated subspace, $\tA = \A \hU_{L}$ and an estimator $\hga$ of $\ga^*$ is obtained by ordinary least squares regression of the response vector $\y$ on the projected data matrix $\tA$,
\begin{align*}
  \hga = \left( \tA^{\top} \tA \right)^{- 1} \tA^{\top} \y.
\end{align*}
Therefore, the standard PCR estimator is $\hbe = \hU_{L} \hga$.

In a gap-free setting, we cannot directly estimate $\U_L$ due to the lack of an eigengap assumption. Instead, we consider an enlarged eigenspace estimator $\V_S = \left[\v_{1} , \ldots , \v_{S} \right]$ given by our Algorithm~\ref{algo:distr_multi} where $S > L$ is defined as before, i.e., $S :\,=\argmax\{l:\,\hlambda_l > (1-\delta)\,\hlambda_L\}$, for some pre-determined parameter $\delta > 0$. In a distributed environment, our data is split uniformly on $K$ local machines. The data on each machine is denoted by $\A_k \in \R^{m \times d}, \y_k \in \R^{m \times 1}$ for $k=1,\ldots, K$. Now we can obtain the corresponding projected data matrix $\tA_k = \A_k \V_S$ on each machine. The $\tA_k^\top \tA_k, \; \tA_k^\top \y_k$ are then computed locally and collected by the central machine. Then the central machine computes an OLS estimator based on $\{\y_k\}$ and $\{\tA_k\}$,
\begin{align}
\label{eq:dist-pcr-regress}
  \widetilde{\ga} = \left( \sum_{k=1}^K \tA_k^\top \tA_k \right)^{- 1} \left(\sum_{k=1}^K \tA_k^\top \y_k \right) = \left( \tA^{\top} \tA \right)^{- 1} \tA^{\top} \y.
\end{align}
Finally, our distributed estimator is obtained from $\tbe = \V_{S} \widetilde{\ga}$.

Proposition~\ref{prop:pcr1} below first describes the upper bound of estimation error for the usual PCR result where data matrix is on one machine and explicit eigengap is assumed. The technical proof in this subsection will be deferred to Appendix~\ref{sec:pcr-sim-proof} in the supplementary material.

\begin{proposition} \label{prop:pcr1} Assume the noise term $\{\epsilon_i\}_{i=1}^n$ are sub-Gaussian($\sigma^2$) random variables that are independent from each other and from covariate $\A$. We further assume $\lambda_L > \lambda_{L+1}$, then the single-machine estimator $\hbe$ satisfies
\begin{align} \label{eq:pcr1-bnd}
\frac{1}{n}\left\| \A \hbe - \A \be^{*}\right\|_2^2 \leq \hlambda_1 \hlambda_L^{-2} \; \frac{d}{n} + \hlambda_1 \frac{\matnorm{\hSigma - \S}{2}^2}{(\lambda_L - \hlambda_{L+1})^2},
\end{align}
with high probability. Here the omitted constant in $\lesssim$ depend on $\sigma^2$ and $ \| \ga^*\|_2$.
\end{proposition}

We would like to make some remarks on Proposition \ref{prop:pcr1}. By Theorem 5.39 in \cite{Vershynin2010}, with high probability $\hlambda_L \geq \lambda_L + o(1)$. Therefore, $\hlambda_L^{-1}$ is bounded with high probability. Moreover, since $\matnorm{\S - \hSigma}{2} = \mathcal{O}_p(\sqrt{d / n})$, this result indicates that the oracle estimator of PCR enjoys a statistical convergence rate of order $\mathcal{O}_p(\sqrt{d / n})$.

Now we are ready to  provide our result on the distributed PCR with no eigengap assumption. The proof of Theorem \ref{thm:pcr2} will be provided in  Appendix~\ref{sec:pcr-sim-proof} in the supplementary material.

\begin{theorem} \label{thm:pcr2} Assume the noise term $\{\epsilon_i\}_{i=1}^n$ are sub-Gaussian($\sigma^2$) random variables that are independent from each other and from covariate $\A$.
If there exists $\delta_S<\delta/2$ such that  $\matnorm{\hU_{\leq (1- \delta_S)\hlambda_S}^{\top} \V_S}{2} \leq \frac{\delta_S}{16 \hlambda_1 / \hlambda_{S+1}}$, then with high probability, the prediction error of our distributed PCR estimator $\tbe=\V_{S} \widetilde{\ga}$ satisfies,
\begin{align}
\label{eq:pcr2-bnd}
\frac{1}{n} \left\| \A \tbe - \A \be^{*} \right\|_{2}^2 \lesssim \; \; &\hlambda_1 \left[(1-\delta_S)\hlambda_S\right]^{-2} \frac{d}{n} + \hlambda_1 \frac{\matnorm{\S - \hSigma}{2}^2}{\left(\lambda_L - (1-\delta)  \hlambda_L\right)^2} \notag \\
+ &\frac{S^2 \delta_S^2 }{(1-\delta_S)^2} \; \frac{\hlambda_1^3}{(1-\delta/2)\hlambda_L - \hlambda_S / (1-\delta_S)}.
\end{align}
Here the omitted constant in $\lesssim$ depend on $\sigma^2$ and $\| \ga^*\|_2$.
\end{theorem}

Notice that when $\lambda_L > \lambda_{L+1}$ as in the explicit eigengap case, we can simply set $\delta = \hlambda_{L} / (\hlambda_{L} - \hlambda_{L+1})$, $S = L$, and $\delta_S = 0$. Then our error bound for distributed estimator $\tbe$ in Equation~\eqref{eq:pcr2-bnd} will be
the same as the classical bound in \eqref{eq:pcr1-bnd} (up to a constant factor).

\subsection{Distributed PCA for single index model}
\label{subsec:setting-sim}

In a standard single index model (denoted as SIM below), we assume,
\begin{align*}
y = f(\langle \be^*, \a \rangle) + \epsilon,
\end{align*}
where $y \in \R$ is the response, $\a$ is the $d$-dimensional covariate vector, $\be^*\in \R^{d}$ is the parametric component and $\epsilon$ is a zero-mean noise that is independent of $\a$. Here, the so-called link function $f: \R \; \mapsto \; \R$ is the nonparametric component. We also focus on the low-dimensional setting where $d$ does not grow with the sample size $n$. For the model identifiability,  we assume that $\|\be^*\|_2 = 1$ since $\| \be^*\|_2$ can be absorbed into $f$. Following \cite{li1992principal,Janzamin2014} and references therein, we can use the second order Stein’s identity to estimate $\be^*$.

\begin{proposition} \label{second-order stein} Assume that the density $p$ of $\a$ is twice differentiable. In addition, we define the second-order score function $T: \R^{d}\mapsto \R^{d \times d }$ as
\begin{align*}
T(\a) = \nabla^{2}p(\a) / p(\a).
\end{align*}
Then, for any twice differentiable function $g: \R^{d} \mapsto \R$ such that $\ep \left[ \nabla^{2} g (\a) \right]$ exists, we have
\begin{align*}
\ep[g(\a) \cdot T(\a)] = \ep\left[\nabla^{2} g(\a) \right].
\end{align*}
\end{proposition}

Now we consider the SIM with Gaussian distribution as a special case, where $\a \sim \mathcal{N}(\boldsymbol{0}, \I_d)$. The second order score function now becomes
\begin{align*}
T(\a) = \a \a^{\top} - \I_d.
\end{align*}
By Proposition~\ref{second-order stein} we have
\begin{align}
\ep[ y \left(\a \a^{\top} - \I_d\right) ] = C_{0}\cdot \be^{*}\be^{* \top },
\end{align}
where $C_{0}= 2 \ep\left[ f^{\prime \prime}\left( \left\langle \a , \be^{*}\right\rangle \right) \right]$.

Therefore, one way to estimator $\be^*$ is to obtain the leading eigenvector of $\ep[y \cdot (\a \a^\top - \I_d)]$ from samples. Given $n$ \emph{i.i.d.} sample $\{\a_i, y_i\}_{i=1}^n$, we can calculate the estimator $\hbe$ by extracting the leading eigenvector of $\frac{1}{n} \sum_{i=1}^n y_i \cdot (\a_i \a_i^\top - \I_d)$, the empirical estimation of $\ep[ y \left(\a \a^{\top} - \I_d\right) ]$. This can also be extended to our distributed setting, where we estimate $\hbe$ by $\widetilde{\be}$ from the distributed PCA Algorithm~\ref{algo:distr_top}.

Let $\lambda_1 \geq \lambda_2 \geq \ldots \geq \lambda_d$ denote  the eigenvalues of population matrix $\ep[ y \left(\a \a^{\top} - \I_d\right) ]$ and $\hlambda_1 \geq \hlambda_2 \geq \ldots \geq \hlambda_d$ the eigenvalues for the empirical matrix $\frac{1}{n} \sum_{i=1}^n y_i \cdot (\a_i \a_i^\top - \I_d)$ calculated with pooled data. Before presenting our theoretical results, we make some standard assumptions.

\begin{assumption} \label{ass:sim}
Under the Gaussian SIM model given above, we further assume that

\begin{enumerate}[(1)]
    \item We assume $f$ and $\a$ are such that $\ep \left[ f^{\prime \prime} \left( \left\langle \a , \be^{*} \right\rangle \right) \right] > 0$,  and moreover, $f(\langle \be^*, \a \rangle)$ is bounded. 
    \item We assume the noise term $\{\epsilon\}_{i=1}^n$ to be independent, zero-mean sub-Gaussian($\sigma$) random variables.
\end{enumerate}
\end{assumption}


Item (1) is commonly assumed in many references, for example, Definition 2.6 in \cite{Yang2017}. We further make a boundedness assumption on $f(\cdot)$. This assumption make use of the fact that $\a$ is spherical Gaussian, thus with high probability, $\langle \be^*, \a \rangle$ is bounded and $f(\cdot)$ only need to be finite in this domain. 


Still, we consider the single-machine case first. The following proposition quantifies the statistical rate of convergence of the non-distributed estimator $\hbe$. We defer the proof of this theorem to Appendix~\ref{sec:pcr-sim-proof} in the supplement.

\begin{theorem} \label{thm:sim1} Under Gaussian SIM model and Assumption~\ref{ass:sim}, our estimator $\hbe$ satisfies
\begin{align}
\label{eq:sim1}
\min_{t \in \{-1,+1\}} \| t \hbe - \be^* \|_2 = \widetilde{\mathcal{O}}_p \left(\frac{d }{\sqrt{n}} \right),
\end{align}
with high probability.
\end{theorem}

We can extend the Gaussian SIM model to our distributed framework where data  are stored on different machines. The  transformed ``covariance matrix'' on each machine $k$ has the form $\frac{1}{m} \sum_{j=1}^m y_j^{(k)} (\a_j^{(k)} \a_j^{(k)\top} - \I_d)$ for $k=1, \ldots, K$. Then it is straightforward to apply Algorithm~\ref{algo:distr_top} to obtain a distributed estimation $\tbe$ of $\be^*$. Combining the results in Theorem~\ref{thm:top_eigenvector}, Corollary~\ref{corr:top_eigenvector}, and Theorem~\ref{thm:sim1}, we can obtain the non-asymptotic upper bound for our distributed estimator $\tbe$.

\begin{proposition}
\label{prop:sim2}
In a distributed environment, if we estimate $\be^*$ with $\tbe$ using Algorithm~\ref{algo:distr_top}. We set the numbers of inner iterations and outer iterations to be $T$ simultaneously. Under Assumption~\ref{ass:sim} and if $\eta\leq \delta\hlambda_1/16$ in Corollary~\ref{corr:top_eigenvector}, our distributed estimator $\tbe$ satisfies,
\begin{align}
\label{eq:sim2-bnd}
 \min_{t \in \{-1,+1\}} \|t \tbe - \be^* \|_2 \leq \widetilde{\mathcal{O}}_p \left(\frac{d }{\sqrt{n}} \right) + \mathcal{O}_p \left( \left[\frac{d \log d}{m}\right]^{T/2} \right).
\end{align}
\end{proposition}

Proposition~\ref{prop:sim2} indicates that when the number of iterations $T$ is sufficiently large, the first error term in \eqref{eq:sim2-bnd} will dominate the second one and therefore our estimator $\tbe$ will have the convergence rate of $\widetilde{\mathcal{O}}_p \left(d / \sqrt{n} \right)$.
\section{Proofs of Theoretical Results for Applications}
\label{sec:pcr-sim-proof}

\subsection{Proofs of theoretical results in Section~\ref{subsec:setting-pcr}}

\subsubsection*{Proof of Proposition~\ref{prop:pcr1}}

\begin{proof}
First, notice that
\begin{align*}
\frac{1}{n}\left\| \A \hbe - \A \be^{*}\right\|_2^2 \leq \frac{1}{n}\matnorm{\A}{2}^2 \left\| \hbe - \be^{*}\right\|_2^2 = \hlambda_1 \left\| \hbe - \be^{*}\right\|_2^2.
\end{align*}
Now consider $\| \hbe - \be^{*}\|_2$, we have
\begin{align}
\label{eq:decomp1}
\left\| \hbe - \be^{*}\right\|_2
  &= \left\| \hU_{L} \OO \OO^{\top}\hga- \U_{L}\ga^{*}\right\|_2 \notag \\
  &= \left\| \hU_{L}\OO \OO^{\top}\hga- \hU_{L} \OO\ga^* + \hU_{L} \OO\ga^* - \U_{L}\ga^{*}\right\|_2 \notag \\
  &\leq \|  \OO^{\top}\hga- \ga^{*}\|_2 + \| \ga^{*}\|_2 \matnorm{\hU_{L} \OO- \U_{L}}{2},
\end{align}
for any orthogonal matrix $\OO \in \R^{L \times L}, \; \OO^\top \OO = \I_L$.

Now consider the first part on the RHS.
\begin{align}
\label{eq:decomp2}
  \left\| \OO^{\top} \hga - \ga^{*} \right\|_2
  &= \left\| \OO^{\top}\left( \tA^{\top}\tA \right)^{- 1}\tA^{\top}\left( \A\left( \U_{L}- \hU_{L} \OO\right) \ga^{*}+ \eps\right) \right\|_2 \notag \\
  &\leq \left\| \left( \tA^{\top}\tA \right)^{- 1}\tA^{\top}\eps\right\|_2 + \matnorm{\left( \tA^{\top}\tA \right)^{- 1}\tA^{\top}\A}{2} \left\| \left( \U_{L}- \hU_{L} \OO\right) \ga^{*}\right\|_2 \notag \\
  &\leq \left\| \left( \tA^{\top}\tA \right)^{- 1}\tA^{\top}\eps\right\|_2  + \left\| \left( \U_{L}- \hU_{L} \OO\right) \ga^{*}\right\|_2.
\end{align}
Therefore, plug inequality~\eqref{eq:decomp2} into inequality~\eqref{eq:decomp1} we have
\begin{align*}
  \left\| \hbe - \be^{*}\right\|_2
  \leq \| ( \tA^{\top}\tA )^{- 1}\tA^{\top}\eps\|_2 + 2 \matnorm{\hU_{L} \OO- \U_{L}}{2} \| \ga^{*}\|_2
\end{align*}
Take the $L \times L$ orthogonal matrix $\OO = \hU_{L}^\top \U_L$, we have
\begin{align}
\label{eq:d-khan}
\matnorm{\hU_{L}\OO- \U_{L}}{2} &= \matnorm{\hU_{L}\hU_{L}^\top \U_L- \U_{L}}{2} = \sqrt{\matnorm{\U_L^\top \left(I_d - \hU_{L} \hU_{L}^\top \right) \U_L}{2}} \notag \\
&= \matnorm{\U_{L}^\top \hU_{L}^\perp}{2} \leq \frac{\matnorm{\hSigma - \S}{2}}{\lambda_L - \hlambda_{L+1}},
\end{align}
where the third equality uses the fact that $\hU_{L} \hU_{L}^\top + \hU_{L}^\perp \hU_{L}^{\perp\top} = \I_d$ and the last inequality follows from Lemma B.3 in \cite{Zhu:16:LazySVD}. Also,
\begin{align*}
\left\| ( \tA^{\top}\tA )^{- 1}\tA^{\top}\eps\right\|_2
\leq \matnorm{\left( \frac{\tA^{\top}\tA}{n} \right)^{- 1}}{2} \left\|  \frac{\tA^{\top}\eps}{n} \right\|_2
= \; \hlambda_L^{-1} \left\|  \frac{\A^{\top}\eps}{n} \right\|_2.
\end{align*}
Denote $\xi = \frac{\tA^{\top}\eps}{n}$. Notice that $\xi$ is a $d$-dimension sub-Gaussian vector with variance proxy $\frac{\sigma^2}{n}$. Therefore, with probability at least $1 - e^{-C^2}$
\begin{align*}
\| \xi \|_2 = \sqrt{\sum_{i=1}^d \xi_i^2} \leq \sqrt{2} C \sigma \sqrt{\frac{d}{n}}.
\end{align*}
Thus we have,
\begin{align}
\label{eq:err-cct}
  \left\| ( \tA^{\top}\tA )^{- 1}\tA^{\top}\eps\right\|_2 \leq \sqrt{2} C \sigma \; \hlambda_L^{-1} \sqrt{\frac{d}{n}}.
\end{align}
Combine inequality~\eqref{eq:d-khan} with inequality~\eqref{eq:err-cct} we obtain the desired result,
\begin{align*}
\frac{1}{n}\left\| \A \hbe - \A \be^{*}\right\|_2^2 \leq 4 C^2 \sigma^2 \; \hlambda_1 \hlambda_L^{-2} \; \frac{d}{n} + 2 \hlambda_1 \frac{\matnorm{\hSigma - \S}{2}^2}{(\lambda_L - \hlambda_{L+1})^2} \| \ga^{*}\|_2^2,
\end{align*}
with probability at least $1 - e^{-C^2}$.

\end{proof}

\subsubsection*{Proof of Theorem~\ref{thm:pcr2}}

\begin{proof}
Notice that
\begin{align}
\label{eq:decom3}
\frac{1}{n} \left\| \A \hbe - \A \be^{*}\right\|_2^2
&= \frac{1}{n} \left \| \A \V_{S} \left(\V_{S}^\top \A^\top \A \V_{S}\right)^{-1} \V_{S}^\top \A^\top  \left(\A \U_L \ga^* + \va\right) - \A \U_L \ga^* \right\|_2^2 \notag \\
&\leq \frac{2}{n} \left\|\A \V_{S} \left(\V_{S}^\top \A^\top \A \V_{S}\right)^{-1} \V_{S}^\top \A^\top\va \right\|_2^2 \notag \\
& \; + \frac{2}{n} \matnorm{\left(\A \V_{S} \left(\V_{S}^\top \A^\top \A \V_{S}\right)^{-1} \V_{S}^\top \A^\top - \I_d \right) \A \U_L \ga^*}{2}^2.
\end{align}
Similar to \eqref{eq:err-cct} in the previous proof, the first term on the RHS satisfies,
\begin{align}
\label{eq:err-cct2}
\frac{2}{n}\left\|\A \V_{S} \left(\V_{S}^\top \A^\top \A \V_{S}\right)^{-1} \V_{S}^\top \A^\top\va \right\|_2^2 \leq 4 C^2 \sigma^2 \; \hlambda_1 \left[(1-\delta_S)\hlambda_S\right]^{-2} \frac{d}{n},
\end{align}
with probability at least $1 - e^{-C^2}$. Here the $(1-\delta_S) \hlambda_S$ term is given by inequality~\eqref{eq:cov_ineq3} in Corollary~\ref{corr:population-gap-free2}. For the second term on the RHS, notice that we have the following identity
\begin{align*}
\left(\A \V_{S} \left(\V_{S}^\top \A^\top \A \V_{S}\right)^{-1} \V_{S}^\top \A^\top - \I_d \right) \A \V_{S} \w = 0,
\end{align*}
for any $\w \in \R^S$. We set $\w = \V_S^\top \U_L \ga^*$ and can obtain that
\begin{align}
\label{eq:decom4}
&\matnorm{\left(\A \V_{S} \left(\V_{S}^\top \A^\top \A \V_{S}\right)^{-1} \V_{S}^\top \A^\top - \I_d \right) \A \U_L \ga^*}{2} \notag \\
= \; &\matnorm{\left(\A \V_{S} \left(\V_{S}^\top \A^\top \A \V_{S}\right)^{-1} \V_{S}^\top \A^\top - \I_d \right) \left(\A \U_L \ga^* - \A \V_{S} \V_S^\top \U_L \ga^*\right)}{2} \notag \\
\leq \; & \matnorm{\A \V_{S} \left(\V_{S}^\top \A^\top \A \V_{S}\right)^{-1} \V_{S}^\top \A^\top - \I_d}{2} \|\A \U_L \ga^* - \A \V_{S} \V_S^\top \U_L \ga^* \|_2
\end{align}
Denote the singular value decomposition of $\A \V_S$ as $\P \W \Q^\top$, where $\P \in \R^{d \times d}, \Q \in \R^{S \times S}$ are two orthogonal matrix and $\W \in \R^{d \times S}$ is a diagonal matrix. Therefore,
\begin{align}
\label{eq:projection-mat}
\matnorm{\A \V_{S} \left(\V_{S}^\top \A^\top \A \V_{S}\right)^{-1} \V_{S}^\top \A^\top - \I_d}{2} = \matnorm{\sum_{i=S+1}^d \p_i \p_i^\top}{2} = 1.
\end{align}
Now plug equality~\eqref{eq:projection-mat} into inequality~\eqref{eq:decom4}, it can be shown that
\begin{align} \label{eq:gap-free-pcr}
&\frac{2}{n} \matnorm{\left(\A \V_{S} \left(\V_{S}^\top \A^\top \A \V_{S}\right)^{-1} \V_{S}^\top \A^\top - \I_d \right) \A \U_L \ga^*}{2}^2 \notag \\
\leq \;\; &\frac{1}{n} \matnorm{\A}{2}^2 \matnorm{\U_L - \V_{S} \V_S^\top \U_L}{2}^2 \| \ga^* \|_2^2 \notag\\
\leq \;\; &2 \hlambda_1  \matnorm{\U_L^\top \left(I_d - \V_{S} \V_S^\top \right) \U_L}{2} \| \ga^* \|_2^2 \notag\\
\leq \;\; &2 \hlambda_1  \matnorm{\U_L^\top \V_{S}^\perp \V_S^{\perp\top} \U_L}{2} \| \ga^* \|_2^2 \notag\\
\leq \;\; &2 \hlambda_1  \matnorm{\U_L^\top \V_{S}^\perp}{2}^2 \| \ga^* \|_2^2 \notag\\
\leq \;\; &2 \hlambda_1 \| \ga^* \|_2^2 \left(\frac{\matnorm{\S - \hSigma}{2}^2}{\left(\lambda_L - (1-\delta/2)  \hlambda_L\right)^2} + 4 S^2 \; \frac{\delta_S^2 \hlambda_1^2}{(1-\delta/2)\hlambda_L - \hlambda_S / (1-\delta_S)}\right) \notag \\
\leq \;\; &2 \hlambda_1 \| \ga^* \|_2^2 \left(\frac{\matnorm{\S - \hSigma}{2}^2}{\left(\lambda_L - (1-\delta)  \hlambda_L\right)^2} + \frac{S^2 \delta_S^2 }{(1-\delta_S)^2} \; \frac{\hlambda_1^2}{(1-\delta/2)\hlambda_L - \hlambda_S / (1-\delta_S)}\right).
\end{align}
Here the last inequality follows from Theorem~\ref{thm:top_L_eigenvectors2} and the proof in Corollary~\ref{corr:population-gap-free}. Our proof is completed once we combine inequality~\eqref{eq:decom3}, ~\eqref{eq:err-cct2} and ~\eqref{eq:gap-free-pcr}.

\end{proof}

\subsection{Proofs of theoretical sesults in Section~\ref{subsec:setting-sim}}

\subsubsection*{Proof of Theorem~\ref{thm:sim1}}

We will need a definition on Orlicz norm from \cite{Ledoux2013}, in order to deal with random variables whose tail is heavier than sub-Exponential variables.

\begin{definition} For $1 \leq \alpha < \infty$, let $\psi_\alpha = \exp(x^\alpha) - 1$. For $0 < \alpha < 1$, let $\psi_\alpha(x) = \exp(x^\alpha) - 1$ for large enough $x \geq x_\alpha$ and $\psi_\alpha$ is linear in $[0, x_\alpha]$ in order to remain global convexity. The Orlicz norm $\psi_\alpha$ of a random variable $X$ is defined as
\begin{align*}
\| X \|_{\psi_\alpha} \triangleq \inf \left\{c \in (0, \infty) | \ep\left[\psi_\alpha(|X| / c) \leq 1 \right] \right\}.
\end{align*}
\end{definition}

\begin{proof}
Denote $\S = \mathbb{E}[ y \left(\a \a^{\top}- \I_d\right) ] $, $\hSigma = \frac{1}{n} \sum_{i=1}^n y_i \cdot (\a_i \a_i^\top - \I_p)$. We first obtain a high probability bound on $\matnorm{ \hSigma - \S}{2}$.

Without loss of generality, we can assume $M = 1$, otherwise we can always multiply $M$ on the bound we obtain. Note that
\begin{align}
\label{eq:err-cov}
&\matnorm{\hSigma - \S }{2} \notag\\
= &\matnorm{ \frac{1}{n} \sum_{i=1}^n y_i (\a_i \a_i^\top - \I_p) - \mathbb{E}[f(\langle \be^*, \a \rangle) \left(\a \a^{\top}- \I_d\right) ]}{2} \notag\\
\leq &\matnorm{\frac{1}{n} \sum_{i=1}^n f(\langle \be^*, \a_i \rangle) (\a_i \a_i^\top - \I_p) - \mathbb{E}[f(\langle \be^*, \a \rangle) \left(\a \a^{\top}- \I_d\right) ]}{2} + \matnorm{ \frac{1}{n} \sum_{i=1}^n \epsilon_i (\a_i \a_i^\top - \I_d)}{2}.
\end{align}
Under Lemma 5.4 in \cite{Vershynin2010}, We can evaluate the operator norm on the RHS of inequality~\eqref{eq:err-cov} on a $\frac{1}{4}$-net $\mathcal{E}$ of the unit sphere $\mathcal{S}^{d-1}$:
\begin{align*}
&\matnorm{\frac{1}{n} \sum_{i=1}^n f(\langle \be^*, \a_i \rangle) (\a_i \a_i^\top - \I_p) - \mathbb{E}[f(\langle \be^*, \a \rangle) \left(\a \a^{\top}- \I_d\right) ] }{2} \\
\leq \; &2 \max_{\v \in \mathcal{E}} \left| \frac{1}{n} \sum_{i=1}^n f(\langle \be^*, \a_i \rangle)(z(\v)_i^2 - 1) - \mathbb{E}[f(\langle \be^*, \a \rangle) (z(\v)^2 - 1) ]\right|,
\end{align*}
where $z(\v) = \v^\top \a \sim \mathcal{N}(0,1)$. Notice that $z(\v)^2 - 1$ is a sub-Exponential with parameter $(2,4)$, and $\mathbb{E}\{f(\langle \be^*, \a \rangle) (z(\v)^2 - 1)\} \geq - 2$. Therefore, denote $\theta = \max\{1, 2M\}$, $\forall |s| \leq \frac{1}{8 \theta}$,
\begin{align*}
&\quad \; \mathbb{E}\left[\exp\{s (f(\langle \be^*, \a \rangle) (z(\v)^2 - 1) - \mathbb{E}f(\langle \be^*, \a \rangle) (z(\v)^2 - 1) )\} \right] \\
&\leq 1 + \sum_{k=2}^\infty \frac{|s|^k}{k!} \mathbb{E}\left|f(\langle \be^*, \a \rangle) (z(\v)^2 - 1) - \mathbb{E}\left(f(\langle \be^*, \a \rangle) (z(\v)^2 - 1)\right) \right|^k \\
&\leq 1 + \sum_{k=2}^\infty \frac{|s|^k}{k!} 2^k \mathbb{E}\left|f(\langle \be^*, \a \rangle) (z(\v)^2 - 1)\right|^k \\
&\leq 1 + \sum_{k=2}^\infty \frac{|s|^k}{k!} (2M)^k \mathbb{E}\left| (z(\v)^2 - 1)\right|^k \\
&\leq 1 + \sum_{k=2}^\infty (|s| 8 M)^k \\
&\leq \exp\{s^2 (16M)^2 / 2\}.
\end{align*}
Thus, $\widetilde{X} \triangleq f(\langle \be^*, \a \rangle) (z(\v)^2 - 1)$ is a sub-Exponential with parameter $(16M, 8 \theta)$. Therefore, by Proposition 5.16 in \cite{Vershynin2010}, we can obtain a Bernstein-type inequality:
\begin{align*}
\mathbb{P}\left[ \left|\frac{1}{n} \sum_{i=1}^n \widetilde{X}_i - \mathbb{E} \widetilde{X} \right| > 8 \theta \epsilon \right] \leq 2\exp \left[- \frac{n}{2} \left(\epsilon^2 \wedge \epsilon\right)\right],
\end{align*}
for any $\epsilon > 0$. Now let $\epsilon = \max(\gamma, \gamma^2)$, where $\gamma = C \sqrt{\frac{d}{n}} + \frac{\delta_1}{\sqrt{n}}$, for some constant $C$ and $\delta_1 > 0$. Now we have,
\begin{align*}
\mathbb{P}\left[ \left|\frac{1}{n} \sum_{i=1}^n \widetilde{X}_i - \mathbb{E} \widetilde{X} \right| > 8 \theta \gamma \right] \leq 2 \exp \left[- \frac{n}{2} \gamma^2 \right] \leq 2 \exp \left[- \frac{1}{2} (C^2 d + \delta_1^2)  \right].
\end{align*}
Notice that by Lemma 5.2 in \cite{Vershynin2010}, we can choose the net $\mathcal{E}$ so that it has cardinality $|\mathcal{E}| \leq 9^d$. Therefore, we take the union bound over all vectors $\v \in \mathcal{E}$, we obtain
\begin{align}
\label{eq:bstn-subexp}
\mathbb{P}\left[ \max_{\v \in \mathcal{E}}\left|\frac{1}{n} \sum_{i=1}^n \widetilde{X}_i - \mathbb{E} \widetilde{X} \right| > 8 \theta \gamma \right] \leq 2 \times 9^d \exp \left[- \frac{1}{2} (C^2 d + \delta_1^2)  \right] = 2 \exp \left(- \frac{\delta_1^2}{2}\right),
\end{align}
where we can choose $C$ sufficiently large, e.g. $C = 2\sqrt{\ln 3}$.

For the second part on the RHS of the inequality~\eqref{eq:err-cov}, we have
\begin{align*}
\matnorm{\frac{1}{n} \sum_{i=1}^n \epsilon_i (\a_i \a_i^\top - \I_d)}{2} \leq \matnorm{ \frac{1}{n} \sum_{i=1}^n \epsilon_i \a_i \a_i^\top }{2} + \left|\frac{1}{n} \sum_{i=1}^n \epsilon_i \right|.
\end{align*}
Given that $\epsilon$s are independent sub-Gaussian($\sigma$) random variables with mean $0$, we have
\begin{align}
\label{eq:bstn-subgs}
\left|\frac{1}{n} \sum_{i=1}^n \epsilon_i \right| \leq \sigma \delta_2 \frac{1}{\sqrt{n}},
\end{align}
with probability at least $1 - 2 \exp(- \delta_2^2 / 2)$.

Now, in order to control $\matnorm{ \frac{1}{n} \sum_{i=1}^n \epsilon_i \a_i \a_i^\top}{2}$, we consider $\matnorm{ \sum_{i=1}^n \| \a_i\|_2^2 \a_i \a_i^\top }{2}$ first. Under a $\frac{1}{4}$-net $\mathcal{E}_2$ of the unit sphere $\mathcal{S}^{d-1}$,
\begin{align*}
\matnorm{ \sum_{i=1}^n \| \a_i\|_2^2 \a_i \a_i^\top - \mathbb{E} \sum_{i=1}^n \| \a_i\|_2^2 \a_i \a_i^\top }{2}
&\leq 2 \max_{\v \in \mathcal{E}_2} \left| \sum_{i=1}^n \| \a_i\|_2^2 (\v^\top \a_i)^2 - \mathbb{E}\sum_{i=1}^n \| \a_i\|_2^2 (\v^\top \a_i)^2\right| \\
&= 2 \max_{\v \in \mathcal{E}_2} \left| \sum_{i=1}^n \| \a_i\|_2^2 (\v^\top \a_i)^2 - \mathbb{E} \sum_{i=1}^n \| \a_i\|_2^2 (\v^\top \a_i)^2\right| \\
&= 2 \max_{\v \in \mathcal{E}_2} \left| \sum_{i=1}^n \sum_{j=1}^d \a_{ij}^2 z(\v)_i^2 - \mathbb{E} \sum_{i=1}^n \sum_{j=1}^d \a_{ij}^2 z(\v)_i^2 \right|,
\end{align*}
where $\a_{ij} \sim \mathcal{N}(0,1)$ is the $j$-th term of $\a_i$ and $z(\v)_i = \v^\top \a_i \sim \mathcal{N}(0,1)$. Notice that $\a_{ij}^2, z(\v)_i^2 \sim \chi^2(1)$, and
\begin{align*}
\mathbb{P}(\a_{ij}^2 z(\v)_i^2 \geq t) \leq \mathbb{P}(\a_{ij}^2 \geq \sqrt{t}) + \mathbb{P}(z(\v)_i^2 \geq \sqrt{t}) \leq 2 \exp(- \sqrt{t}/2),
\end{align*}
for all $i \in \{1, \cdots, n\}, j \in \{1, \cdots, d\}$. Therefore, denote
\begin{align*}
\| \a_{ij}^2 z(\v)_i^2 \|_{\psi_{1/2}} \leq K_\psi,
\end{align*}
where $K_\psi$ is a finite constant. By Theorem 8.4 in \cite{Ma2015}, there exists a constant $K_\alpha$ such that
\begin{align*}
\left\| \sum_{i=1}^n \sum_{j=1}^d \a_{ij}^2 z(\v)_i^2 - \mathbb{E} \, \sum_{i=1}^n \sum_{j=1}^d \a_{ij}^2 z(\v)_i^2 \right\|_{\psi_{1/2}} \leq K_\alpha K_\psi \sqrt{nd} \log (nd).
\end{align*}
Denote $Z = \sum_{i=1}^n \sum_{j=1}^d \a_{ij}^2 z(\v)_i^2 - \mathbb{E} \, \sum_{i=1}^n \sum_{j=1}^d \a_{ij}^2 z(\v)_i^2$ and $K_Z = \| \sum_{i=1}^n \sum_{j=1}^d \a_{ij}^2 z(\v)_i^2 - \mathbb{E} \, \sum_{i=1}^n \sum_{j=1}^d \a_{ij}^2 z(\v)_i^2 \|_{\psi_{1/2}}$. Using Markov inequality, we have
\begin{align*}
\mathbb{P}\left(\left| Z \right| > t\right) \leq \frac{\mathbb{E}\left[\psi_{1/2}(Z / K_Z) + 1\right]}{\psi_{1/2}(t / K_Z) + 1} \leq 2 \exp\left\{- \left(\frac{t}{K_Z}\right)^{1/2}\right\}.
\end{align*}
By a union bound, we have
\begin{align*}
\mathbb{P}\left[\matnorm{ \sum_{i=1}^n \| \a_i\|_2^2 \a_i \a_i^\top - \mathbb{E} \sum_{i=1}^n \| \a_i\|_2^2 \a_i \a_i^\top}{2} \geq t\right] \leq 2 \times 9^d \exp\left\{- \left(\frac{t}{K_Z}\right)^{1/2}\right\}.
\end{align*}
Therefore, with probability at least $1 - 2 \exp(- \delta_3)$,
\begin{align*}
\matnorm{ \sum_{i=1}^n \| \a_i\|_2^2 \a_i \a_i^\top - \mathbb{E} \sum_{i=1}^n \| \a_i\|_2^2 \a_i \a_i^\top }{2} \leq 5 K_\alpha K_\psi \sqrt{nd} \log (nd) d^2 + \delta_3^2 K_\alpha K_\psi \sqrt{nd} \log (nd).
\end{align*}
Note that
\begin{align*}
\matnorm{ \mathbb{E} \sum_{i=1}^n \| \a_i\|_2^2 \a_i \a_i^\top }{2} \leq
 n \matnorm{ \mathbb{E} \| \a_i\|_2^2 \a_i \a_i^\top\ }{2} \leq 3 n d.
\end{align*}
Thus, with probability at least $1 - 2 \exp(- \delta_3)$,
\begin{align*}
\matnorm{ \sum_{i=1}^n \| \a_i\|_2^2 \a_i \a_i^\top }{2}
&\leq 5 K_\alpha K_\psi \sqrt{nd} \log (nd) d^2 + \delta_3^2 K_\alpha K_\psi \sqrt{nd} \log (nd) + 3 nd \\
&\leq 10 K_\alpha K_\psi \sqrt{nd} \log n d^2 + 2 \delta_3^2 K_\alpha K_\psi \sqrt{nd} \log n + 3 nd.
\end{align*}
By Theorem 4.1.1 in \cite{Tropp2015},
\begin{align*}
\mathbb{P}\left(\matnorm{ \frac{1}{n} \sum_{i=1}^n \epsilon_i \a_i \a_i^\top }{2} \geq t\right) \leq 2 d \exp\left\{ - \frac{t^2}{2 \matnorm{ \frac{1}{n^2} \sum_{i=1}^n \| \a_i\|_2^2 \a_i \a_i^\top }{2}}\right\},
\end{align*}
which yields that with a probability over $1 - 2 \exp(- \delta_3) - 2 \exp(- \delta_4)$,
\begin{align}
\label{eq:bstn-orlicz}
\matnorm{\frac{1}{n} \sum_{i=1}^n \epsilon_i \a_i \a_i^\top }{2}
&\leq \sqrt{2 \ln d \matnorm{ \frac{1}{n^2} \sum_{i=1}^n \| \a_i\|_2^2 \a_i \a_i^\top }{2}} + \sqrt{2 \delta_4 \matnorm{ \frac{1}{n^2}\sum_{i=1}^n \| \a_i\|_2^2 \a_i \a_i^\top }{2}} \notag \\
&\leq \left(\frac{\sqrt{\delta_4 d} + \sqrt{d \log d}}{\sqrt{n}}\right) \sqrt{\frac{(20 + 4 \delta_3^2) K_\alpha K_\psi d \sqrt{d} \log n}{\sqrt{n}} + 6} \notag \\
&\leq \sqrt{(20 + 4 \delta_3^2) K_\alpha K_\psi C + 6} \left(\frac{\sqrt{\delta_4 \log n} \, d + d \sqrt{\log n \log d}}{\sqrt{n}}\right),
\end{align}
where we use Assumption~\ref{ass:sim}. Combine inequality~\eqref{eq:err-cov}, \eqref{eq:bstn-subexp}, \eqref{eq:bstn-subgs} as well as inequality~\eqref{eq:bstn-orlicz}, we have with probability at least $1 - 2 \exp(- \delta_1^2 / 2)- 2 \exp(- \delta_2^2 / 2)- 2 \exp(- \delta_3) - 2 \exp(- \delta_4)$,
\begin{align}
\label{eq:bound-cov}
\matnorm{ \hSigma - \S }{2}
&\leq 17\theta \sqrt{\frac{d}{n}} + \sqrt{(20 + 4 \delta_3^2) K_\alpha K_\psi C \delta_4 + 6\delta_4} \frac{d \sqrt{\log n}}{\sqrt{n}} + \frac{8 \theta \delta_1 + \sigma \delta_2}{\sqrt{n}} \notag \\
&+ \sqrt{(20 + 4 \delta_3^2) K_\alpha K_\psi C + 6} \frac{d \sqrt{\log d \log n}}{\sqrt{n}}.
\end{align}
By Corollary 3.1 in \cite{Vu2013},
\begin{align}
\label{eq:bound-vec}
\min_{t \in \{-1,+1\}} \| t \hbe - \be^* \|_2 = \min_{t \in \{-1,+1\}} \sqrt{2} \matnorm{\sin \Theta (t \hbe, \be^*)}{\mathrm{F}} \leq 2 \matnorm{\hSigma - \S}{2}.
\end{align}
Combine bound~\eqref{eq:bound-vec} with \eqref{eq:bound-cov}, we have the desired result.

\end{proof}

%
\renewcommand\thefigure{\thesection.\arabic{figure}}
\section{Additional Experiments}
\label{sec:add-exp}

In this section, we present the experimental results on distributed PCR and distributed Gaussian SIM.

\subsection{Numerical results of distributed PCR}
\label{subsec:add-exp-pcr}

We provide numerical results of distributed PCR in this section. Recall the problem setting in  Section~\ref{subsec:setting-pcr}. We assume the real coefficient $\be^*$ lies in the top-$3$-dim eigenspace of $\S$, i.e.,
$
\be^* = \U_3 \ga^*,
$
where $\U_3 = [\u_1, \u_2, \u_3]^\top$ consists of the top-$3$ eigenvectors of $\S$. As in previous experiments, data dimension $d$ is set to $50$, and sample size on each machine is set to $500$, i.e., $\a \in \R^{50}$, $m = 500$. We vary the the number of machines. 
The response vector $\y \in \R^{mK}$ is generated by $\y = \A \be^* + \eps$, where noise term $\eps \sim \mathcal{N}(\0, \sigma^2 \I_{mK})$, and $\sigma^2$ is a constant, which is set to $0.2$ and $0.5$ in the following experiments. Here covariate data matrix $\A =[\a_1, \ldots, \a_{mK}]^\top \in \R^{mK \times d}$ is drawn \emph{i.i.d.} from $\mathcal{N}(\0, \S)$ with $\S = \U \La \U^\top, \,\La = \mathrm{diag}(2.5, 2, 1.5, 1, \ldots ,1)$. $\ga$ is sampled only once from $\mathcal{N}(\0, \I_3)$ and is fixed in our $100$ Monte-Carlo simulations. Moreover, the underlying true regression coefficient is $\be^* = \U_3 \frac{\ga}{\| \ga \|_2}$.


\begin{figure}[!t]
\centering
\subfigure[Distributed PCR ($\sigma^2=0.2$)]{
\includegraphics[width=0.48\textwidth]{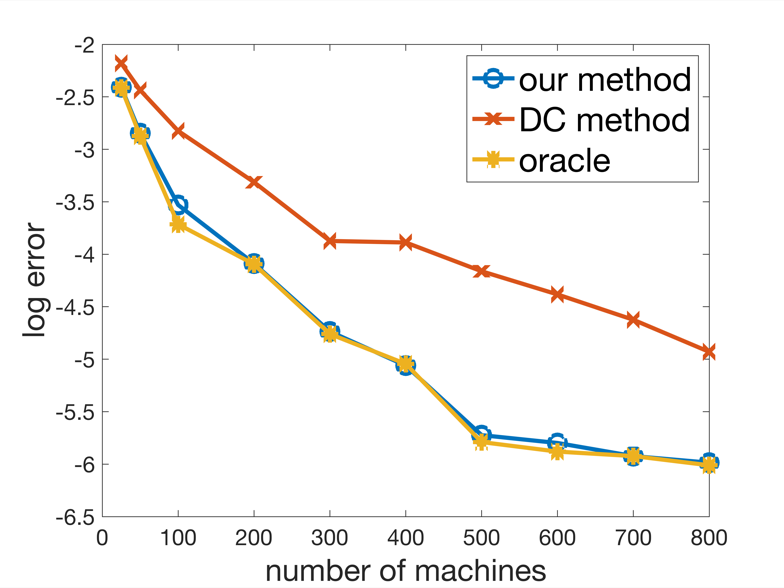}}
\subfigure[Distributed PCR ($\sigma^2=0.5$)]{
\includegraphics[width=0.48\textwidth]{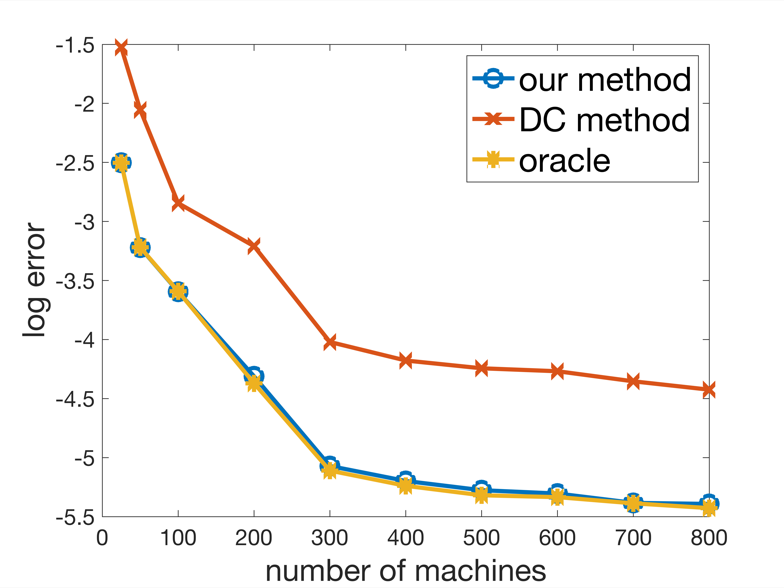}}
\caption{Comparison between algorithms in PCR when the number of machines varies. The $x$-axis is the the number of machines and the $y$-axis is the log-$l_2$ error. On the left, the noise term has variance $0.2$ and on the right figure, $0.5$.}
\label{pic:pcr-exp}
\end{figure}

We estimate $\U_3$ using $3$ different estimation methods and compare their performances. The measurement we use here is the $l_2$ distance between estimator $\hbe$ and real coefficient $\be^*$, i.e.,
$
\mathrm{error}(\hbe) = \| \hbe - \be^* \|_2.
$
The numbers of outer iterations and inner iterations in our algorithm are fixed as $40$ and $10$, respectively. The results are shown in Figure~\ref{pic:pcr-exp}. In accordance with previous experiments, our method almost keeps the same error rate as the oracle one.

\subsection{Numerical results of distributed SIM}
\label{subsec:add-exp-sim}

\begin{figure}[!t]
	\centering
	\subfigure[$f(u) = u^2$]{
		\includegraphics[width=0.3\textwidth]{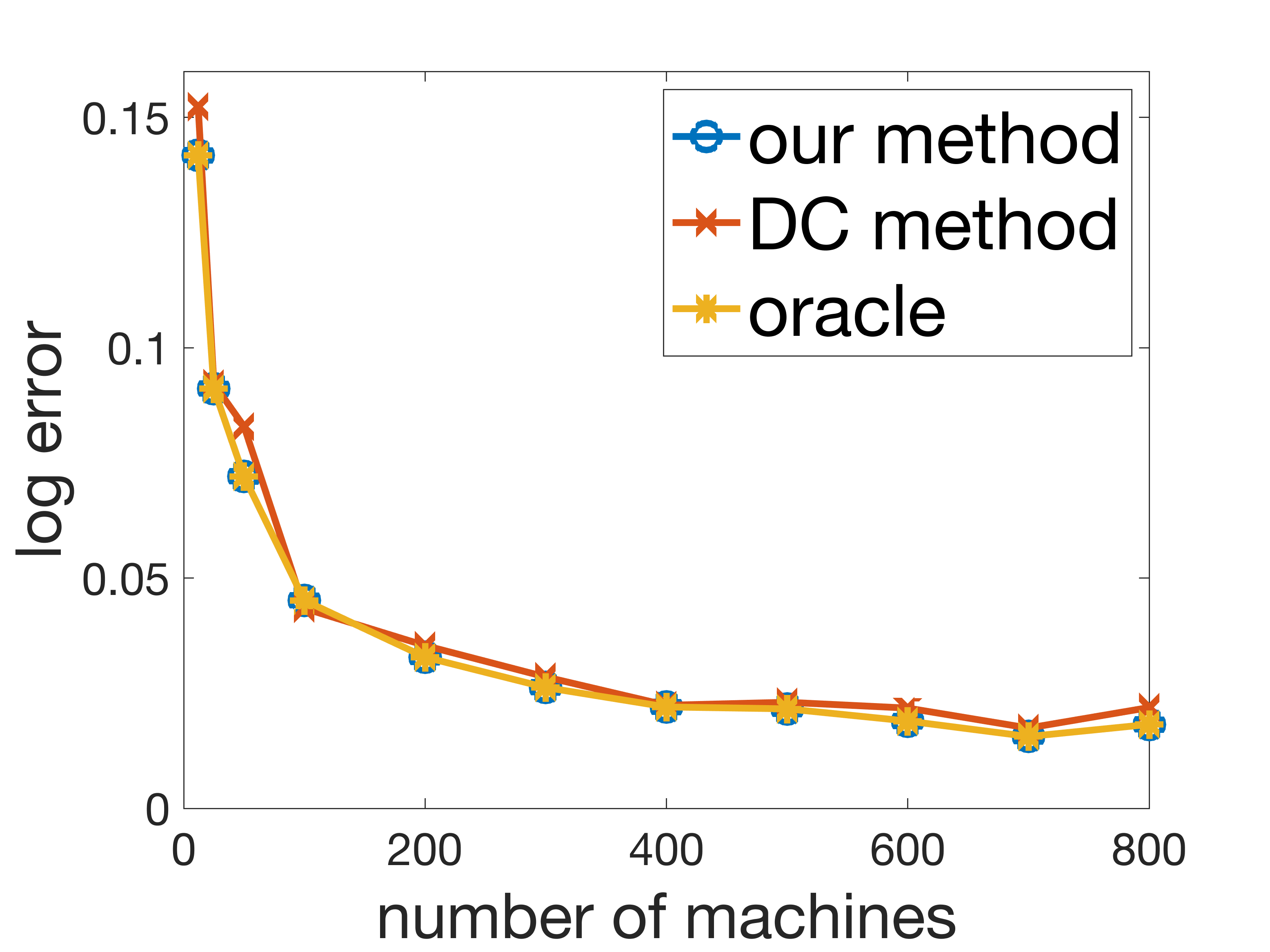}}
	\subfigure[$f(u) = |u|$]{
		\includegraphics[width=0.3\textwidth]{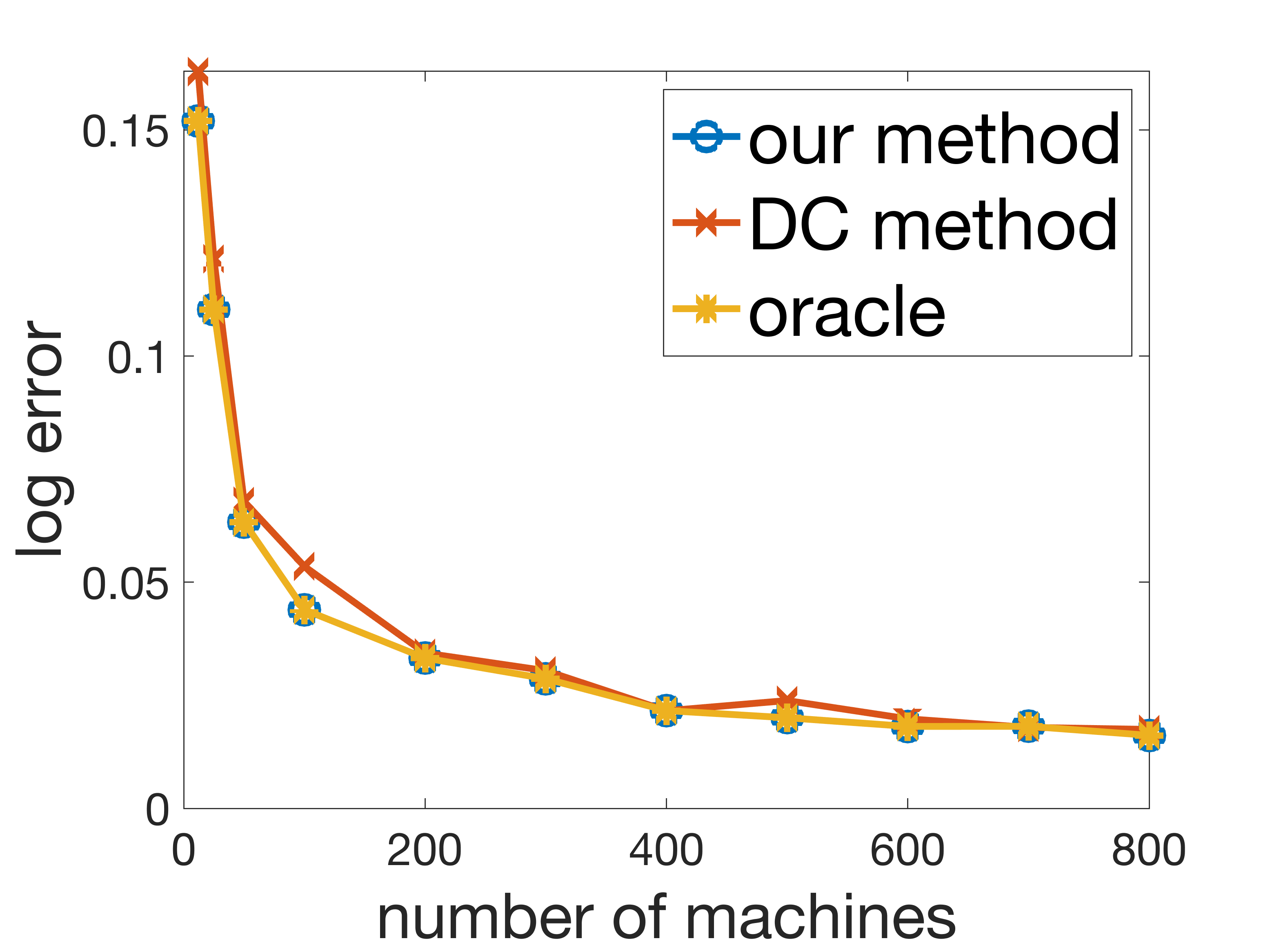}}
	\subfigure[ $f(u) = 4 u^2 + 3 \cos(u)$]{
		\includegraphics[width=0.3\textwidth]{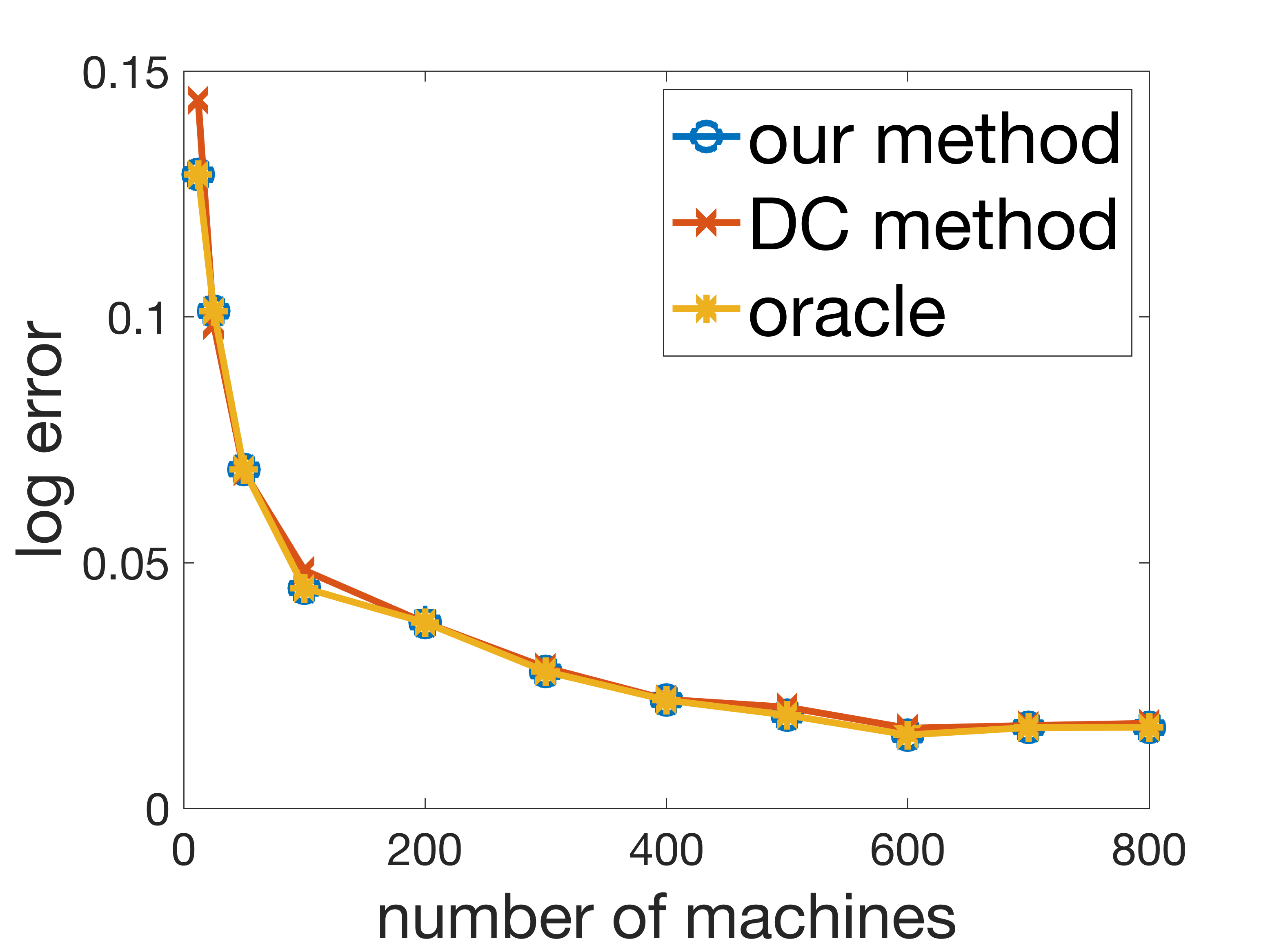}}
	\caption{Comparison between algorithms for SIM. The $x$-axis is the the number of machines and the $y$-axis is the log-$l_2$ error.}
	\label{pic:sim-exp}
\end{figure}

In our last part of the experiments, we conduct simulations on Gaussian single index model. Consider data dimension $d$ to be $50$, sample size on each machine to be $500$, i.e., $\a \in \R^{50}$, $m = 500$. Our covariate data matrix $\A \in \R^{mK \times d}$ is drawn independently, where each row $\a_i \sim \mathcal{N}(\0, \I_d)$ follows a standard normal distribution. For the data generating process of $\y = (y_1, \ldots, y_{mK})^\top \in \R^{mK}$, we have $y_i = f(\langle \be^*, \a_i \rangle) + \epsilon_i, \; \; \forall i \in \{1, \ldots, mK\}$,
where $f(\cdot)$ is our specific choice of link function, $\be^*$ is a normalized vector only drawn once during Monte-Carlo process from $\mathcal{N}(\0, \I_d)$, i.e., $\be^* = \be / \| \be\|_2, \, \be \sim \mathcal{N}(\0, \I_d)$ and $\{\epsilon_i\}$ are \emph{i.i.d.} normal $\mathcal{N}(\0, \sigma^2)$ with the constant variance $\sigma^2$ fixed to be $0.2$. During our estimation process, we estimate top eigenvector $\hbe$ of $\frac{1}{mK} \sum_{i=1}^{mK} y_i \cdot (\a_i \a_i^\top - \I_d)$. 

In the following experiment, we consider three different link functions: $f(u) = u^2$, $f(u) = |u|$ and $f(u) = 4 u^2 + 3 \cos(u)$. The $l_2$ distance $\| \hbe - \be^* \|_2$ is used here to measure the performance. In Figure~\ref{pic:sim-exp}, for all choices of link function, our estimators have the same errors as the oracle results. For this experiment, the DC method also works well, which is mainly because the problem of estimating the top eigenvector is relatively simple and $\a$ follows a symmetric normal distribution.

\end{appendices}

\end{document}